\documentclass[12pt]{amsart}
\usepackage{amssymb,latexsym,amsmath,enumerate,nicefrac}
\usepackage[dvipsnames,svgnames,x11names]{xcolor}
\usepackage{float}

\usepackage{caption}
\usepackage{comment}

\includecomment{showTikz}

\excludecomment{showPDF}

\excludecomment{showTOC}

\excludecomment{showData}

\usepackage{amsaddr}
\usepackage[total={6.5in,9in}, top=1in, left=1in, headsep = 0in]{geometry}

\usepackage{fancyhdr}
\usepackage{color}

\definecolor{egyptianblue}{rgb}{0.06, 0.2, 0.65}
\definecolor{cobalt}{rgb}{0.0, 0.28, 0.67}
\definecolor{yaleblue}{rgb}{0.06, 0.3, 0.57}
\definecolor{maroon}{rgb}{0.5, 0.0, 0.0}
\definecolor{indigo}{rgb}{0.0, 0.25, 0.42}
\definecolor{armygreen}{rgb}{0.29, 0.33, 0.13}
\definecolor{huntergreen}{rgb}{0.21, 0.37, 0.23}
\definecolor{royalblue}{rgb}{0.0, 0.14, 0.4}
\definecolor{zaffre}{rgb}{0.0, 0.08, 0.66}

\definecolor{myurlcolor}{rgb}{0.5, 0.0, 0.0}
\definecolor{mycitecolor}{rgb}{0.06, 0.2, 0.65}
\definecolor{myrefcolor}{rgb}{0.06, 0.2, 0.65}

\usepackage[bookmarks=false,pagebackref=true]{hyperref}
\hypersetup{colorlinks,
linkcolor=myrefcolor,
citecolor=mycitecolor,
urlcolor=myurlcolor}

\renewcommand*{\backref}[1]{}

\usepackage[all]{xy}
\usepackage{tikz}
\usetikzlibrary{calc}
\usepackage{graphics}
\usepackage{graphicx}
\usepackage{dsfont}
\usepackage{amsfonts}
\usepackage{mathtools}
\usetikzlibrary{topaths}

\usepackage{listings}
\usepackage{hyperref}
\usepackage{stmaryrd}
\usepackage{tikz-cd}
\usepackage{adjustbox}
\usepackage{comment}
\usepackage{verbatim}

\usetikzlibrary{arrows}
 \usetikzlibrary{matrix,positioning}

\DeclareMathOperator{\Hom}{Hom}
\pgfkeys{/tikz/.cd, arrowhead/.default=-1pt, arrowhead/.code={}}

\DeclareMathOperator{\Con}{Con}

\DeclareMathOperator{\Id}{Id}

\DeclareMathOperator{\Ob}{Ob}

\DeclareMathOperator{\Lie}{Lie}
\DeclareMathOperator{\LSpan}{Span}
\DeclareMathOperator{\proj}{proj}
\newcommand{\Stab}{\operatorname{Stab}}

\newcommand{\catname}[1]{{\mathsf{#1}}} 

\newcommand{\simple}{\text{decomposes external constraints}}

\newcommand{\elementary}{decomposes into constraints}

\newcommand{\Set}{\catname{Set}}

\newcommand{\Diff}{{\catname{Diff}}}
\newcommand{\SurjSub}{{\catname{SurjSub}}}
\newcommand{\Cat}{{\mathfrak M}}
\newcommand{\Kin}{{\catname{Kin}}}

\newcommand{\KSpan}{\mathcal{K}}

\newcommand{\CatJ}{{\catname{J}}}
\newcommand{\CatC}{{\catname{C}}}
\newcommand{\C}{\CatC}

\newcommand{\SA}{S_A}
\newcommand{\SL}{S_L}
\newcommand{\SR}{S_R}
\newcommand{\sL}{s_L}
\newcommand{\sR}{s_R}

\newcommand{\CA}{C_A}
\newcommand{\CL}{C_L}
\newcommand{\CR}{C_R}
\newcommand{\cL}{c_L}
\newcommand{\cR}{c_R}

\newcommand{\QA}{Q_A}
\newcommand{\QL}{Q_L}
\newcommand{\QR}{Q_R}
\newcommand{\qL}{q_L}
\newcommand{\qR}{q_R}

\newcommand{\PA}{P_A}

\newcommand{\pL}{p_L}
\newcommand{\pR}{p_R}

\newcommand{\id}{\text{id}}

\newcommand{\R}{\mathds{R}}

\newcommand{\Span}{\mathsf{Span}}

\newcommand{\D}{\mathscr{D}}

\newcommand{\A}{\mathcal{A}}

\newcommand{\SE}{\mathds{SE}}
\newcommand{\SO}{\mathds{SO}}

\pgfkeys{/tikz/.cd, arrowhead/.default=-1pt, arrowhead/.code={}}

\allowdisplaybreaks

\usepackage{tikz}
\usetikzlibrary{calc}
\usepackage{extarrows}

\usetikzlibrary{decorations.pathmorphing}
\usetikzlibrary{arrows}
\usetikzlibrary{calc}
\usepackage{graphics}
\usepackage{dsfont}
\usepackage{amsmath}
\usepackage{amsfonts,amssymb}
\usepackage{mathrsfs}
\usepackage{mathtools}

\pgfkeys{/tikz/.cd, arrowhead/.default=-1pt, arrowhead/.code={}}

\newcommand{\x}{[x]}

\renewcommand{\SS}{{\mathds S}}
\newcommand{\ZZ}{\mathds{Z}}

\newcommand{\F}{\mathcal F}

\renewcommand{\Cat}{\mathsf{Cat}}

\newcommand{\Ext}{\text{Ext}}
\newcommand{\isoarrow}{\xlongrightarrow{\ \sim\ }}

\newtheorem{theorem}{Theorem}[section]
\newtheorem{lemma}{Lemma}[section]

\newtheorem{proposition}{Proposition}[section]
\theoremstyle{definition}
\newtheorem{definition}{Definition}[section]
\newtheorem{example}{Example}
\theoremstyle{remark}

\newtheorem*{remark}{Remark}

\makeatletter
\def\namedlabel#1#2{\begingroup
    #2%
    \def\@currentlabel{#2}%
    \phantomsection\label{#1}\endgroup
}
\makeatother

\begin{document}

\title{A compositional framework for open classical kinematic systems}

\author{Andrea Abeje-Stine \and David Weisbart}%
\address{%
\begin{tabular}[h]{cc}
Department of Mathematics %
 \\University of California, Riverside %
 \end{tabular}
 }
\email{astin005@ucr.edu}\email{weisbart@math.ucr.edu}%

\begin{abstract}

Our aim is to introduce a framework sufficiently general to describe the kinematics of a wide variety of open systems in classical mechanics while uniquely characterizing systems with specified simplest components. The data describing a physical system are local, so the construction of a global configuration space requires compatibility among local interactions. We model open systems as morphisms in a category $\Kin(\F)$, where composition encodes how subsystems attach to one another and embed into larger systems. The framework supports a precise treatment of geometric constraints and clarifies when locally specified subsystems are compatible. It also yields a structural approach to the study of linkages: we prove the nonconstructibility of sliding hinges and universal joints from two rigid bodies attached by a surface constraint compatible with rigid-motion symmetries.

\end{abstract}

\maketitle

\tableofcontents

\thispagestyle{empty}


\section{Introduction}\label{Sec:Intro}


\subsection{Summary of Results}

The study of mechanical linkages is ancient; mechanisms such as the trammel of Archimedes date back at least to Proclus and possibly to Archimedes himself \cite{Wetzel}. In 1875, Franz Reuleaux, the ``father of kinematics" \cite{Moon}, articulated \cite{Reuleaux} the compositional viewpoint central to the present work:
\begin{quote}
The elementary parts of a machine are not single, but occur always in pairs, so that the machine, from a kinematic point of view, must be divided rather into pairs of elements than into single elements.
\end{quote}
The principal novelty of this work is a general category-theoretic framework for open kinematical systems in classical mechanics. The framework specializes to the machines described by Reuleaux and gives structural criteria for determining which machines cannot be constructed from pairs of rigid bodies constrained by contact along a surface, the classical \emph{lower kinematic pairs}. 

Theorems~\ref{MainA} and \ref{MainB} show that, under natural locality and compatibility assumptions, a configuration space for a classical mechanical system exists and appears as a universal object determined uniquely by interactions between elementary parts. When compatibility fails, locally specified interaction data need not assemble into a global configuration space; explicit examples exhibit this phenomenon (Examples~\ref{nonexample:Flimitb} and \ref{example:lock3Bar}). The framework that this work constructs appears to provide a minimal structure needed to formulate and prove classical non--existence results for lower kinematic pairs (Theorems~\ref{thm:NoTwoActorUniversalJoint}, \ref{thm:NoTwoActorPlanarSlidingHinge}, and \ref{thm:NoTwoActorFiberProductSpatialSlidingHinge}). It places Reuleaux's pair-based viewpoint into a broader mathematical setting and serves as a foundation for subsequent work on dynamics.

\subsection{Motivation and Background}

In part because of their numerous engineering applications, linkages have been studied from many mathematical perspectives. Gr\"{u}bler formalized a mobility formula for planar linkages in 1917 \cite{Gruebler}, and Kutzbach extended the formula to spatial linkages in 1929 \cite{Kutzbach}. Freudenstein developed an efficient parameterization of four-link mechanisms \cite{Freudenstein}, building on earlier work of Burmester \cite{Burmester}. Denavit and Hartenberg introduced coordinate parameterizations for composing constraint equations, including closed-loop systems \cite{DH}. Mruthyunjaya developed a computational approach to the structural synthesis of linkages corresponding to arbitrary graphs \cite{Mruthyunjaya}.  The study of linkages intersects with algebraic geometry \cite{McCarthy} and Lie theory \cite{MurrayLiSastry}, with connections reaching back to Kempe's Universality Theorem \cite{Kempe} and to screw theory \cite{Ball}. Kapovich and Millson studied the topology of moduli spaces of polygons and applied their study to linkages \cite{KapMillson1, KapMillson2}. The current work reframes the study of linkages in the broader context of open classical mechanical systems from a category theoretic perspective.

An open system is one that can interact with external systems, and its description must account for how it embeds into larger systems. A compositional approach to the study of open systems asserts that the rules governing subsystem composition, together with the properties of the subsystems, determine the properties of the whole. Although modeling in classical mechanics usually begins with a configuration space and then imposes constraints on it, the data describing a system is local: components interact only through constraints involving small collections of parts. The existence of a global configuration space is therefore not automatic but a compatibility question: Does the collection of local interaction data determine a consistent global configuration space?
 
A central program in applied category theory involves modeling open systems as morphisms in a category, with composition encoding subsystem composition. This perspective originates in extended topological quantum field theory, where cobordism composition assembles manifolds \cite{Baez-Dolan,BFV,Freed,Hau}. Subsequent work applies this approach to electrical circuits \cite{Rebro}, dynamical systems \cite{SSV}, and cyber-physical systems \cite{BFV2}. In their study of classical mechanical systems, Baez et al.\ introduced a span-based framework for certain Lagrangian and Hamiltonian systems \cite{BWY}.

For now, view a span in a category $\C$ to be a pair of morphisms in $\C$ that have the same source. The essential idea of Baez et al.\ was to identify an open classical mechanical system with a span $S$ in a category $\C$ \cite{BWY}. The source of the morphisms that define $S$ encodes the trajectories of the system. Span composition determines how systems compose to form larger systems.

Technical obstructions arise when constructing span categories for classical mechanical systems. These obstructions complicate direct application of existing open system frameworks \cite{Courser}. The principal difficulty is that the relevant categories do not admit pullbacks. Weisbart and Yassine addressed this issue using generalized span categories by introducing $\F$\emph{-pullbacks} to determine span composition uniquely \cite{WY}. Street introduced the notion of a \emph{fake pullback} to treat span categories in categories without pullbacks \cite{Street}. In order to align with the prior work of Baez et al. \cite{BWY}, the present work adopts the generalized span framework, which Section~\ref{Sec:2} reviews. Street's approach offers an alternative treatment of span composition and its applicability here merits further investigation.

The prior framework \cite{BWY} directly treats systems constructed from finitely many linearly ordered subsystems. A spring--mass system consisting of finitely many masses constrained to move along a line and connected by springs provides a basic example. Such a system looks like this:

\medskip

\begin{center}
\begin{tikzpicture}[scale = .75]

%
\coordinate(r1) at (-1.5,0);
\coordinate(r2) at (0,0);
\coordinate(r3) at (1.5,0);
\coordinate(r4) at (3,0);
\coordinate(r5) at (4.5,0);

\coordinate(dotsa) at (2,0);
\coordinate(dotsb) at (2.5,0);

\coordinate(dots1) at (2.15,0);
\coordinate(dots2) at (2.25,0);
\coordinate(dots3) at (2.35,0);

\draw[decorate,decoration={coil,segment length=4pt, amplitude = 4pt},rotate=0] (r1)  -- (r2);
\draw[decorate,decoration={coil,segment length=4pt, amplitude = 4pt},rotate=0] (r2)  -- (r3);
\draw[decorate,decoration={coil,segment length=4pt, amplitude = 4pt},rotate=0] (r3)  -- (dotsa);
\draw[decorate,decoration={coil,segment length=4pt, amplitude = 4pt},rotate=0] (dotsb)  -- (r4);
\draw[decorate,decoration={coil,segment length=4pt, amplitude = 4pt},rotate=0] (r4)  -- (r5);

\draw[fill= black] (r1) circle (2.5pt);
\draw[fill= black] (r2) circle (2.5pt);
\draw[fill= black] (r3) circle (2.5pt);
\draw[fill= black] (r4) circle (2.5pt);
\draw[fill= black] (r5) circle (2.5pt);

\draw[fill= black] (dots1) circle (.5pt);
\draw[fill= black] (dots2) circle (.5pt);
\draw[fill= black] (dots3) circle (.5pt);
\end{tikzpicture}
\end{center}

\medskip

\noindent For a system with three masses and two springs, identifying the right mass of the left subsystem with the left mass of the right subsystem decomposes the larger system into smaller systems, like this:

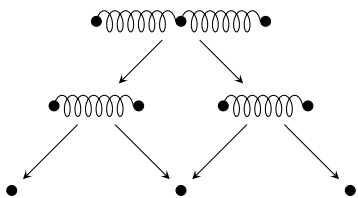
\medskip

\begin{center}
\begin{tikzpicture}[scale = .75]

\coordinate(r1) at (-1.5,3);
\coordinate(r2) at (0,3);
\coordinate(r3) at (1.5,3);

\coordinate(r4) at (-1.5,1.5);
\coordinate(ra) at (-2.25,1.5);
\coordinate(rb) at (-.75,1.5);

\coordinate(r5) at (1.5,1.5);
\coordinate(rc) at (.75,1.5);
\coordinate(rd) at (2.25,1.5);

\coordinate(br) at (-3,0);
\coordinate(bm) at (0,0);
\coordinate(bl) at (3,0);

\coordinate(l1) at (-.75,3);
\coordinate(l2) at (.75,3);
\coordinate(l3) at (-1.5,1.5);
\coordinate(l4) at (1.5,1.5);
\draw[decorate,decoration={coil,segment length=4pt, amplitude = 4pt},rotate=0] (r1)  -- (r2);
\draw[decorate,decoration={coil,segment length=4pt, amplitude = 4pt},rotate=0] (r2)  -- (r3);
\draw[decorate,decoration={coil,segment length=4pt, amplitude = 4pt},rotate=0] (ra)  -- (rb);
\draw[decorate,decoration={coil,segment length=4pt, amplitude = 4pt},rotate=0] (rc)  -- (rd);

\draw[fill= black] (r1) circle (2.5pt);
\draw[fill= black] (r2) circle (2.5pt);
\draw[fill= black] (r3) circle (2.5pt);

\draw[fill= black] (ra) circle (2.5pt);
\draw[fill= black] (rb) circle (2.5pt);
\draw[fill= black] (rc) circle (2.5pt);
\draw[fill= black] (rd) circle (2.5pt);

\draw[fill= black] (br) circle (2.5pt);
\draw[fill= black] (bm) circle (2.5pt);
\draw[fill= black] (bl) circle (2.5pt);

\draw[shorten >=12pt,shorten <=10pt, -stealth, arrowhead=5pt, line width=.4pt] (r2) -- (r4);
\draw[shorten >=12pt,shorten <=10pt, -stealth, arrowhead=5pt, line width=.4pt] (r2) -- (r5);
\draw[shorten >=6pt,shorten <=10pt, -stealth, arrowhead=5pt, line width=.4pt] (r4) -- (br);
\draw[shorten >=6pt,shorten <=10pt, -stealth, arrowhead=5pt, line width=.4pt] (r4) -- (bm);
\draw[shorten >=6pt,shorten <=10pt, -stealth, arrowhead=5pt, line width=.4pt] (r5) -- (bm);
\draw[shorten >=6pt,shorten <=10pt, -stealth, arrowhead=5pt, line width=.4pt] (r5) -- (bl);

\end{tikzpicture}
\end{center}

\medskip

From the Hamiltonian perspective, the state space of each mass is $T^\ast\mathds R$, the cotangent bundle of $\mathds R$. Each arrow in the diagram represents a canonical projection between these subsystem state spaces:

\medskip

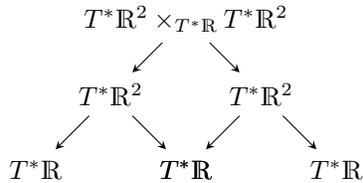
\begin{center} 
\begin{tikzpicture}

\node[font = \footnotesize] (SL) at (-1,0) {$T^\ast\mathds R$};
\node[font = \footnotesize] (SR) at (1,0) {$T^\ast\mathds R$};
\node[font = \footnotesize] (SA) at (0,1) {$T^\ast\mathds R^2$};
\node[font = \footnotesize] (QL) at (1,0) {$T^\ast\mathds R$};
\node[font = \footnotesize] (QR) at (3,0) {$T^\ast\mathds R$};
\node[font = \footnotesize] (QA) at (2,1) {$T^\ast\mathds R^2$};
\node[font = \footnotesize] (PA) at (1,2) {$T^\ast\mathds R^2\times_{T^\ast\mathds R}T^\ast\mathds R^2$};

\draw[-stealth, shorten >= 0pt, shorten <= 0pt] (SA) -- (SL);
\draw[-stealth, shorten >= 0pt, shorten <= 0pt] (SA) -- (SR);
\draw[-stealth, shorten >= 0pt, shorten <= 0pt] (QA) -- (QL);
\draw[-stealth, shorten >= 0pt, shorten <= 0pt] (QA) -- (QR);
\draw[-stealth, shorten >= 0pt, shorten <= 0pt] (PA) -- (SA);
\draw[-stealth, shorten >= 0pt, shorten <= 0pt] (PA) -- (QA);

\end{tikzpicture}
\end{center}

\medskip

\noindent The canonical projections are surjective Poisson maps between symplectic manifolds. The two spring--mass subsystems in the middle of the diagram have state space $T^\ast\mathds R^2$. The state space of the total system, describing three masses interacting in series, is the fibered product of two copies of $T^\ast\mathds R^2$ over $T^\ast\mathds R$. From the Lagrangian perspective, tangent bundles serve as the state spaces, and the maps are surjective Riemannian submersions.

The fibered product models the identification of the right mass of the left spring--mass system with the left mass of the right spring--mass system and yields a six-dimensional manifold rather than the eight-dimensional Cartesian product of the subsystem state spaces. The use of fibered products in this context goes back at least to Dazord \cite{Daz}, who constructed configuration and state spaces for certain geometrically constrained systems in this way. Marle subsequently described a broader class of constrained systems as submanifolds of unconstrained configuration or state spaces; in the holonomic case considered here, his results agree with those of Dazord \cite{Marle}.


Categories with pullbacks and terminal objects have finite limits. In such categories, the compositional structure of spans may be expanded beyond sequential composition to allow finite graph-shaped patterns of interconnection. Duality transfers the equivalence between hypergraph categories and cospan algebras into the corresponding span statement \cite[Theorem~4.13]{https://arxiv.org/pdf/1806.08304}. Since $\Diff$ and $\SurjSub$ do not have pullbacks, this result does not immediately apply. Further, existence of $\F$-pullbacks does not guarantee existence of $\F$-limits, as shown by Example~\ref{nonexample:Flimitb}, which would complicate composition in a hypergraph category of spans in $\SurjSub$. Thus, the earlier framework \cite{BWY,WY} does not immediately appear to accommodate systems where subsystems are not linearly ordered. The following system, in which three springs mediate interactions among three masses moving in $\mathds R^2$, exhibits feedback not captured by that framework:

\smallskip

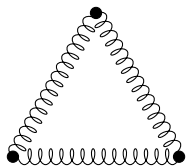
\begin{center}
\begin{tikzpicture}[scale = .85]

%
\coordinate(r1) at (90:1.5);
\coordinate(r2) at (210:1.5);
\coordinate(r3) at (330:1.5);

\draw[decorate,decoration={coil,segment length=4pt, amplitude = 3pt},rotate=0] (r2)  -- (r1);
\draw[decorate,decoration={coil,segment length=4pt, amplitude = 3pt},rotate=0] (r3)  -- (r2);
\draw[decorate,decoration={coil,segment length=4pt, amplitude = 3pt},rotate=0] (r1)  -- (r3);

\draw[fill= black] (r1) circle (2.5pt);
\draw[fill= black] (r2) circle (2.5pt);
\draw[fill= black] (r3) circle (2.5pt);
\end{tikzpicture}
\end{center}

\smallskip
  
\noindent Although mild extensions of the earlier framework accommodate certain systems with feedback, the framework does not appear to directly describe more complicated systems by assembling local interactions.

The present work introduces actor-constraint mediated systems as a compositional framework for open kinematical systems in classical mechanics. The framework treats actors not merely as point particles, but as carriers of internal degrees of freedom that encode geometric constraints between actors. This added structure supports a local compositional study of lower kinematic pairs, their assembly into linkages, and their role in standard engineering classifications. Rather than beginning with a global configuration space, the framework asks when linkage mobility can be reconstructed from local actor data and pairwise assembly. To our knowledge, no existing work formulates this problem through generalized span composition in a setting that admits feedback.

The actor-index categories and ACM-diagrams defined here resemble the constraint graphs studied in constraint satisfaction problems, beginning with the work of Friedman and Leondes \cite{FriedmanLeondes1, FriedmanLeondes2}. A category of ACM-systems allows spans in a generalized span category to be arranged according to an undirected graph, analogous to the role of hypergraph categories for ordinary cospans. The construction does not require $\F$-limits, which a direct application of hypergraph categories to generalized spans might. Recent work models constraint hypergraphs of control systems categorically \cite{MorrisMockoWagner}. Unlike constraint hypergraphs of sets, ACM-diagrams encode the structure of the target category and satisfy more restrictive axioms. 

ACM-diagrams also resemble the construction of $D$-categories \cite{Ames}. The relationship between a constraint skeleton and its associated ACM-diagram resembles the functor $IC$ from $\mathsf{SimpGph}$ to $\mathsf{Quiv}$ \cite{Moeller}. A key distinction is that $D$-categories and the functor $IC$ provide structure for diagrams of actors and constraints, whereas the present framework also requires pairwise interactions between actors. As in the hypergraph category setting \cite{Fong-Spivak}, the assembly of interactions in an ACM-diagram resembles operadic composition in a multicategory \cite{Leinster}.  The notion of openness also resembles the study of open networks using props \cite{Rebro}. 

\subsection{Organization of the Paper}

Section~\ref{Sec:2} reviews earlier work on open systems \cite{BWY} and identifies its limitations for systems with feedback. It also presents the assumptions about kinematic systems in classical mechanics that motivate the introduction of a rigid inclusion category of \emph{ACM-systems}. Given categories $\CatC$ and $\CatC^\prime$ and a functor $\F\colon \CatC\to\CatC^\prime$, Section~\ref{Sec3} defines ACM-diagrams and formulates technical conditions on $\F$ related to the span-tightness condition \cite{WY}. It proves the results required for Section~\ref{Sec4}. Section~\ref{Sec4} introduces ACM-systems and their composition in a rigid inclusion category and proves that the conditions from Section~\ref{Sec3} yield an ACM-system unique up to isomorphism.

Take $\Diff$ to be the category of smooth manifolds with smooth maps, and $\SurjSub$ to be the category of smooth manifolds with surjective submersions. Section~\ref{Sec:CMK-systems-as-rigid-inclusions} shows that the forgetful functor $\F\colon \SurjSub\to\Diff$ satisfies the conditions introduced in Sections~\ref{Sec3} and \ref{Sec4}, yielding a rigid inclusion category over $\F$. It presents examples of open kinematic systems in the rigid inclusion framework and shows that certain linkages, although constructible from lower kinematic pairs, require more than two actors. In particular, it proves that the universal joint and the sliding hinge are not lower kinematic pairs because they require at least three distinct actors. The section also initiates a classification of lower kinematic pairs. This classification arises naturally within the abstract framework and illustrates its applicability to classical mechanics. The emergence of a rigid inclusion category in this setting indicates possible applications beyond classical mechanics.


\section{Composition of systems}\label{Sec:2}


An \emph{actor} in a classical mechanical system, henceforth a CM-system, is a point particle with positive mass. Simple models involve a finite collection of actors and their pairwise interactions. The state space of each actor is an object in a category $\CatC$ that Section~\ref{Sec:CMK-systems-as-rigid-inclusions} specifies. The development proceeds abstractly to construct a compositional framework for actor--constraint mediated systems (ACM-systems). A kinematic system in classical mechanics, henceforth a CMK system, is an ACM-system in which $\F\colon \SurjSub \to \Diff$ is the forgetful functor from the category of surjective submersions between smooth manifolds to the category of smooth functions between smooth manifolds (Definition~\ref{def:CMK-systeM}).

\subsection{Constraint skeletons of systems}

The diagrammatic structure of classical mechanics motivates the following notion. Informally, the \emph{constraint skeleton} of an ACM-system is a finite undirected graph (see Definition~\ref{Def:Skeleton} for the formal definition). As an example, consider actors $A$, $B$, $C$, and $D$ moving in three-dimensional Euclidean space. A massless rigid bar constrains $B$ to move on a sphere centered at the position of $A$. Actors $C$ and $D$ have no geometric constraints. Springs connect the pairs $(A,D)$, $(B,C)$, and $(C,A)$, each applying a force along the line joining the two actors. The following graph represents the interactions of this ACM-system:

\smallskip

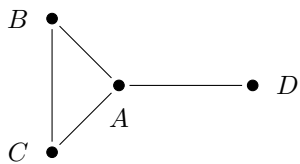
\begin{center}
\begin{tikzpicture}

%
\coordinate(A) at (0in,0);
\coordinate(B) at (-.35in,.35in);
\coordinate(C) at (-.35in,-.35in);
\coordinate(D) at (.7in,0);

\draw[shorten >= 4pt, shorten <= 4pt] (A) -- (B);
\draw[shorten >= 4pt, shorten <= 4pt] (A) -- (D);
\draw[shorten >= 4pt, shorten <= 4pt] (B) -- (C);
\draw[shorten >= 4pt, shorten <= 4pt] (C) -- (A);

\draw[fill = black] (A) circle (2pt) node[anchor = north, font = \footnotesize, outer sep = 5pt] {$A$};
\draw[fill = black] (B) circle (2pt) node[anchor = east, font = \footnotesize, outer sep = 5pt] {$B$};
\draw[fill = black] (C) circle (2pt) node[anchor = east, font = \footnotesize, outer sep = 5pt] {$C$};
\draw[fill = black] (D) circle (2pt) node[anchor = west, font = \footnotesize, outer sep = 5pt] {$D$};
\end{tikzpicture}
\end{center}

\smallskip
 
\noindent The vertices correspond to actors; the edges represent \emph{elementary interactions}---those between exactly two actors. These may be geometric constraints or dynamical interactions (i.e. forces). In the kinematic description of a CM-system, henceforth a CMK system, only the geometric constraints appear. Thus, the above graph reduces to a constraint skeleton like this:
 
\smallskip

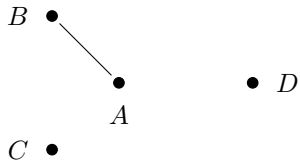
\begin{center}
\begin{tikzpicture}

%
\coordinate(A) at (0in,0);
\coordinate(B) at (-.35in,.35in);
\coordinate(C) at (-.35in,-.35in);
\coordinate(D) at (.7in,0);

\draw[shorten >= 4pt, shorten <= 4pt] (A) -- (B);
\draw[fill = black] (A) circle (2pt) node[anchor = north, font = \footnotesize, outer sep = 5pt] {$A$};
\draw[fill = black] (B) circle (2pt) node[anchor = east, font = \footnotesize, outer sep = 5pt] {$B$};
\draw[fill = black] (C) circle (2pt) node[anchor = east, font = \footnotesize, outer sep = 5pt] {$C$};
\draw[fill = black] (D) circle (2pt) node[anchor = west, font = \footnotesize, outer sep = 5pt] {$D$};
\end{tikzpicture}
\end{center}

\smallskip

\noindent The goal is to define an abstract open CMK system in a compositional framework that uniquely determines the state space from a collection of \emph{elementary interactions}---CMK systems with two actors. The general framework governs open ACM-systems. Compositionality requires that the model of the total system be constructed in finitely many steps by incorporating one elementary interaction at a time. Each step depends only on the data of the interacting pair of actors. The resulting model must be independent of the order in which the elementary interactions are included.

\subsection{Compositional description of linearly-ordered systems}

For any categories $\CatJ$ and $\CatC$, a \emph{diagram} of \emph{shape} $\CatJ$ in $\CatC$ is a functor $\D\colon \CatJ\to\CatC$. The diagram $\D$ is \emph{finite} if $\CatJ$ is finite.  The category $\CatJ$ is an \emph{index category} for diagrams in $\CatC$ with shape $\CatJ$.  This work restricts index categories to finite posets.  Consequently, $\CatJ$ is thin.  Any diagram $\D\colon \CatJ\to\CatC$ is therefore a commutative diagram in $\CatC$.

A \emph{span} $S$ and a \emph{cospan} $C$ in $\CatC$ are diagrams in $\CatC$ with shape determined by three distinct objects $\{A, L, R\}$. A span assigns morphisms from $A$ to $L$ and from $A$ to $R$, where a cospan assigns morphisms from $L$ to $A$ and from $R$ to $A$, like this:

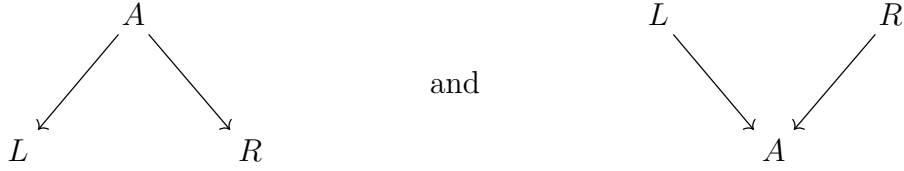
\begin{center}
\begin{tikzcd}[ampersand replacement=\&,cramped,row sep=small]
	\& A \&\&\&\&\& L \&\& R \\
	\&\&\&\& {\text{and}} \\
	L \&\& R \&\&\&\&\& A
	\arrow[from=1-2, to=3-1]
	\arrow[from=1-2, to=3-3]
	\arrow[from=1-7, to=3-8]
	\arrow[from=1-9, to=3-8]
\end{tikzcd}
    
\end{center}

\noindent Denote by $\sL\colon \SA \to \SL$ and $\sR\colon \SA \to \SR$ the images of the left and right arrows, respectively, in $\CatC$. The objects $\SL$ and $\SR$ are, respectively, the \emph{left foot} and \emph{right foot} of $S$, and $\SA$ is the \emph{apex} of $S$. When more careful specification is needed, write \[S = \langle\sL,\sR\rangle.\] Similarly, define $\cL\colon \CL \to \CA$ and $\cR\colon \CR \to \CA$ to be the images of the arrows in $\CatC$, and write \[C = \; \rangle\cL,\cR\langle.\]

For any spans $S$ and $Q$ with the same feet, a \emph{span morphism} from $S$ to $Q$ is a morphism $\Phi$ in $\CatC$ from $\SA$ to $\QA$ with \[\sL = \qL \circ \Phi \quad \text{and} \quad \sR = \qR \circ \Phi,\] so that this diagram commutes:

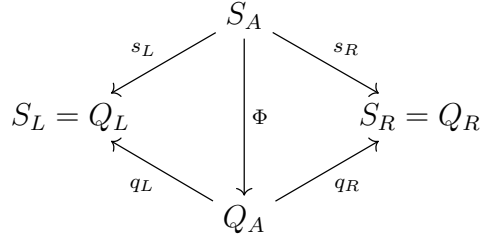
\begin{center} 
\begin{tikzcd}[ampersand replacement=\&,cramped]
	\& {S_A} \\
	{S_L = Q_L} \&\& {S_R = Q_R} \\
	\& {Q_A}
	\arrow["{s_L}"', from=1-2, to=2-1]
	\arrow["{s_R}", from=1-2, to=2-3]
	\arrow["\Phi", from=1-2, to=3-2]
	\arrow["{q_L}", from=3-2, to=2-1]
	\arrow["{q_R}"', from=3-2, to=2-3]
\end{tikzcd}
\end{center} 

\noindent A span morphism $\Phi$ is a \emph{span isomorphism} if $\Phi$ is an isomorphism. Denote by $[S]$ the isomorphism class of spans isomorphic to $S$.

A span $Q$ is \emph{paired with} a cospan $C$ if \[\QL = \CL, \quad \QR = \CR, \quad \text{and} \quad \cL\circ\qL = \cR\circ\qR.\] It is a \emph{pullback} of $C$ if, in addition, the following universal property holds:  For any span $S$ paired with $C$, there exists a unique span morphism $\Phi$ in $\CatC$ from $S$ to $Q$. Commutativity of these diagrams expresses the pairing and universal properties of a pullback:

\smallskip

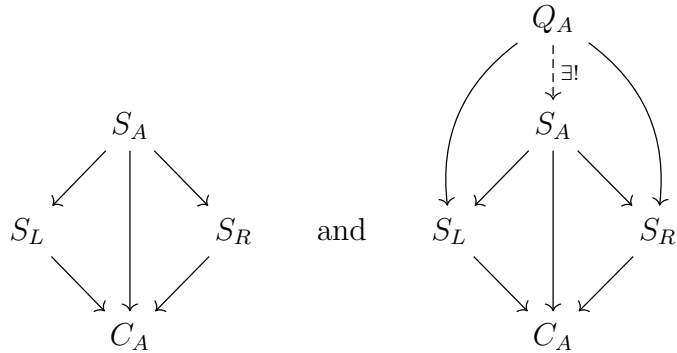
\begin{center}
\begin{tikzcd}[ampersand replacement=\&,column sep=small]
	\&\&\&\&\& {Q_A} \\
	\& {S_A} \&\&\&\& {S_A} \\
	{S_L} \&\& {S_R} \& {\text{and}} \& {S_L} \&\& {S_R} \\
	\& {C_A} \&\&\&\& {C_A}
	\arrow["{\exists !}", dashed, from=1-6, to=2-6]
	\arrow[bend left=-30, from=1-6, to=3-5]
	\arrow[bend right=-30, from=1-6, to=3-7]
	\arrow[from=2-2, to=3-1]
	\arrow[from=2-2, to=3-3]
	\arrow[from=2-2, to=4-2]
	\arrow[from=2-6, to=3-5]
	\arrow[from=2-6, to=3-7]
	\arrow[from=2-6, to=4-6]
	\arrow[from=3-1, to=4-2]
	\arrow[from=3-3, to=4-2]
	\arrow[from=3-5, to=4-6]
	\arrow[from=3-7, to=4-6]
\end{tikzcd}
\end{center}

\smallskip

For any spans $S$ and $Q$, if $\SR$ is equal to $\QL$ and $\PA$ is a pullback of the cospan $\rangle \sR, \qL\langle$, define \[S\circ_P Q = \langle \sL\circ\pL, \qR\circ\pR\rangle\] to be the span formed from morphisms in the left-hand diagram and equal to the span shown on the right:

\smallskip

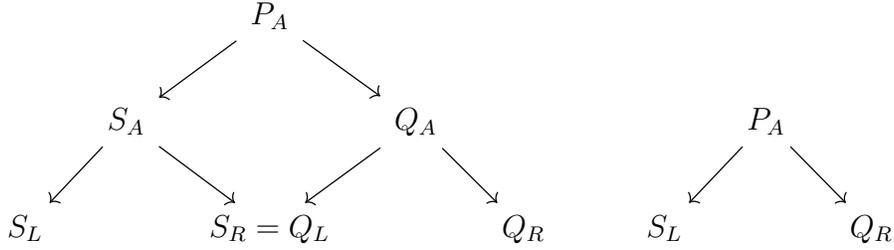
\begin{center} 

\begin{tikzcd}[ampersand replacement=\&,column sep=small]
	\&\& {P_A} \\
	\& {S_A} \&\& {Q_A} \&\&\&\& {P_A} \\
	{S_L} \&\& {S_R = Q_L} \&\& {Q_R} \&\& {S_L} \&\& {Q_R}
	\arrow[from=1-3, to=2-2]
	\arrow[from=1-3, to=2-4]
	\arrow[from=2-2, to=3-1]
	\arrow[from=2-2, to=3-3]
	\arrow[from=2-4, to=3-3]
	\arrow[from=2-4, to=3-5]
	\arrow[from=2-8, to=3-7]
	\arrow[from=2-8, to=3-9]
\end{tikzcd}

\end{center} 

\smallskip

The uniqueness of pullbacks of cospans up to span isomorphism defines a procedure for composing spans.  Pullbacks of cospans, however, do not necessarily exist in a category---and fail to exist in the categories most relevant to the study of CM-systems.  This motivates a generalization of the pullback concept. Take $\CatC^\prime$ to be any category and $\F\colon \CatC\to\CatC^\prime$ to be a functor.  For any span $S$ and cospan $C$, the image $\F(S)$ is a span in $\CatC^\prime$, and $\F(C)$ is a cospan in $\CatC^\prime$.

 \begin{definition}
For any cospan $C$ in $\CatC$, a span $S$ in $\CatC$ is an \emph{$\F$-pullback} of $C$ if $\F(S)$ is a pullback of $\F(C)$.
 \end{definition}

\begin{definition}
The category $\CatC$ \emph{has $\F$-pullbacks} if every cospan in $\CatC$ admits an $\F$-pullback.
\end{definition}

This diagram illustrates the relationship between $S$, $C$, $\F(S)$, and $\F(C)$ when $S$ is an $\F$-pullback of $C$:

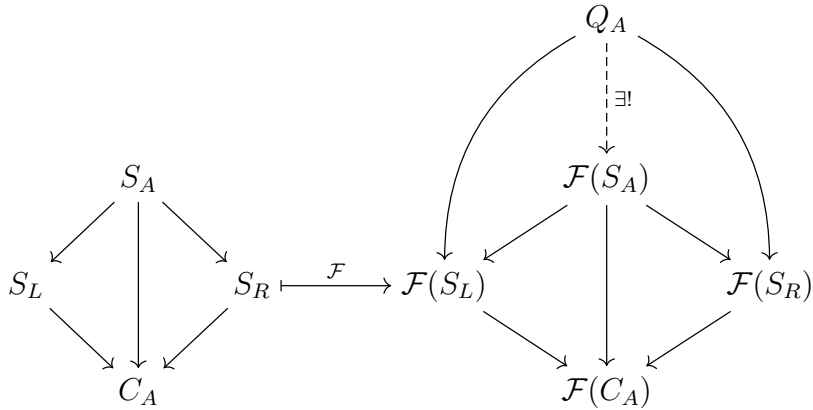
\smallskip

\begin{center}
\begin{tikzcd}[cramped,column sep=scriptsize]
	&&&&& {Q_A} \\
	\\
	& {S_A} &&&& {\F(S_A)} \\
	{S_L} && {S_R} && {\F(S_L)} && {\F(S_R)} \\
	& {C_A} &&&& {\F(C_A)}
	\arrow["{\exists !}", dashed, from=1-6, to=3-6]
	\arrow[bend left=-30, from=1-6, to=4-5]
	\arrow[bend right=-30, from=1-6, to=4-7]
	\arrow[from=3-2, to=4-1]
	\arrow[from=3-2, to=4-3]
	\arrow[from=3-2, to=5-2]
	\arrow[from=3-6, to=4-5]
	\arrow[from=3-6, to=4-7]
	\arrow[from=3-6, to=5-6]
	\arrow[from=4-1, to=5-2]
	\arrow["\F", maps to, from=4-3, to=4-5]
	\arrow[from=4-3, to=5-2]
	\arrow[from=4-5, to=5-6]
	\arrow[from=4-7, to=5-6]
\end{tikzcd}
\end{center}

\smallskip

\noindent The functor $\F$ maps each morphism $\alpha$ in the span $S$ to a morphism $\F(\alpha)$ in the span $\F(S)$.

A closely related concept is that of an $\F$-product and of $\CatC$ having $\F$-products. Define an \emph{$\F$-product} as an $\F$-pullback over a terminal object. Since the pullback of a cospan with morphisms into a terminal object is a product, if $\CatC$ has $\F$-pullbacks and a terminal object, and if $\F$ preserves terminal objects, then $\CatC$ has $\F$-products. If $\F$ is the identity functor and $S$ is an $\F$-pullback of a cospan $C$, then $S$ is a pullback of $C$. The notion of an $\F$-pullback thus generalizes the standard notion of a pullback.

For any objects $A$, $B$, and $C$ in $\CatC$ with morphisms $f\colon A\to C$ and $g\colon B\to C$, if $\CatC$ is the category $\Set$, then the fiber product $A\times_C B$ is the set
\[
A\times_C B = \big\{(a,b)\in A\times B\mid f(a) = g(b)\big\}.
\]
It is a pullback of $\rangle f, g\langle$, and if $\F$ is the identity functor on $\Set$, then it is also an $\F$-pullback of $\rangle f, g\langle$. Continue to write $A\times_C B$ for the apex of an $\F$-pullback of $\rangle f, g\langle$, and write $A\times B$ when $C$ is terminal (so $f$ and $g$ are unique). In the latter case, $A\times B$ is an $\F$-product and the morphisms from it to $A$ and $B$ are $\F$-projections.

Prior work restricted the notion of span-tightness to functors $\F$ whose source category admits $\F$-pullbacks \cite{WY}.  This condition can be dropped.  

\begin{definition}
A functor $\F\colon \CatC \to \CatC^\prime$ is \emph{span-tight} if for any cospan $C$ in $\CatC$, and for every pair of $\F$-pullbacks $S$ and $Q$ of $C$, the unique span morphism $\Phi$ from $\F(S)$ to $\F(Q)$ is of the form $\F(\Psi)$ for some $\Psi$ in $\CatC$.
\end{definition}

For any spans $S$ and $Q$ in $\CatC$ so that $\SR$ is equal to $\QL$ and $\rangle \sR, \qL\langle$ has an $\F$-pullback $P$, the composite $[S]\circ [Q]$ is the span $[\langle \sL\circ\pL, \qR\circ\pR\rangle]$.  Composition, when defined, depends neither on the choice of representatives $S$ and $Q$, nor on the choice of $\F$-pullback $P$.  For any object $X$ in $\CatC$, denote by $\Id_X$ the identity arrow with source and target equal to $X$.  Identify $\SR$ as the source of $[S]$ and $\SL$ as the target of $S$.  Theorem~\ref{GenSpanCat:Theorem} is a main result of the prior work \cite[Theorem~5.1]{WY}.

\begin{theorem}\label{GenSpanCat:Theorem}
If $\CatC$ has $\F$-pullbacks and is span-tight, then $\circ$ defines composition in a category $\Span(\F)$ whose objects are the objects of $\CatC$, whose morphisms are isomorphism classes of spans in $\CatC$, and whose identity morphism at any object $X$ is $[\langle\Id_X, \Id_X\rangle]$. 
\end{theorem}

For appropriate choices of $\CatC$ and $\F$, the prior works identify open systems with morphisms in $\Span(\F)$ (formerly denoted $\Span(\C,\F)$) \cite{BWY,WY}.  Composition in $\Span(\F)$ directly treats open systems connected in an acyclic chain, like this:

\smallskip

\begin{center} 
\begin{tikzpicture}[scale = 1]
\begin{scope}[xshift = 0in]
\coordinate(A) at (0,0);
\coordinate(B) at (1,0);
\coordinate(C) at (2,0);
\coordinate(D) at (3,0);
\coordinate(E) at (4,0);

\draw[shorten >= 4pt, shorten <= 4pt] (A) -- (B);
\draw[shorten >= 4pt, shorten <= 4pt] (B) -- (C);
\draw[shorten >= 4pt, shorten <= 4pt] (D) -- (E);

\draw[fill = black] (A) circle (2pt);
\draw[fill = black] (B) circle (2pt);
\draw[fill = black] (C) circle (2pt);
\draw[fill = black] (D) circle (2pt);
\draw[fill = black] (E) circle (2pt);

\draw[fill = black] (2.375, 0) circle (.5pt);
\draw[fill = black] (2.5, 0) circle (.5pt);
\draw[fill = black] (2.625, 0) circle (.5pt);

\end{scope}
\end{tikzpicture}
\end{center} 

\smallskip

\noindent Systems that cannot be expressed as an acyclic chain as above appear to require a more flexible notion of span composition. Example~\ref{nonexample:Flimitb} demonstrates that spans of surjective submersions arranged in an undirected graph may be impossible to compose, so a framework for generalized span composition in an undirected graph must identify when composition is permissible.

A \emph{cone} $\Phi^X$ over a diagram $\D\colon \CatJ\to\CatC$ is an object $X$ in $\CatC$, the \emph{apex} of the cone, together with \emph{component morphisms} (\emph{legs}) $\Phi^X_Z\colon X\to \D(Z)$ for each object $Z$ in $\CatJ$ such that for any morphism $f\colon A\to B$ in $\CatJ$,
\[
\D(f)\circ \Phi^X_A = \Phi^X_B.
\]
A span paired with a cospan in $\CatC$ is a cone over that cospan. For any diagram $\D\colon \CatJ\to\CatC$ and any objects $X$ and $Y$ in $\CatC$, take $\Phi^X$ and $\Phi^Y$ to be cones over $\D$. A morphism $\Psi\colon X\to Y$ is a \emph{cone morphism} if for every object $A$ in $\CatJ$,
\[
\Phi^X_A = \Phi^Y_A\circ \Psi.
\]
A cone morphism $\Psi$ is a \emph{cone isomorphism} if $\Psi$ is an isomorphism in $\CatC$. Cones $\Phi^X$ and $\Phi^Y$ are \emph{cone isomorphic} if there exists a cone isomorphism from $Y$ to $X$.

A \emph{limit} $\Phi^X$ of a diagram $\D$ is a cone over $\D$ with the property that for any cone $\Phi^Y$ over $\D$, there exists a unique cone morphism $\Psi\colon \Phi^Y \to \Phi^X$. As in the case of spans, limits are too restrictive for the categories relevant here. This motivates the introduction of $\F$-limits, which generalize $\F$-pullbacks. Indeed, when $\F$ is the identity functor on $\CatC$, an $\F$-limit is a limit. If $\CatC$ is the category $\Set$, then the apex of a limit of a small diagram $\D\colon \CatJ \to \CatC$ is isomorphic to
\[
\left\{x \in \prod_{j \in \text{Ob}(\CatJ)}\D(j) \hspace{10pt} \Bigg| \hspace{10pt} \forall [f\colon i \rightarrow j] \in \text{Mor}(\CatJ), \D(f)(\pi_i(x)) = \pi_j(x) \right\}.
\]
See the textbook of Leinster for further discussion of this construction \cite{LeinsterBasic}.

\begin{definition}
    For any two categories $\CatC$ and $\CatC^\prime$ and functor $\F$ from $\CatC$ to $\CatC^\prime$, an $\F$-limit of a diagram $\D$ is a cone $\Phi$ over $\D$ so that $\F(\Phi)$ is a limit in $\CatC^\prime$ of $\F\circ\D$. 
\end{definition}

Given any index set $I$, category $\CatC$, objects $C_i$ in $\CatC$ for each $i$ in $I$, object $A$ in $\CatC$, and any set of morphisms \[F = \{f_i\colon A \rightarrow C_i\mid i \in I\}\] in $\CatC$, the set $F$ defines a cone over the discrete diagram formed by the collection $C_i$. If there is an $\F$-product $P$ of the collection $C_i$ (equivalently, an $\F$-limit of the discrete diagram $\D$ from $I$ to the collection $C_i$), then an $\F$-product morphism \[\bigtimes_{i \in I} f_i\colon A \rightarrow P\] is a cone morphism from $A$ to $P$. This makes precise what it means for $A$ to factor through a product when a product does not exist, but an $\F$-product does.

Technical issues arise in the construction of a category that has subsystem inclusions as its morphisms, namely the existence of certain $\F$-limits and uniqueness up to isomorphism of the source of such $\F$-limits.  These issues motivate some restriction on both $\F$ and on the system of actors to be studied.

\begin{definition}
A functor $\F\colon \CatC\to \CatC^\prime$ is \emph{cone-tight} if for every diagram $\D\colon \CatJ\to \CatC$ and every pair of $\F$-limits $\Phi^X$ and $\Phi^{X^\prime}$ of $\D$, the unique cone isomorphism $\Psi\colon \F(X)\to \F(X^\prime)$ from the cone $\F(\Phi^X)$ to the cone $\F(\Phi^{X^\prime})$ is $\F(\widetilde{\Psi})$ for some cone isomorphism $\widetilde{\Psi}\colon X\to X^\prime$.
\end{definition}

If $\F$ is cone-tight and preserves terminal objects which $\C$ and $\C'$ have, then $\F$ is span-tight. Any finite limit can be constructed by pullbacks \cite[Section~3.4]{Riehl}.  The question is whether this is also true for $\F$-limits with respect to $\F$-pullbacks. In general, the answer is no (see Example~\ref{nonexample:Flimitb}), which complicates the present work. The goal of the next section is to identify a structured class of diagrams for which $\F$-limits exist.


\section{ACM-diagrams and their $\F$-limits}\label{Sec3}


The central construction of this section is the \emph{ACM-diagram}. Lemma~\ref{lem:CanAlwaysWeldSimple} provides the principal technical ingredient in proving that $\F$-limits exist for ACM-diagrams that are \emph{reducible to decomposable} by \emph{welding} actors.

\subsection{ACM-diagrams and weldings}\label{ACM-diagrams-and-weldings}

Adding data that specifies how parts connect allows composites beyond those arising from CM-systems with acyclic skeletons. Section~\ref{Sec:CMK-systems-as-rigid-inclusions} justifies this requirement and explains its implications for classical kinematics. Fix categories $\CatC$ and $\CatC^\prime$ and a functor $\F\colon \CatC \to \CatC^\prime$ such that $\CatC$ admits $\F$-pullbacks in $\CatC^\prime$. Assume that $\CatC$ and $\CatC^\prime$ admit terminal objects and that $\F$ preserves them.

\begin{definition}
An \emph{actor index category} $\CatJ$ is a finite poset whose set of objects decompose as union \[\Ob(\CatJ) = \Ob_A(\CatJ) \cup \Ob_C(\CatJ) \cup \Ob_I(\CatJ),\] 
where $\Ob_A(\CatJ)$, $\Ob_C(\CatJ)$, and  $\Ob_I(\CatJ)$ are disjoint sets and
\begin{enumerate}
\item[(i)] for some $N$ in $\mathds N$, $\Ob_A(\CatJ)$ is a set of $N$ distinct elements called \emph{actor indices};
\item[(ii)] for each actor index $i$, there is an $M_i$ in $\mathds N$ and a set $C_i$ of \emph{constraint indices for $i$} with \[C_i=\{c^i_1,\dots,c^i_{M_i}\}\subseteq\Ob_C(\CatJ)\]
         that all contain a single distinguished object $\star$, and \[\Ob_C(\CatJ) = \bigcup_{i\in \Ob_A(\CatJ)} C_i;\]
\item[(iii)] the set of \emph{interaction indices} is the set \[\Ob_I(\CatJ)=\{\{i,j\}\mid i,j \in \Ob_A(\CatJ), i\ne j\}.\]
\end{enumerate} 
The order is generated by \[c\le i \quad \text{and} \quad  i\le\{i,j\}\] for every $c$ in $C_i$ and for every pair of distinct actor indices $i$ and $j$.
\end{definition}

\begin{definition}\label{defACMDiag}
An \emph{Actor-Constraint mediated diagram} $\D$ (henceforth an \emph{ACM-diagram}) is, for some actor index category $\CatJ$, a diagram of shape $\CatJ$ in $\CatC$ with the following specified properties: 
\begin{enumerate}
\item[(i)] $\D$ takes each actor index $i$ to an object $A_i$ in $\CatC$, called an \emph{actor}, and each constraint index $c$ in $C_i$ to an object in $\CatC$, called a \emph{constraint} for $A_i$.
\item[(ii)] $\D$ takes each morphism $c \leq i$ to a \emph{constraint morphism} $f_{i,c}\colon A_i \rightarrow \D(c)$.
\item[(iii)] $\D$ takes $\star$ to a terminal object in $\CatC$ and $\D$ takes no other object in $\CatJ$ to a terminal object in $\CatC$.
\item[(iv)]\label{ACMcond4} For any actors $A_i$ and $A_j$ with $i$ and $j$ distinct, there are $\F$-product morphisms \[\bigtimes_{c \in C_i \cap C_j} f_{i,c} \colon A_i \rightarrow \prod_{c \in C_i \cap C_j} \D(c) \quad \text{and} \quad \bigtimes_{c \in C_i \cap C_j} f_{j,c} \colon A_j \rightarrow \prod_{c \in C_i \cap C_j} \D(c).\]
\item[(v)] For any two distinct $i$ and $j$, the functor $\D$ takes the span with source $\{i, j\}$ and targets $i$ and $j$ to an $\F$-pullback $\KSpan(A_i,A_j)$ of the cospan given by the pair of morphisms \[\bigtimes_{c \in C_i \cap C_j} f_{i,c}: \D(i) \rightarrow \prod_{c \in C_i \cap C_j} \D(c) \quad \text{and} \quad \bigtimes_{c \in C_i \cap C_j} f_{j,c}\colon \D(j) \rightarrow \prod_{c \in C_i \cap C_j} \D(c).\]
\end{enumerate}
\end{definition}

Although the current work does not require condition~(v) in the definition of an ACM-diagram, it appears to be necessary for dynamics, where additional information such as a choice of potential depends on pairs of actors rather than on individual actors.

\begin{definition}[Constrained actor]
A \emph{constrained actor $A_i$ with constraint set $\Con(A_i)$}, given by a pair $(A_i,\Con(A_i))$, is an ACM-diagram $\A\colon \CatJ \to \CatC$ where $\CatJ$ is an actor-index category with a single actor index $i$. There exists a finite index set $K$ and objects $B_k$ of $\CatC$ for each $k \in K$ such that
\[
\Con(A_i) = \{ f_{i,k} \colon A_i \to B_k \mid k \in K \}
\]
is a non-empty set of morphisms. Refer to $\Con(A_i)$ as the \emph{constraint set of $A_i$}. Each $f_{i,k}$ is a \emph{constraint morphism}.
\end{definition}

Every actor considered in this work is a constrained actor. Accordingly, adopt the convention of writing $A_i$ for both the actor and the constrained actor $(A_i,\Con(A_i))$.  Define the set of constraints of a diagram $\D$ by 
\[
\Con(\D)=\big\{\D(c\le a)\ \bigm|\ a\in \Ob_A(\CatJ),\ c\in \Ob_C(\CatJ),\ c\le a\big\}.
\]

For any category $\CatC$, denote by $\Cat\downarrow \CatC$ the slice category whose objects are diagrams in $\CatC$ and whose morphisms are functors between index categories that commute with those diagrams \cite[Section~2.5, Ex.~12]{Goldblatt}, as illustrated here:
\[
\begin{tikzcd}
	\CatJ && {\CatJ'} \\
	\\
	& \CatC
	\arrow["\F", from=1-1, to=1-3]
	\arrow["\D"', from=1-1, to=3-2]
	\arrow["{\D'}", from=1-3, to=3-2]
\end{tikzcd}
\]
\begin{definition}
For any categories $\CatJ_1$ and $\CatJ_2$, an \emph{inclusion} $\iota$ from $\CatJ_1$ to $\CatJ_2$ is a faithful functor
\[
\iota\colon \CatJ_1\hookrightarrow \CatJ_2
\]
that is injective on $\Ob(\CatJ_1)$. In this case, $\CatJ_1$ is a \emph{subcategory} of $\CatJ_2$.

If $\CatJ_1$ and $\CatJ_2$ are actor index categories, then $\iota$ \emph{respects the ACM structure} if it sends actor indices to actor indices, constraint indices to constraint indices, interaction indices to interaction indices, and satisfies \[\iota(\star) = \star.\] For any diagram $\D_2\colon \CatJ_2\to \CatC$, define
\begin{equation}\label{def:eq:ACMSub}
\D_1=\D_2\circ \iota .
\end{equation}
Denote by $\iota^\ast\colon \D_1\to \D_2$ the induced morphism in $\Cat\downarrow \CatC$. Call $\iota^\ast$ an \emph{inclusion} of $\D_1$ into $\D_2$. If $\iota$ respects the ACM structure, then call $\iota^\ast$ an \emph{ACM subdiagram} of $\D_2$, and call $\D_1$ an \emph{ACM subdiagram} of $\D_2$.
\end{definition}

An inclusion $\iota^\ast$ of ACM-diagrams that respects the ACM structure sends actor, constraint, and interaction indices to actor, constraint, and interaction indices. Since the $\F$-product morphisms and interactions for the target diagram exist in $\CatC$, they also exist for the source diagram. A direct check of the definition of an ACM-diagram yields Lemma~\ref{lem:InclusionsAreSubdiagrams}.

\begin{lemma}\label{lem:InclusionsAreSubdiagrams}
If $\iota\colon \CatJ_1 \rightarrow \CatJ_2$ is an inclusion that respects the ACM structure of actor index categories $\CatJ_1$ and $\CatJ_2$ and $\D_2\colon \CatJ_2 \rightarrow \CatC$ is an ACM-diagram, then \[\D_1=\D_2\circ \iota\] is an ACM-diagram.
\end{lemma}

The way in which $\D_1$ is a subdiagram of $\D_2$ depends on the specified inclusion. A subdiagram inclusion involves only an inclusion functor from one actor index category to another, and this inclusion need not be a set-inclusion. This has physical significance: there may be multiple ways in which a part could fit into a larger system; an inclusion is a choice of how that part fits into the whole.

An actor index category $\CatJ$ is a finite poset, and it makes sense to identify a subset $\CatJ'$ of $\CatJ$ with the induced partial order. Take $\CatJ'$ to be an actor index category such that %
\[
\Ob_A(\CatJ')\subseteq \Ob_A(\CatJ),\quad \Ob_C(\CatJ')\subseteq \Ob_C(\CatJ)\] and \[\Ob_I(\CatJ')=\{\{i,j\}\mid i,j\in \Ob_A(\CatJ'),\ i\ne j\}.
\]
Denote by $\iota$ the set-inclusion functor $\CatJ'\hookrightarrow \CatJ$ that preserves the partial order. The restriction \[\D'=\D|_{\CatJ'}\] is an ACM subdiagram of $\D$.

The category $\Cat$ of locally small categories, itself not locally small, has pullbacks \cite[Proposition~3.5.6]{Riehl}. Monomorphisms in $\Cat$ are the faithful functors that are injective on objects, hence they are the inclusions in $\Cat$. In any category, for any cospan $\rangle f,g\langle$ with pullback $\langle \rho_1,\rho_2\rangle$, if $f$ is a monomorphism then $\rho_2$ is a monomorphism \cite[Exercise~3.13.1]{Goldblatt}. A pullback of the inclusions \[\iota_1\colon \CatJ_1 \rightarrow \CatJ \quad \text{and} \quad \iota_2\colon \CatJ_2 \rightarrow \CatJ\] thus consists of a category $\CatJ_{12}$ and inclusions \[\pi_1\colon \CatJ_{12} \rightarrow \CatJ_1 \quad \text{and} \quad \pi_2\colon \CatJ_{12} \rightarrow \CatJ_2.\]

\begin{definition}
Take $\D_1$ and $\D_2$ to be subdiagrams of $\D$, where
\[
\D_1\colon \CatJ_1 \to \CatC,\qquad
\D_2\colon \CatJ_2 \to \CatC,\qquad
\D\colon \CatJ \to \CatC,
\]
and take
\[
\iota_1\colon \CatJ_1 \to \CatJ
\quad\text{and}\quad
\iota_2\colon \CatJ_2 \to \CatJ
\]
to be the corresponding inclusions. An \emph{intersection of subdiagrams} $\D_1$ and $\D_2$ in $\D$ consists of a pullback $\langle \pi_1,\pi_2\rangle$ with apex $\CatJ_{12}$ of the cospan $\rangle \iota_1,\iota_2\langle$ and a diagram $\D_{12}\colon \CatJ_{12}\to \CatC$ such that
\[
\D_{12}=\D_1\circ \pi_1=\D_2\circ \pi_2.
\]
Equivalently, $\langle \pi_1^\ast,\pi_2^\ast\rangle$ is an $\F$-pullback of $\iota_1^\ast$ and $\iota_2^\ast$ in $\Cat\downarrow \CatC$.

If $\D_1$ and $\D_2$ are ACM subdiagrams of $\D$, then $\langle \pi_1,\pi_2\rangle$ is an \emph{intersection of ACM subdiagrams} in $\D$. Call $\D_{12}$ an \emph{intersection of ACM subdiagrams}. If, in addition,
\[
\Ob_A(\CatJ)=\iota_1\big(\Ob_A(\CatJ_1)\big)\cup \iota_2\big(\Ob_A(\CatJ_2)\big)
\quad\text{and}\quad
\Ob_C(\CatJ)=\iota_1\big(\Ob_C(\CatJ_1)\big)\cup \iota_2\big(\Ob_C(\CatJ_2)\big),
\]
then $\rangle \iota_1,\iota_2\langle$ is a \emph{union of ACM subdiagrams} over $\D_{12}$, and $\D$ is a \emph{union of ACM subdiagrams} $\D_1$ and $\D_2$.
\end{definition}

In $\mathsf{Set}$, the union of two sets is a pushout over the inclusions of their intersection. A union of ACM subdiagrams is not always a pushout in $\mathsf{Cat}$; rather, it is obtained by forming a pushout on the subdiagrams consisting of constraints and actors and then constructing interactions for each pair of actors. If $\CatC$ has $\F$-pullbacks, then a pushout of two ACM-diagrams determines an ACM-diagram whenever the pushout satisfies condition~\ref{ACMcond4}~(iv), namely that for any two actors there exist $\F$-product morphisms into the $\F$-product of the constraints in the intersection of their constraint sets. Proposition~\ref{sec3:prop:ACMintersections} shows that the graph-theoretic intersection of $\iota_1(\CatJ_1)$ and $\iota_2(\CatJ_2)$ determines a pullback of ACM-diagrams.

\begin{proposition}\label{sec3:prop:ACMintersections}
For any ACM subdiagrams $\D_1$ and $\D_2$ of $\D$ with inclusions on shapes
\[
\iota_1\colon \CatJ_1 \rightarrow \CatJ
\quad\text{and}\quad
\iota_2\colon \CatJ_2 \rightarrow \CatJ,
\]
an intersection $\D_{12}$ of $\D_1$ and $\D_2$ exists.

If for every $c$ in $\Ob_C(\iota_1(\CatJ_1)\cap \iota_2(\CatJ_2))$ there exists an 
$a$ in $\Ob_A(\iota_1(\CatJ_1)\cap \iota_2(\CatJ_2))$ with $c\le a$, then $\D_{12}$ is an ACM-diagram, is isomorphic to $\iota_1(\CatJ_1)\cap \iota_2(\CatJ_2)$, and is an ACM subdiagram of both $\D_1$ and $\D_2$.
\end{proposition}

\begin{proof}
Since $\Cat$ has pullbacks, the slice category $\Cat\downarrow \CatC$ has pullbacks \cite[Proposition~3.3.8]{Riehl}. Hence an intersection $\D_{12}$ of $\D_1$ and $\D_2$ exists.

Define $\CatJ_{12}$ to be the graph-theoretic intersection of $\iota_1(\CatJ_1)$ and $\iota_2(\CatJ_2)$ inside $\CatJ$, namely, take
\[
\Ob(\CatJ_{12})=\iota_1\bigl(\Ob(\CatJ_1)\bigr)\cap \iota_2\bigl(\Ob(\CatJ_2)\bigr),
\]
and take the morphisms of $\CatJ_{12}$ to be the morphisms of $\CatJ$ between objects of $\CatJ_{12}$. Composition and identities restrict from $\CatJ$, so $\CatJ_{12}$ is a subcategory of $\CatJ$. 

For $i$ in $\{1,2\}$ define $\pi_i\colon \CatJ_{12}\to \CatJ_i$ as follows. 
For each object $x$ of $\CatJ_{12}$, since $x$ is in $\iota_i\big(\Ob(\CatJ_i)\big)$ and $\iota_i$ is injective on objects, there is a unique object $\pi_i(x)$ of $\CatJ_i$ such that 
\begin{equation}\label{eq:iotacircpixisx}
\iota_i\big(\pi_i(x)\big)=x.
\end{equation}
For any morphism $f\colon x\to y$ in $\CatJ_{12}$, since $x$ and $y$ are in $\iota_i\big(\Ob(\CatJ_i)\big)$ and $\iota_i(\CatJ_i)$ is a subcategory of $\CatJ$, the morphism $f$ is in $\Hom_{\iota_i(\CatJ_i)}(x,y)$. Since $\iota_i$ is faithful, there exists a unique morphism $\pi_i(f)\colon \pi_i(x)\to \pi_i(y)$ in $\CatJ_i$ such that %
\[
\iota_i\big(\pi_i(f)\big)=f.
\]%
The functors $\pi_1$ and $\pi_2$ form a pullback of the cospan $\rangle \iota_1,\iota_2\langle$ in $\Cat$.

Given the hypothesis on constraints in $\iota_1(\CatJ_1)\cap \iota_2(\CatJ_2)$, the category $\CatJ_{12}$ is an actor index category: all constraint indices are associated with actors; for any two actor indices in the intersection, the interaction index between those two actors is also in the intersection; and it contains $\star$ because $\iota_1$ and $\iota_2$ preserve the distinguished constraint index $\star$.

Define the diagram
\[
\D_{12}\colon \CatJ_{12}\to \CatC
\quad\text{by}\quad
\D_{12}=\D\circ \jmath,
\]
where $\jmath\colon \CatJ_{12}\hookrightarrow \CatJ$ denotes the inclusion. Equivalently, \begin{equation}\label{Eq:D12intersect}\D_{12}=\D_1\circ \pi_1=\D_2\circ \pi_2.\end{equation}
Equation~\ref{eq:iotacircpixisx} implies that the inclusions $\pi_i$ for $i$ in $\{1, 2\}$ respect the ACM structure since each $\iota_i$ respects the ACM structure. Since each $\pi_i$ respects the ACM structure, Lemma~\ref{lem:InclusionsAreSubdiagrams} implies that $\D_{12}$ is an ACM-diagram, and \eqref{Eq:D12intersect} exhibits $\D_{12}$ as an ACM subdiagram of both $\D_1$ and $\D_2$.

\end{proof}

Commutativity of this diagram captures the data of Proposition \ref{prop:IsomorphicUnions}:

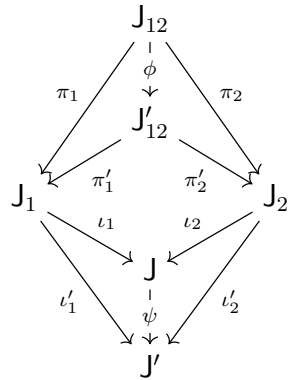
\smallskip

\begin{center}
\begin{tikzcd}[cramped,row sep=small]
	& {\CatJ_{12}} \\
	\\
	& {\CatJ_{12}'} \\
	{\CatJ_1} && {\CatJ_2} \\
	& \CatJ \\
	\\
	& {\CatJ'}
	\arrow["\phi"{description}, from=1-2, to=3-2]
	\arrow["{\pi_1}"', from=1-2, to=4-1]
	\arrow["{\pi_2}", from=1-2, to=4-3]
	\arrow["{\pi_1'}", from=3-2, to=4-1]
	\arrow["{\pi_2'}"', from=3-2, to=4-3]
	\arrow["{\iota_1}", from=4-1, to=5-2]
	\arrow["{\iota_1'}"', from=4-1, to=7-2]
	\arrow["{\iota_2}"', from=4-3, to=5-2]
	\arrow["{\iota_2'}", from=4-3, to=7-2]
	\arrow["\psi"{description}, from=5-2, to=7-2]
\end{tikzcd}
\end{center}

\smallskip

\begin{proposition}\label{prop:IsomorphicUnions}
For any ACM-diagrams
\[
\D_1\colon \CatJ_1 \to \CatC,\quad
\D_2\colon \CatJ_2 \to \CatC,\quad
\D\colon \CatJ \to \CatC,\quad \text{and} \quad
\D'\colon \CatJ' \to \CatC
\]
such that $\D$ and $\D'$ are each unions of $\D_1$ and $\D_2$ with inclusions on shapes
\[
\iota_1\colon \CatJ_1\to \CatJ,\quad
\iota_2\colon \CatJ_2\to \CatJ,\quad
\iota_1'\colon \CatJ_1\to \CatJ',\quad \text{and}\quad 
\iota_2'\colon \CatJ_2\to \CatJ',
\]
take $\langle \pi_1,\pi_2\rangle$ with apex $\CatJ_{12}$ to be an intersection (pullback in $\Cat$)
of the cospan $\rangle \iota_1,\iota_2\langle$, and $\langle \pi_1',\pi_2'\rangle$ with apex
$\CatJ_{12}'$ to be an intersection of $\rangle \iota_1',\iota_2'\langle$. If there exists an isomorphism of categories $\phi\colon \CatJ_{12}\to \CatJ_{12}'$
such that for $i$ in $\{1,2\}$,
\[
\pi_i'\circ \phi=\pi_i,
\]
then there exists an isomorphism of categories $\psi\colon \CatJ\to \CatJ'$ that respects the ACM
structure and a natural isomorphism
\[
\eta\colon \D \to \D'\circ \psi .
\]
\end{proposition}

\begin{proof}
Define \[\psi \colon \CatJ\to \CatJ' \quad \text{and} \quad \psi'\colon \CatJ' \rightarrow \CatJ\] in the following way: for any $x$ in $\iota_1(\CatJ_1)$ and $y$ in $\iota_2(\CatJ_2)$, take \[\psi(x)=\iota_1'\big(\iota_1^{-1}(x)\big), \quad \psi(y)=\iota_2'\big(\iota_2^{-1}(y)\big), \quad \psi'(x)=\iota_1\big(\iota_1'^{-1}(x)\big), \quad \text{and} \quad \psi'(y)=\iota_2\big(\iota_2'^{-1}(y)\big).\]
If $x$ is in $\iota_1(\CatJ_1)\cap \iota_2(\CatJ_2)$, then 
\[
x=\iota_1(\pi_1(z))=\iota_2(\pi_2(z))
\]
for some $z$ in $\CatJ_{12}$. The equalities 
\[
\pi_i'\circ \phi=\pi_i \quad \text{and} \quad \iota_1'\circ \pi_1'=\iota_2'\circ \pi_2'
\] 
in $\CatJ_{12}'$ together imply that 
\[
\iota_1'\big(\iota_1^{-1}(x)\big)
=\iota_1'\big(\pi_1(z)\big)
=\iota_1'\big(\pi_1'(\phi(z))\big)
=\iota_2'\big(\pi_2'(\phi(z))\big)
=\iota_2'\big(\pi_2(z)\big)
=\iota_2'\big(\iota_2^{-1}(x)\big),
\]
so the definition of $\psi$ is consistent on the intersection $\iota_1(\CatJ_1)\cap \iota_2(\CatJ_2)$. An equivalent argument verifies that $\psi'$ is also consistent on this intersection.

Extend $\psi$ and $\psi'$ to interaction indices by
\[
\psi(\{a,b\})=\{\psi(a),\psi(b)\} \quad \text{and} \quad \psi'(\{a,b\})=\{\psi'(a),\psi'(b)\}.
\]
Since $\CatJ$ and $\CatJ'$ are posets whose orders are generated by relations $c\le a$ and $a\le \{a,b\}$, each of $\psi$ and $\psi'$ define an order-preserving map on objects and hence a functor.

For any $x$ in $\CatJ$, there is an $i$ in $\{1, 2\}$ such that $x$ is in $\iota_i(\CatJ)$, hence \[\psi' \circ \psi(x) = \iota_i(\iota_i^{\prime-1}(\iota_i'(\iota_i^{-1}(x)))) = x \quad  \text{and} \quad \psi \circ \psi'(x) = \iota_i'(\iota_i^{-1}(\iota_i(\iota_i^{\prime-1}(x)))) = x.\] The functors $\psi'$ and $\psi$ are inverses, hence $\psi$ is an isomorphism. The isomorphism $\psi$ respects actor, constraint, and interaction indices, and sends $\star$ to $\star$.

Define a natural transformation $\eta\colon \D\to \D'\circ \psi$ componentwise.
\begin{enumerate}
\item For any $x$ in $\iota_1(\CatJ_1)$,
\[
\D(x)=\D_1(\iota_1^{-1}(x))
\quad\text{and}\quad
\D'(\psi(x))=\D'\bigl(\iota_1'(\iota_1^{-1}(x))\bigr)=\D_1(\iota_1^{-1}(x)),
\]
so take \[\eta_x=\id_{\D(x)}.\] Define $\iota_2(\CatJ_2)$ similarly for any $x$ in $\iota_2(\CatJ_2)$.
\item If $x$ is an interaction index $\{a,b\}$ with $a$ in $\iota_1(\CatJ_1)$ and $b$ in $\iota_2(\CatJ_2)$, then $\D(x)$ and $\D'(\psi(x))$ are (by the ACM-diagram axioms) chosen $\F$-pullbacks of the same cospan in $\CatC$,
because $\eta$ is already the identity on the actor objects and on the shared-constraint products.
The universal property of pullbacks implies that there exists a unique isomorphism %
\[
\eta_x\colon \D(x)\xrightarrow{\cong}\D'(\psi(x))
\]
that commutes with the pullback projections.
\end{enumerate}
Because the index categories are thin, these choices automatically satisfy naturality once they commute with the structure maps (in particular, the projection maps into the actor objects and shared-constraint products). Thus $\eta$ is a natural isomorphism.

\end{proof}

Take $\CatJ$ to be any actor index category and $i$ to be in $\text{Ob}_A(\CatJ)$. The \emph{complementary constraint set} to $i$ is the set
\[C_{i,\CatJ}^\perp = \bigcup_{j \neq i} C_j.\]
This is the set of all constraint indices for actor indices in $\text{Ob}_A(\CatJ)$ other than $i$. The \emph{external constraint set} of $i$ is the set
\[\text{Ext}[i:\CatJ] = C_i \cap C_{i,\CatJ}^\perp,\]
which may be empty and consists of all constraint indices involving both $i$ and at least one actor index distinct from $i$. The set $\Ext[i:\CatJ]$ and the mappings to it capture the connection of the actor $A_i$ to the subsystem obtained by excluding $A_i$.

Take $\iota\colon \CatJ' \rightarrow \CatJ$ to be an inclusion that respects the ACM structure of another actor index category $\CatJ'$. The \emph{constraint set} of $\CatJ'$ is the set
\[C_{\CatJ'} = \bigcup_{j \in \iota(\text{Ob}_A(\CatJ'))} C_j.\]
This is the set of all constraint indices in $\CatJ$ involving at least one actor index in $\iota(\CatJ')$. The set $C_{\CatJ'}$ is not necessarily equal to $\iota(\text{Ob}_C(\CatJ'))$ because there may be constraints for actors in $\iota(J')$ that are not constraints in $\iota(J')$. The \emph{complementary constraint set} to $\CatJ'$ is the set
\[C_{\CatJ', \CatJ}^\perp = \bigcup_{a_j \notin \iota(\text{Ob}_A(\CatJ'))} C_j.\]
This is the set of all constraint indices for actor indices in $\CatJ$ not in $\iota(\CatJ')$. The set $C_{\CatJ', \CatJ}^\perp$ is not necessarily equal to $\text{Ob}_C(\CatJ) \setminus \iota(\text{Ob}_C(\CatJ'))$ because actors in $\text{Ob}_A(\CatJ) \setminus \iota(\text{Ob}_A(\CatJ'))$ may have constraints in $\iota(\CatJ')$. The \emph{external constraint set} of $\CatJ'$ is the set
\[\text{Ext}[\CatJ': \CatJ] = C_{\CatJ'} \cap C_{\CatJ', \CatJ}^\perp,\] whose elements are the constraint indices for constraints that are between at least one actor for an actor index in $\CatJ'$ and one actor for an actor index that is not in $\CatJ'$.

\begin{definition}\label{defSimple}
    An ACM-diagram $\D$ \emph{\simple} if for each actor $A$, there exists an $\F$-product morphism \[f_i := \bigtimes_{c \in \Ext[i\colon \CatJ]} f_{i,c}\colon A_i \rightarrow \prod_{c \in \Ext[i\colon \CatJ]} \D(c).\]
\end{definition}

There is an important difference between \ref{defACMDiag} (iv) and \ref{defSimple}. The former requires that for every pair of actors, each actor maps through a product of those two actors' shared constraints; the latter requires that every actor maps through an $\F$-product of all of its constraints that are shared with at least one other actor, which potentially is an $\F$-product of many more constraints.

\begin{definition}\label{defSuperSuperSimple}
An ACM-diagram $\D$ \emph{\elementary} if for each actor index $i$ with corresponding actor
\[
A_i:=\D(i),
\]
there is an $\F$-product morphism
\[
\bigtimes_{c\in C_i}\D(c\le i)\colon A_i\to \prod_{c\in C_i}\D(c)
\]
which is also an isomorphism.
\end{definition}

\begin{definition}
Take $A_i$ and $A_j$ to be any distinct actors in an ACM-diagram, and take
\[
C := \prod_{c \in C_i \cap C_j} \D(c)
\]
to be the $\F$-product of their shared constraints.  The \emph{welding} of $A_i$ and $A_j$ is the  actor $\KSpan(A_i,A_j)$ (see Definition~\ref{defACMDiag} (v)) with constraint set $C_i \cup C_j$. Refer to $\KSpan(A_i,A_j)$ as a \emph{welded actor}.
\end{definition}

Take $\D$ to be any ACM-diagram with actor-index category $\CatJ$ that contains distinct actor indices $i$ and $j$. Construct the actor-index category $W_{i,j}(\CatJ)$ by:
\begin{enumerate}
\item[$\bullet$] replacing $i$ and $j$ with a single index $ij$;
\item[$\bullet$] deleting the interaction index $\{i,j\}$ and its arrows;
\item[$\bullet$] assigning $ij$ the constraint set $C_{ij} =C_i\cup C_j$;
\item[$\bullet$] substituting $ij$ for every occurrence of $i$ or $j$ elsewhere.
\end{enumerate}
The actor index category $W_{i,j}(\CatJ)$ contains all actor indices of $\CatJ$ except $i$ and $j$, and instead contains the welded actor index $ij$.

\begin{definition}
For any ACM-diagram $\D$ with shape $\CatJ$ having actors $A_i$ and $A_j$, $\D$ is \emph{reducible by welding $A_i$ and $A_j$} if there is an ACM-diagram $\D^\prime$ on $W_{i,j}(\CatJ)$ so that $\D^\prime$ and $\D$ agree on every actor index, constraint index, and morphism common to $\CatJ$ and $W_{i,j}(\CatJ)$, and where $\D'$ takes: 
\begin{enumerate}
\item[$\bullet$] the welded actor $ij$ to  \[A_{ij} \coloneqq \D(\{i, j\}) = \KSpan(A_i, A_j);\]
\item[$\bullet$] any constraint morphism of the welded actor $c \leq ij$ to
\[
\D'(c \le ij)
\coloneqq
\begin{cases}
\D(c \leq i)\circ \D(i \leq \{i, j\}), & c\in C_i,\\[4pt]
\D(c \leq j)\circ \D(j \leq \{i, j\}), & c\in C_j;
\end{cases}
\]
\item[$\bullet$] any interaction index $\{ij,k\}$ to $\KSpan(A_{ij},A_k)$, an $\F$-pullback of $A_{ij}$ and $A_k$ over the $\F$-product $\prod_{c \in C_{ij} \cap C_k} \D(c)$,
\item[$\bullet$] any interaction morphism of the welded actor $ij \leq \{ij,k\}$ to the projection \[\pi_{ij}\colon \KSpan(A_{ij},A_k) \rightarrow A_{ij}.\]
\end{enumerate}

In this case write \[\D\xrightarrow{(i,j)}\D^\prime\] to mean that the diagram $\D$ reduces by welding $A_i$ and $A_j$ to the diagram $\D^\prime$.
\end{definition} 

Example~\ref{example:lock3Bar} shows that not every ACM-diagram is reducible by welding. Another actor $A_k$ may map through the product of the constraints it shares with $A_i$ and through the product of the constraints it shares with $A_j$, but not through the product of the constraints that it shares with the welded actor $\KSpan(A_i, A_j)$.  

\begin{definition}
    A \emph{chain of reductions of} an ACM-diagram $\D$ to an ACM-diagram $\D'$ is a finite sequence of ACM-diagrams $(\D_i)_{i \in \{1, \dots, n\}}$ such that \[\D = \D_1 \quad \text{and} \quad \D' = \D_n,\] and for each $i$ in $\{1, \dots, n-1\}$, the diagram $\D_{i+1}$ is a reduction of $\D_i$ by welding two actors of $\D_i$.
\end{definition}

\begin{definition}
An ACM-diagram $\D$ in $\CatC$ is \emph{reducible to a decomposable ACM-diagram} if it \simple\  or if there is a chain of reductions of $\D$ to $\D'$ such that $\D'$ \simple. In this case, the diagram $\D'$ is a \emph{decomposing reduction} of the diagram $\D$. 
\end{definition}

\subsection{$\F$-limits of ACM-diagrams}
Technical issues arise in constructing a category whose morphisms are subsystem inclusions, namely the existence of an $\F$-limit for a given ACM-diagram and uniqueness up to isomorphism of the apex of such an $\F$-limit. These issues motivate restricting both $\F$ and the class of ACM-diagrams.

\begin{definition}
For any diagram $\D\colon \CatJ\to \CatC$ and any cone $\Phi^S$ over $\D$ with apex $S$, define the \emph{cone diagram of $\D$ with apex $S$} to be the extension
\[
\D^S\colon \CatJ^S\to \CatC
\]
obtained as follows:
\begin{enumerate}
\item[$\bullet$] Take $\CatJ^S$ to have the same objects and morphisms as $\CatJ$, together with a new maximal object $s$ and a morphism $x\le s$ for each object $x$ of $\CatJ$.
\item[$\bullet$] Require that \[\D^S|_{\CatJ}=\D, \quad \D^S(s)=S,\] and for each object $x$ of $\CatJ$, take
\[
\D^S(x\le s)=\Phi^S_x\colon \D(x)\to S.
\]
\end{enumerate}
When $\Phi^X$ is an $\F$-limit cone over $\D$, refer to $\D^X$ as the \emph{$\F$-limit diagram of $\D$ with apex $X$}.
\end{definition}

The component morphisms of $\Phi^S$ form a cone over $\D$ if and only if the cone diagram $\D^S$ is commutative.

\begin{lemma}\label{lem:actorsonly}
    Take $\F$ to be a faithful functor and $\D$ to be an ACM-diagram. For any cone $\Phi^S$ over $\D$, the component morphisms into actors uniquely determine $\Phi^S$.
\end{lemma}

\begin{proof}
     For any distinct actor indices $i$ and $j$ in $\CatJ$, denote by $\Phi_{i}^S\colon S \rightarrow A_i$ the component morphism of the cone $\Phi^S$ into the actor $A_i$ and by $\Phi_{ij}^S$ the component morphism from $S$ to $\KSpan(A_i, A_j)$. This diagram represents the relevant morphisms and objects:
     

\smallskip

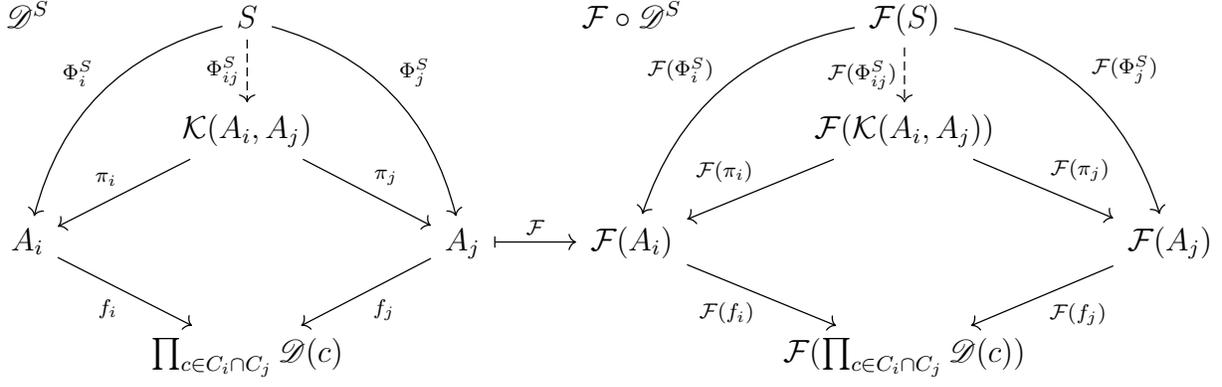
\begin{center}
\begin{tikzcd}
	{\D^S} & S && {\F \circ \D^S} & {\F(S)} \\
	& {\KSpan(A_i, A_j)} &&& {\F(\KSpan(A_i, A_j))} \\
	{A_i} && {A_j} & {\F(A_i)} && {\F(A_j)} \\
	& {\prod_{c \in C_i \cap C_j} \D(c)} &&& {\F(\prod_{c \in C_i \cap C_j} \D(c))}
	\arrow["{\Phi_{ij}^S}"', dashed, from=1-2, to=2-2]
	\arrow["{\Phi_i^S}"', bend left=-30, from=1-2, to=3-1]
	\arrow["{\Phi_j^S}", bend right=-30, from=1-2, to=3-3]
	\arrow["{\F(\Phi_{ij}^S)}"', dashed, from=1-5, to=2-5]
	\arrow["{\F(\Phi_i^S)}"', bend left=-30, from=1-5, to=3-4]
	\arrow["{\F(\Phi_j^S)}", bend right=-30, from=1-5, to=3-6]
	\arrow["{\pi_i}"', from=2-2, to=3-1]
	\arrow["{\pi_j}", from=2-2, to=3-3]
	\arrow["{\F(\pi_i)}"', from=2-5, to=3-4]
	\arrow["{\F(\pi_j)}", from=2-5, to=3-6]
	\arrow["{f_i}"', from=3-1, to=4-2]
	\arrow["\F", maps to, from=3-3, to=3-4]
	\arrow["{f_j}", from=3-3, to=4-2]
	\arrow["{\F(f_i)}"', from=3-4, to=4-5]
	\arrow["{\F(f_j)}", from=3-6, to=4-5]
\end{tikzcd}
\end{center}

\smallskip

The diagram $\D^S$ is commutative so
    \[\Phi_i^S = \pi_i \circ \Phi_{ij}^S \quad \text{and} \quad \Phi_j^S = \pi_j \circ \Phi_{ij}^S.\]
Functoriality of $\F$ implies that
    \begin{equation}\label{eq:faithful:lem}\F(\Phi_{i}^S) = \F(\pi_i) \circ \F(\Phi^S_{ij}) \quad \text{and} \quad \F(\Phi_{j}^S) = \F(\pi_j) \circ \F(\Phi_{ij}^S).\end{equation}
    The span $\langle \pi_i, \pi_j\rangle$ is an $\F$-pullback, so there is a unique cone morphism from $\F(S)$ to $ \F(\KSpan(A_i, A_j))$ that satisfies \eqref{eq:faithful:lem}, hence $\F(\Phi_{ij}^S)$ is the unique morphism that satisfies \eqref{eq:faithful:lem}. Faithfulness of $\F$ implies that $\Phi_{ij}^S$ is the unique morphism in $\C$ that $\F$ sends to $\F(\Phi_{ij}^S)$, so component morphisms into interactions are determined by the morphisms into actors. The diagram $\D^S$ commutes, so the component morphism $\Phi^S_c$ into any constraint $\D(c)$ satisfies
\[
\Phi^S_c = f_{i,c}\circ \Phi^S_i
\]
for any actor index $i$ such that $c$ is in $C_i$.
\end{proof}

\begin{definition}
    A functor $\F$ is \emph{span-extendable} if whenever the cospan $\rangle c_L, c_R\langle$ has an $\F$-pullback $\langle p_L, p_R\rangle$ and the cospan $\rangle c_L \circ f, c_R \circ g\langle$ has an $\F$-pullback $\langle q_L, q_R\rangle$, the unique span morphism $\Phi$ from $\langle \F(f\circ q_L), \F(g\circ q_R)\rangle$ to $\langle \F(p_L), \F(p_R)\rangle$ is $\F(\Psi)$, where \[\Phi\colon \F(\QA) \rightarrow \F(\PA) \quad \text{and} \quad \Psi\colon \QA \rightarrow \PA,\] $\PA$ is the apex of the $\F$-pullback $\langle p_L, p_R\rangle$, $\QA$ is the apex of the $\F$-pullback $\langle q_L, q_R\rangle$, and $\Psi$ is a span morphism. 
\end{definition}

This diagram captures the meaning of the functor $\F$ being \emph{span-extendable}:


\smallskip

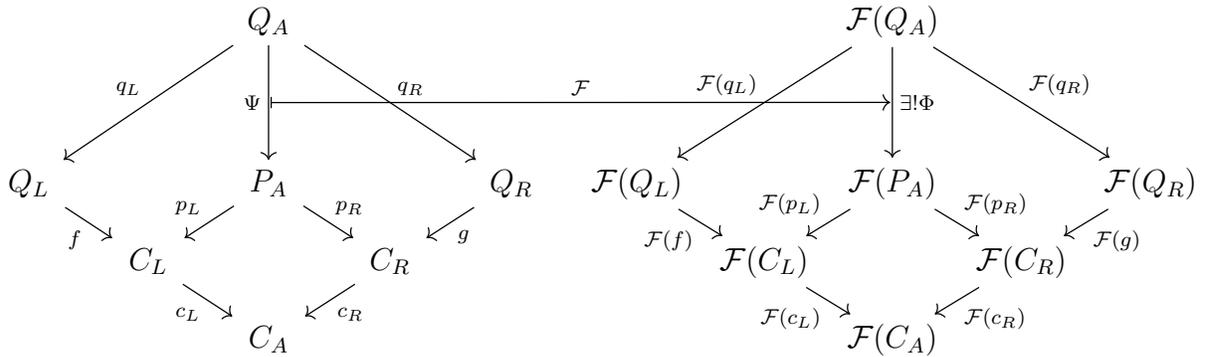
\begin{center}
\begin{tikzcd}[cramped,column sep=tiny,row sep=small]
	&& {\hspace{8pt} Q_A \hspace{8pt}} &&&&& {\F(Q_A)} \\
	\\
	\\
	\\
	{\hspace{8pt} Q_L \hspace{8pt}} && {\hspace{8pt} P_A \hspace{8pt}} && {\hspace{8pt} Q_R \hspace{8pt}} & {\F(Q_L)} && {\F(P_A)} && {\F(Q_R)} \\
	& {\hspace{8pt} C_L \hspace{8pt}} && {\hspace{8pt} C_R \hspace{8pt}} &&& {\F(C_L)} && {\F(C_R)} \\
	&& {\hspace{8pt} C_A \hspace{8pt}} &&&&& {\F(C_A)}
	\arrow["{q_L}"', from=1-3, to=5-1]
	\arrow[""{name=0, anchor=center, inner sep=0}, "\Psi"', from=1-3, to=5-3]
	\arrow["{q_R}", from=1-3, to=5-5]
	\arrow["{\F(q_L)}"', from=1-8, to=5-6]
	\arrow[""{name=1, anchor=center, inner sep=0}, "{\exists ! \Phi}", from=1-8, to=5-8]
	\arrow["{\F(q_R)}", from=1-8, to=5-10]
	\arrow["f"', from=5-1, to=6-2]
	\arrow["{p_L}"', from=5-3, to=6-2]
	\arrow["{p_R}", from=5-3, to=6-4]
	\arrow["g", from=5-5, to=6-4]
	\arrow["{\F(f)}"', from=5-6, to=6-7]
	\arrow["{\F(p_L)}"', from=5-8, to=6-7]
	\arrow["{\F(p_R)}", from=5-8, to=6-9]
	\arrow["{\F(g)}", from=5-10, to=6-9]
	\arrow["{c_L}"', from=6-2, to=7-3]
	\arrow["{c_R}", from=6-4, to=7-3]
	\arrow["{\F(c_L)}"', from=6-7, to=7-8]
	\arrow["{\F(c_R)}", from=6-9, to=7-8]
	\arrow["\F", rightarrow, maps to, from=0, to=1]
\end{tikzcd}
\end{center}

\smallskip

If $\C$ and $\C'$ both have terminal objects and $\F\colon \C \rightarrow \C'$ preserves terminal objects, then span-extendability of $\F$ implies Lemma~\ref{lem:productmorphisms} by taking $C_A$ to be a terminal object of $\C$.

\begin{lemma}\label{lem:productmorphisms} If $\F$ is span-extendable, $f\colon A \rightarrow B$ and $g\colon A' \rightarrow B'$ are morphisms in $\CatC$, and $\langle \pi_A, \pi_{A'} \rangle$ with apex $A \times A'$ and $\langle \pi_B, \pi_{B'} \rangle$ with apex $B \times B$ are $\F$-products, then there is a product morphism $f \times g\colon A \times A' \rightarrow B \times B'$ in $\CatC$.
\end{lemma}

Lemma \ref{lem:productmorphisms} establishes that a middle arrow in the following commutative diagram exists whenever $f$ and $g$ exist:

\smallskip

%
\begin{center}
    \begin{tikzcd}[cramped]
	A & {A \times A'} & {A'} \\
	B & {B \times B'} & {B'} \\
	& 1
	\arrow["f"', from=1-1, to=2-1]
	\arrow["{\pi_A}"', from=1-2, to=1-1]
	\arrow["{\pi_{A'}}", from=1-2, to=1-3]
	\arrow["{f \times g}"', from=1-2, to=2-2]
	\arrow["g", from=1-3, to=2-3]
	\arrow[from=2-1, to=3-2]
	\arrow["{\pi_B}"', from=2-2, to=2-1]
	\arrow["{\pi_{B'}}", from=2-2, to=2-3]
	\arrow[from=2-3, to=3-2]
\end{tikzcd}
\end{center}

\smallskip

While $\F(f) \times \F(g)$ is the unique morphism that satisfies an equivalent diagram to the above in $\CatC'$, the morphism $f \times g$ is not necessarily unique in $\CatC$ without additionally assuming that $\F$ is faithful.

\begin{lemma}\label{lem:projectionscancel}
Suppose that $\CatC$ and $\CatC'$ have terminal objects that $\F$ preserves.  Take $S$, $Q$, and $T$ to be objects in $\CatC$ with 
$\langle p_L^1,p_R^1\rangle$ an $\F$-product of $S$ and $Q$ with apex $S\times Q$,
$\langle p_L^2,p_R^2\rangle$ an $\F$-product of $S\times Q$ and $T$ with apex $(S\times Q)\times T$,
and $\langle p_L^3,p_R^3\rangle$ an $\F$-product of $Q$ and $T$ with apex $Q\times T$.
If $\langle p_L,p_R\rangle$ is an $\F$-pullback of the cospan $\rangle p_R^1,p_L^3\langle$
with apex $(S\times Q)\times_Q (Q\times T)$, then the following cones over the discrete diagram on $\{S,Q,T\}$ are isomorphic:
\[
(S\times Q)\times T
\quad\text{with legs}\quad
\big(p_L^1\circ p_L^2,\; p_R^1\circ p_L^2,\; p_R^2\big),
\]
\[
(S\times Q)\times_Q (Q\times T)
\quad\text{with legs}\quad
\big(p_L^1\circ p_L,\; p_R^1\circ p_L(=p_L^3\circ p_R),\; p_R^3\circ p_R\big).\]
\end{lemma}

The objects and morphisms involved in Lemma \ref{lem:projectionscancel} are arranged like this:

\smallskip

\begin{center}
\begin{tikzcd}[cramped, column sep=tiny]
	& \scalebox{0.9}{$(S \times Q) \times T$} &&& \scalebox{0.9}{$(S \times Q) \times_Q (Q \times T)$} \\
	\scalebox{0.9}{$S \times Q$} &&& \scalebox{0.9}{$S \times Q$} && \scalebox{0.9}{$Q \times T$} \\
	\scalebox{0.9}{$S$} & \scalebox{0.9}{$Q$} & \scalebox{0.9}{$T$} & \scalebox{0.9}{$S$} & \scalebox{0.9}{$Q$} & \scalebox{0.9}{$T$} \\
	\scalebox{0.9}{$S_R = Q_L \cong 1$} && \scalebox{0.9}{$Q_R = T_L \cong 1$} & \scalebox{0.9}{$S_R = Q_L \cong 1$} && \scalebox{0.9}{$Q_R = T_L \cong 1$}
	\arrow["\Phi", from=1-2, to=1-5]
	\arrow["{p^2_L}"', from=1-2, to=2-1]
	\arrow["{p^2_R}"{description}, from=1-2, to=3-3]
	\arrow["{p_L}"', from=1-5, to=2-4]
	\arrow["{p_R}", from=1-5, to=2-6]
	\arrow["{p_L^1}"', from=2-1, to=3-1]
	\arrow["{p_R^1}", from=2-1, to=3-2]
	\arrow["{p_L^1}"', from=2-4, to=3-4]
	\arrow["{p_R^1}", from=2-4, to=3-5]
	\arrow["{p_L^3}"', from=2-6, to=3-5]
	\arrow["{p_R^3}", from=2-6, to=3-6]
	\arrow["{s_R}", from=3-1, to=4-1]
	\arrow["{q_L}"', from=3-2, to=4-1]
	\arrow["{q_R}", from=3-2, to=4-3]
	\arrow["{t_L}"', from=3-3, to=4-3]
	\arrow["{s_R}", from=3-4, to=4-4]
	\arrow["{q_L}"', from=3-5, to=4-4]
	\arrow["{q_R}", from=3-5, to=4-6]
	\arrow["{t_L}"', from=3-6, to=4-6]
\end{tikzcd}
\end{center}

\smallskip

\begin{proof}

Under the standing hypothesis that $\CatC$ and $\CatC'$ have terminal objects that $\F$ preserves, prior work
\cite[Lemma~5.4]{WY} constructs a span isomorphism
\[
\Phi\colon (S\times Q)\times T \;\to\; (S\times Q)\times_Q (Q\times T)
\]
between the spans
\[
\langle p_L^1\circ p_L^2,\; p_R^2\rangle
\qquad\text{and}\qquad
\langle p_L^1\circ p_L,\; p_R^3\circ p_R\rangle.
\]
Moreover, $\Phi$ is constructed as a span morphism
\[
\langle p_L^2,\; p_R^2\rangle \;\to\; \langle p_L,\; p_R^3\circ p_R\rangle,
\]
so in particular
\begin{equation}\label{eq:Lem:Assoc}
p_L^2 = p_L\circ \Phi
\qquad\text{and}\qquad
p_R^2 = (p_R^3\circ p_R)\circ \Phi.
\end{equation}
It remains to show that $\Phi$ respects the $Q$-leg.  Use the left equality of \eqref{eq:Lem:Assoc} to obtain the equality
\begin{equation}\label{eq:Lem:AssocB}
p_R^1\circ p_L^2
= (p_R^1\circ p_L)\circ \Phi.
\end{equation}
Since $\langle p_L,p_R\rangle$ is an $\F$-pullback of the cospan $\rangle p_R^1,p_L^3\langle$,
\begin{equation}\label{eq:Lem:AssocC}
p_R^1\circ p_L = p_L^3\circ p_R.
\end{equation}
Equations \eqref{eq:Lem:AssocB} and \eqref{eq:Lem:AssocC} together imply that
\[
p_R^1\circ p_L^2
= (p_L^3\circ p_R)\circ \Phi,
\]
so $\Phi$ is a cone morphism over $\{S,Q,T\}$.
Since $\Phi$ is a span isomorphism  (hence an isomorphism in $\CatC$), it has an inverse
$\Phi^{-1}$, which by a similar argument is also a cone morphism, so the two cones are isomorphic.
\end{proof}

Apply Lemma~\ref{lem:projectionscancel} with the roles of $S$ and $T$ permuted to obtain a cone isomorphism
\[
(S\times Q)\times T \cong S\times (Q\times T)
\]
over the discrete diagram on $\{S,Q,T\}$.
Thus $\F$-products are associative up to cone isomorphism, so write $S\times Q\times T$.

\begin{remark}
The conclusion of Lemma~\ref{lem:projectionscancel} does not require terminal objects.  The assumptions that $\CatC$ and $\CatC'$ have terminal objects that $\F$ preserves allows the use of the prior work to provide an explicit span isomorphism
\[
\Phi\colon (S\times Q)\times T \;\to\; (S\times Q)\times_Q (Q\times T)
\]
that is already a cone isomorphism on the $S$- and $T$-legs. The proof above only checks that $\Phi$ also respects the $Q$-leg.
\end{remark}

If $\mathcal{A}$ and $\mathcal{B}$ are two sets of objects in $\C$, then Lemma~\ref{lem:projectionscancel} implies that there is a cone isomorphism \begin{equation}\label{eq:lem:projectionscancel}\prod_{X \in \mathcal{A} \cup \mathcal{B}} X \cong \prod_{A \in \mathcal{A}} A\times_{\prod_{Y \in \mathcal{A} \cap \mathcal{B}} Y} \prod_{B \in \mathcal{B}} B.\end{equation}

\begin{definition}
       A category $\C$ \emph{has nondegenerate} $\F$-pullbacks if it has $\F$-pullbacks, $\F$ preserves terminal objects, and if for any $\F$-pullback $P$ whose feet are not terminal objects, the apex of $P$ is also not a terminal object. 
\end{definition}

Proposition~\ref{prop:redtosimpthenFlimit} is a statement about the existence of $\F$-limits, namely, it is the formal statement that if a diagram is reducible to decomposable, then that diagram has an $\F$-limit. Proving Lemmata~\ref{lem:CanAlwaysWeldSimple} and \ref{lem:WeldedLimitsAreLimits} are critical steps in proving Proposition~\ref{prop:redtosimpthenFlimit}. For the remainder of this section, assume $\C$ has non-degenerate $\F$-pullbacks and $\F$ is span-extendable and cone-tight.

\begin{lemma}\label{lem:CanAlwaysWeldSimple}
    For any ACM-diagram $\D\colon \CatJ \to \C$ with distinct actor indices $i_0$ and $j_0$,
if $\D$ \simple, then it has a reduction $\D'\colon W_{i_0,j_0}(\CatJ)\to \C$ obtained by welding $A_{i_0}$ and $A_{j_0}$.
Moreover, $\D'$ is an ACM-diagram that \simple.
\end{lemma}

Preservation of the ACM-diagram structure under welding requires the additional assumption of non-degeneracy of $\F$-pullbacks, hence the introduction of the new definition.

\begin{proof}

Decompose unions and intersections to obtain the equalities
\begin{align*}
\Ext[i_0:\CatJ] \cap \Ext[j_0:\CatJ]
&= \Big(C_{i_0} \cap \bigcup_{j \neq i_0} C_j\Big) \cap \Big(C_{j_0} \cap \bigcup_{j \neq j_0} C_j\Big) \\ 
&= \Big(C_{i_0} \cap \Big(C_{j_0} \cup \bigcup_{j \neq i_0, j_0} C_j\Big)\Big) \cap \Big(C_{j_0} \cap \Big(C_{i_0} \cup \bigcup_{j \neq i_0, j_0} C_j\Big)\Big) \\ 
&= \Big(\Big(C_{i_0} \cap C_{j_0}\Big) \cup \Big(C_{i_0} \cap \bigcup_{j \neq i_0, j_0} C_j\Big)\Big) \\ & \hspace{2.2in} \cap \Big(\Big(C_{i_0} \cap C_{j_0}\Big) \cup \Big(C_{j_0} \cap \bigcup_{j \neq i_0, j_0} C_j\Big)\Big) \\ 
&= \Big(C_{i_0} \cap C_{j_0}\Big) \cup \Big(\Big(C_{i_0} \cap \bigcup_{j \neq i_0, j_0} C_j\Big) \cap \Big(C_{j_0} \cap \bigcup_{j \neq i_0, j_0} C_j\Big)\Big) \\ 
&= \Big(C_{i_0} \cap C_{j_0}\Big) \cup \Big(C_{i_0} \cap C_{j_0} \cap  \bigcup_{j \neq i_0, j_0} C_j\Big) = C_{i_0} \cap C_{j_0},
\end{align*}
hence
\begin{equation}\label{eq:extconprodA}
\Ext[i_0:\CatJ] \cap \Ext[j_0:\CatJ]= C_{i_0} \cap C_{j_0}.
\end{equation}

Similarly,
\begin{align}\label{eq:Lem:weldunioncon}
\Ext[i_0:\CatJ] \cup \Ext[j_0:\CatJ]
&= \Big(C_{i_0} \cap \bigcup_{j \neq i_0} C_j\Big) \cup \Big(C_{j_0} \cap \bigcup_{j \neq j_0} C_j\Big) \notag\\
&= \Big(C_{i_0} \cap C_{j_0}\Big) \cup \Big((C_{i_0}\cup C_{j_0})\cap \bigcup_{b\neq i_0,j_0} C_b\Big) \notag\\
&= (C_{i_0}\cap C_{j_0}) \cup \Ext[i_0j_0:\CatJ'].
\end{align}
Equation~\eqref{eq:lem:projectionscancel} together with \eqref{eq:extconprodA} and~\eqref{eq:Lem:weldunioncon} implies that there is an $\F$-product isomorphism
\begin{equation}\label{eq:welded-prod-iso}
\prod_{c\in\Ext[i_0:\CatJ]}\D(c) \times_{\prod_{c\in C_{i_0}\cap C_{j_0}}\D(c)}
\prod_{c\in\Ext[j_0:\CatJ]}\D(c)
\isoarrow
\prod_{c\in \Ext[i_0:\CatJ]\cup \Ext[j_0:\CatJ]}\D(c).
\end{equation}
\\
Define $A_{i_0j_0}$ to be the apex of the interaction $\KSpan(A_{i_0}, A_{j_0})$. Since $\F$ is span-extendable, there is an induced span morphism
\[
\psi\colon A_{i_0j_0}\to 
\prod_{c\in \Ext[i_0:\CatJ]}\D(c)\times_{\prod_{c\in C_{i_0}\cap C_{j_0}}\D(c)}
\prod_{c\in \Ext[j_0:\CatJ]}\D(c).
\]
Compose $\psi$ with the isomorphism in \eqref{eq:welded-prod-iso} to obtain an $\F$-product morphism
\[
\tilde f_{i_0j_0}\colon A_{i_0j_0}\to \prod_{c\in \Ext[i_0:\CatJ]\cup \Ext[j_0:\CatJ]}\D(c).
\]
\\
Denote by $\CatJ'$ the welded index category $ W_{i_0,j_0}(\CatJ)$. Equation~\eqref{eq:Lem:weldunioncon} implies that
\[
\Ext[i_0j_0:\CatJ']
=(C_{i_0}\cup C_{j_0})\cap \bigcup_{b\neq i_0,j_0} C_b
\subseteq \Ext[i_0:\CatJ]\cup \Ext[j_0:\CatJ],
\]
and associativity of $\F$-products as a consequence of Lemma~\ref{lem:projectionscancel} determines an $\F$-product morphism
\[
\pi_\cup\colon 
\prod_{c\in \Ext[i_0:\CatJ]\cup \Ext[j_0:\CatJ]}\D(c)
\to
\prod_{c\in \Ext[i_0j_0:\CatJ']}\D(c).
\]
Define
\begin{equation}\label{eq:MorphismWeldedActorExternalExists}
f_{i_0j_0}:=\pi_\cup\circ \tilde f_{i_0j_0}\colon 
A_{i_0j_0}\to \prod_{c\in \Ext[i_0j_0:\CatJ']}\D(c).
\end{equation}
This is the required $\F$-product morphism for the welded actor. The equality \[\Ext[b:\CatJ']=\Ext[b:\CatJ]\] together with the assumption that $\D$ decomposes external constraints determines an $\F$-product morphism
\begin{equation}\label{eq:Otheractorsdecomposeexternal}
f_b\colon A_b\to \prod_{c\in \Ext[b:\CatJ']}\D(c).
\end{equation}

Finally, construct the interaction between $A_{i_0j_0}$ and $A_b$.
Since $C_{i_0j_0}\cap C_b$ is a subset of $\Ext[i_0j_0:\CatJ']$, associativity of $\F$-products as a consequence of Lemma~\ref{lem:projectionscancel} determines an $\F$-product projection
\[
\prod_{c\in \Ext[i_0j_0:\CatJ']}\D(c)\to \prod_{c\in C_{i_0j_0}\cap C_b}\D(c),
\]
so $f_{i_0j_0}$ induces a morphism
\begin{equation}\label{eq:WeldedInteractionsExist1}
A_{i_0j_0}\to \prod_{c\in C_{i_0j_0}\cap C_b}\D(c).
\end{equation}
Similarly, since \[\Ext[b:\CatJ']=\Ext[b:\CatJ]\supseteq C_b\cap(C_{i_0}\cup C_{j_0})=C_b\cap C_{i_0j_0},\]
associativity of $\F$-products as a consequence of Lemma~\ref{lem:projectionscancel} determines an $\F$-product projection
\[
\prod_{c\in \Ext[b:\CatJ']}\D(c)\to \prod_{c\in C_{i_0j_0}\cap C_b}\D(c),
\]
so $f_b$ induces a morphism
\begin{equation}\label{eq:WeldedInteractionsExist2}
A_b\to \prod_{c\in C_{i_0j_0}\cap C_b}\D(c).
\end{equation}
The interaction $\KSpan(A_{i_0j_0},A_b)$ exists as an $\F$-pullback of the cospan determined by the morphisms \eqref{eq:WeldedInteractionsExist1} and \eqref{eq:WeldedInteractionsExist2}.

Define $\D'\colon \CatJ' \rightarrow \C$ to be a reduction of $\D$ by welding $A_{i_0}$ and $A_{j_0}$. Every arrow of $\D'$ is either an arrow of $\D$ or a composite of such arrows with projections from an $\F$-pullback, so $\D'$ respects identities and composition. Since every interaction $\KSpan(A_{i_0 j_0}, A_b)$ exists, $\D'$ is a well-defined functor.

Verify the ACM conditions to show that $\D'$ is an ACM-diagram.  Conditions (i) and (ii) of Definition~\ref{defACMDiag} hold for $\D'$ because they hold for $\D$ and the welding does not alter any part of the diagram away from the new object $i_0j_0$. Condition (v) holds by construction of the new interactions as $\F$-pullbacks. Since $\C$ has nondegenerate $\F$-pullbacks and $\D$ sends no object other than $\star$ to a terminal object, the new apices introduced by welding and by the new interactions are not terminal, so (iii) holds for $\D'$.
Condition (iv) follows from items~\eqref{eq:WeldedInteractionsExist1} and~\eqref{eq:WeldedInteractionsExist2}. Thus $\D'$ is an ACM-diagram.

The construction of the morphisms~\eqref{eq:MorphismWeldedActorExternalExists} and~\eqref{eq:Otheractorsdecomposeexternal} demonstrates that every actor of $\D'$ admits an $\F$-product morphism to the product of its external constraints, so $\D'$ decomposes external constraints.
   
\end{proof}

Commutativity of this diagram captures the data of Lemma \ref{lem:WeldedLimitsAreLimits}:

\smallskip

\begin{center}
\begin{tikzcd}
	&& S \\
	&& L \\
	&& {A_{ijk}} \\
	& {A_{ij}} && {A_{jk}} \\
	{A_i} && {A_j} && {A_k}
	\arrow["\Phi"', dashed, from=1-3, to=2-3]
	\arrow["{\Phi^S_{ijk}}", bend left=30, from=1-3, to=3-3]
	\arrow["{\Phi^S_{ij}}"', bend right=20, from=1-3, to=4-2]
	\arrow["{\Phi^S_k}", bend left=30, from=1-3, to=5-5]
	\arrow["{\Phi^L_{ijk}}"', from=2-3, to=3-3]
	\arrow["{\rho_{ij}}"',{pos=0.3}, from=3-3, to=4-2]
	\arrow["{\phi_{jk}}", dashed, from=3-3, to=4-4]
	\arrow["{\rho_k}"{pos=0.3}, bend left=30, from=3-3, to=5-5]
	\arrow[from=4-2, to=5-1]
	\arrow[from=4-2, to=5-3]
	\arrow[from=4-4, to=5-3]
	\arrow[from=4-4, to=5-5]
\end{tikzcd}
\end{center}

\smallskip

\begin{lemma}\label{lem:WeldedLimitsAreLimits}
Take $\D$ to be an ACM-diagram that \simple\ of shape $\CatJ$ in $\C$ with actors $A_i$ and $A_j$ and take $\D'$ to be a reduction of $\D$ by welding $A_i$ and $A_j$. If $\D'$ has an $\F$-limit $\Phi^L$, then $\D$ has an $\F$-limit $\Psi^L$ with identical apex $L$ so that for any $x$ in both $\CatJ$ and $W_{i,j}(\CatJ)$, 
\[
\Psi^L_x = \Phi^L_x.
\] 
Furthermore, \[\Psi^L_i = \rho_i\circ \Phi^L_{ij} \quad \text{and} \quad \Psi^L_j = \rho_j \circ \Phi^L_{ij} \] where $\langle \rho_i, \rho_j \rangle$ is the span $\KSpan(A_i, A_j)$.
\end{lemma}

\begin{proof}

Take $\langle \rho_i,\rho_j\rangle$ to be the span $\KSpan(A_i,A_j)$ with apex $A_{ij}$. For each actor index $k$ that is not in $\{i,j\}$ and each $a$ in $\{i,j\}$, take
\[
\langle \rho_{ij},\rho_k\rangle := \KSpan(A_{ij},A_k)
\qquad\text{and}\qquad
\langle \rho_a,\rho_k\rangle := \KSpan(A_a,A_k),
\]
with respective apices $A_{ijk}$ and $A_{ak}$. The morphism $\rho_a\colon A_{ij}\to A_a$ and the identity on $A_k$ determine a morphism of cospans from the shared-constraint cospan of $(A_{ij},A_k)$ to the shared-constraint cospan of $(A_a,A_k)$. Since $\F$ is span-extendable, this morphism of cospans induces a unique span morphism
\[
\phi_{ak}\colon A_{ijk}\to A_{ak}
\]
satisfying commutativity with the pullback projections.

Construct a cone $\Psi^L$ over $\D$ with apex $L$ as follows.
\begin{enumerate}
\item If $x$ is an object of $\CatJ$ that is also an object of $W_{i,j}(\CatJ)$, set \[\Psi^L_x:=\Phi^L_x.\]
\item Define the actor legs
\[
\Psi^L_i:=\rho_i\circ \Phi^L_{ij}
\quad \text{and} \quad
\Psi^L_j:=\rho_j\circ \Phi^L_{ij}.
\]
\item For each $k$ not in $\{i,j\}$ and each $a$ in $\{i,j\}$, define the interaction leg
\[
\Psi^L_{ak}:=\phi_{ak}\circ \Phi^L_{ijk}\colon L\to A_{ak}.
\]
\end{enumerate}

Denote by $\CatJ'$ the category $W_{i,j}(\CatJ)$. For any cone $\Psi^S$ over $\F\circ \D$ with apex $S$, construct a cone $\Phi^S$ over $\F\circ \D'$ with the same apex $S$ as follows.
\begin{enumerate}
\item If $x$ is an object of $\CatJ'$ that is also an object of $\CatJ$, set \[\Phi^S_x := \Psi^S_x.\]
\item Since $\F(\langle \rho_i,\rho_j\rangle)$ is a pullback span and $(\Psi^S_i,\Psi^S_j)$ is paired with the shared-constraint cospan of $(A_i,A_j)$, the pullback universal property gives a unique morphism
\[
\Phi^S_{ij}\colon S\to \F(A_{ij})\quad \text{such that} \quad \F(\rho_i)\circ \Phi^S_{ij}=\Psi^S_i \quad \text{and} \quad \F(\rho_j)\circ \Phi^S_{ij}=\Psi^S_j.\]
\item For each $k$ that is not in $\{i,j\}$, the pair $(\Phi^S_{ij},\Psi^S_k)$ is paired with the shared-constraint cospan of $(A_{ij},A_k)$, hence the pullback universal property of $\F(\langle \rho_{ij},\rho_k\rangle)$ gives a unique morphism
\[
\Phi^S_{ijk}\colon S\to \F(A_{ijk})
\]
whose composites with the two pullback projections are $\Phi^S_{ij}$ and $\Psi^S_k$, respectively.
\end{enumerate}

To check that $\Phi^S$ is a cone over $\F\circ \D'$, take $\alpha$ to be any morphism $\alpha\colon x\to y$ in $\CatJ'$. The cone condition is the equality
\[
\F(\D'(\alpha))\circ \Phi^S_x=\Phi^S_y.
\]
If $\alpha$ lies in the part of $\CatJ'$ inherited from $\CatJ$ away from the welded objects, then
\[
\D'(\alpha)=\D(\alpha), \quad \Phi^S_x=\Psi^S_x, \quad \text{and} \quad \Phi^S_y=\Psi^S_y,
\] 
so the equality holds because $\Psi^S$ is a cone over $\F\circ \D$. If $\alpha$ is one of the projection morphisms out of a pullback span defining $A_{ij}$ or $A_{ijk}$, then the equality holds by the defining identities for $\Phi^S_{ij}$ and $\Phi^S_{ijk}$ obtained from the corresponding pullback universal properties. Therefore $\Phi^S$ is a cone over $\F\circ \D'$.

Since $\F(\Phi^L)$ is a limit cone of $\F\circ \D'$, there is a unique cone morphism
\[
\Phi\colon S\to \F(L)
\]
from $\Phi^S$ to $\F(\Phi^L)$. The definitions of $\Psi^L$ in terms of $\Phi^L$ imply that the equalities
\[
\Phi^S_x=\F(\Phi^L_x)\circ \Phi
\] 
for all objects $x$ of $\CatJ'$ determine
\[
\Psi^S_y=\F(\Psi^L_y)\circ \Phi
\]
for every object $y$ of $\CatJ$. Hence $\Phi$ is a cone morphism $\Psi^S\to \F(\Psi^L)$ over $\F\circ \D$.

If $\Phi'\colon S\to \F(L)$ is any cone morphism $\Psi^S\to \F(\Psi^L)$ over $\F\circ \D$, then the same verification on the generating morphisms of $\CatJ'$ shows that $\Phi'$ is also a cone morphism $\Phi^S\to \F(\Phi^L)$ over $\F\circ \D'$. Uniqueness of $\Phi$ over $\F\circ \D'$ implies $\Phi'$ is equal to $\Phi$. Therefore $\F(\Psi^L)$ is a limit of $\F\circ \D$, and so $\Psi^L$ is an $\F$-limit of $\D$.
 
\end{proof}

\begin{proposition}\label{prop:redtosimpthenFlimit}
    If $\D$ is a reducible to decomposable ACM-diagram, then there is a chain of reductions of $\D$ to a diagram $\D'$ where $\D'$ has one actor. Furthermore, any ACM-diagram that is reducible to a decomposable ACM-diagram admits an $\F$-limit in $\CatC$ which has an apex that isomorphic to the actor in $\D'$.
\end{proposition}

\begin{proof}
For any ACM-diagram $\mathcal H$ that \simple\ with exactly one actor index $a$ and corresponding actor $A$, Lemma~\ref{lem:actorsonly} implies that cones over $\F\circ \mathcal H$ with apex $S$ are in bijection with morphisms $S\to \F(A)$. Therefore the identity $\id_{\F(A)}\colon \F(A)\to \F(A)$ determines a limit of $\F\circ \mathcal H$, hence $\mathcal H$ admits an $\F$-limit with apex $A$.

Take $n$ to be any natural number. Assume that every ACM-diagram that decomposes external constraints with $n$ actor indices admits an $\F$-limit. Take any ACM-diagram $\mathcal G_{n+1}$ that decomposes external constraints with $n+1$ actor indices. Choose distinct actor indices $i$ and $j$ of $\mathcal G_{n+1}$ and weld them to obtain a reduction
\[
\mathcal G_{n+1} \to \mathcal G_n
\]
with $n$ actor indices. Lemma~\ref{lem:CanAlwaysWeldSimple} implies that $\mathcal G_n$ decomposes external constraints, so the inductive hypothesis gives an $\F$-limit of $\mathcal G_n$. Lemma~\ref{lem:WeldedLimitsAreLimits} then gives an $\F$-limit of $\mathcal G_{n+1}$ with the same apex. This completes the induction, so every ACM-diagram that decomposes external constraints admits an $\F$-limit.

For any ACM-diagram $\D$ that is reducible to decomposable, fix a chain of decomposing reductions
\[
\D=\D_0 \;\to\; \D_1 \;\to\; \cdots \;\to\; \D_m=: \mathcal E_N,
\]
where $\mathcal E_N$ \simple\ and has $N$ actor indices. Repeatedly weld actors to obtain a chain of decomposing reductions
\[
\mathcal E_N \;\to\; \mathcal E_{N-1} \;\to\; \cdots \;\to\; \mathcal E_1,
\]
where $\mathcal E_1$ has exactly one actor index and corresponding actor $L$.

Since $\mathcal E_1$ has only one actor index, it admits an $\F$-limit with apex $L$. The chain of weldings
\[
\mathcal E_N \to \mathcal E_{N-1} \to \cdots \to \mathcal E_1
\]
is finite, and Lemma~\ref{lem:WeldedLimitsAreLimits} transports an $\F$-limit along the chain in the direction opposite to the arrows. Therefore $\mathcal E_N$ admits an $\F$-limit with apex $L$. The same reasoning applied to the chain
\[
\D=\D_0 \to \D_1 \to \cdots \to \D_m=\mathcal E_N
\]
yields an $\F$-limit of $\D$ with apex $L$.

\end{proof}

Proposition~\ref{prop:prodUnionHasLimit} provides the technical foundation for Theorem~\ref{thm:UnionIsSystem}, which shows how the framework of ACM categories permits the construction of larger systems from subsystems. In general, the union of ACM-diagrams may fail to exist because a larger collection of actors could potentially involve additional constraints for each actor, and require mapping actors into too many constraints. Even if a union of two ACM-diagrams exists, the question remains: when does that ACM-diagram have an $\F$-limit? Both $\D_1$ and $\D_2$ decomposing external constraints is not enough, as Example~\ref{nonexample:Flimitb} shows.

\begin{proposition}\label{prop:prodUnionHasLimit}
Suppose that $\C$ has $\F$-pullbacks and take $\F$ to be a cone-tight, faithful, span-extendable functor that preserves terminal objects. If $\D_1$ and $\D_2$ are ACM subdiagrams of an ACM-diagram $\D$ that decomposes into constraints so that their intersection $\D_{12}$ is either:
\begin{enumerate}
	\item A diagram with no actors, or
        \item An ACM subdiagram of $\D_1$ and $\D_2$ such that for every actor $a$ in $\D_{12}$, \[\iota_1(C_a) = C_{\iota_1(a)} \quad \text{and} \quad \iota_2(C_a) = C_{\iota_2(a)},\]
\end{enumerate} and if $\D$ is their union over $\D_{12}$, then there are $\F$-limits $\Phi_1^{X_1}$, $\Phi_2^{X_2}$, and $\Phi_{12}^{X_{12}}$ of $\D_1$, $\D_2$, and $\D_{12}$ respectively. Furthermore, there are $\F$-product morphisms from $X_1$ to $X_{12}$ and from $X_2$ to $X_{12}$ with an $\F$-pullback which defines an $\F$-limit of $\D$. 
\end{proposition}

\begin{proof}
Lemma~\ref{lem:projectionscancel} implies that any ACM-diagram $\D'\colon\CatJ'\to\CatC$ that decomposes into constraints has an $\F$-limit that is cone isomorphic to the $\F$-product of its actors.  For each $i$ in $\{1,2\}$, take $\Phi_i^{X_i}$ to be the $\F$-limit of $\D_i$ with apex
\[
X_i=\prod_{c\in\Ob_C(\CatJ_i)}\D(c),
\]
whose legs are the unique morphisms that $\F$ takes to the canonical projections to the actors and constraint of $\F\circ\D_i$, together with the induced legs to the interactions of $\F \circ \D_i$.  Since each interaction is an $\F$-pullback and $\F$ is span-extendable, the projections to its actors determine a unique leg from $X_i$.

\smallskip

In case~(1), the diagram $\D_{12}$ has no actors, so it is discrete and admits an
$\F$-limit, namely the $\F$-product
\[
X_{12}=\prod_{c\in\Ob_C(\CatJ_{12})}\D(c),
\]
with the unique morphisms in $\CatC$ that $\F$ takes to the canonical projections. Henceforth abuse terminology by referring to morphisms appearing in this way as canonical projections.

In case~(2), the diagram $\D_{12}$ is an ACM subdiagram of $\D_1$ and $\D_2$ and the hypotheses
\[
\iota_1(C_a)=C_{\iota_1(a)}
\quad\text{and}\quad
\iota_2(C_a)=C_{\iota_2(a)}
\]
ensure that every actor of $\D_{12}$ is isomorphic to the same $\F$-product of constraints as in $\D_1$ and $\D_2$.  Therefore $\D_{12}$ \elementary.  Take $\Phi_{12}^{X_{12}}$ as an $\F$-limit of $\D_{12}$ with apex the $\F$-product
\[
X_{12}=\prod_{c\in\Ob_C(\CatJ_{12})}\D(c),
\]
with the canonical projections to its actor and constraint objects and induced legs to its interaction objects.

Since the construction of an $\F$-limit $\Phi^{X_{12}}_{12}$ of $\D_{12}$ in case (1) is identical to the construction in case (2), proceed in either case as follows.

The containment 
\[
\Ob_C(\CatJ_{12})\subseteq \Ob_C(\CatJ_i)
\] 
together with associativity of $\F$-products as a consequence of Lemma~\ref{lem:projectionscancel} determine $\F$-product projections
\[
\Psi_1\colon X_1\to X_{12} \quad \text{and} \quad \Psi_2\colon X_2\to X_{12}.
\]
Take $\langle\rho_1,\rho_2\rangle$ to be an $\F$-pullback of the cospan $\rangle\Psi_1,\Psi_2\langle$ with apex $X$.

Define a cone $\Phi^X$ over $\D$ on the actors, constraints, and interactions of either $\D_1$ or $\D_2$ as follows. For any object $x$ of $\CatJ$, if $x$ lies in $\CatJ_i$, then define
\[
\Phi^{X}_x:=\Phi^{X_i}_x\circ\rho_i\colon X\to \D(x).
\]
Take $y$ to be any object in $\CatJ_{12}$. Since $\langle\rho_1,\rho_2\rangle$ is a pullback,
\[
\Psi_1\circ\rho_1=\Psi_2\circ\rho_2.
\]
Compose each of these morphisms with the leg $X_{12}\to \D(y)$ to obtain
\[
\Phi^{X_{12}}_y\circ \Psi_1 \circ \rho_1=\Phi^{X_{12}}_y\circ \Psi_2 \circ \rho_2,
\]
Each $\Psi_i$ for $i$ in $\{1,2\}$ is an $\F$-product projection, so each are cone morphisms over $\D_{12}$. Thus,
\[
\Phi^{X_1}_y = \Phi^{X_{12}}_y\circ \Psi_1 \quad \text{and} \quad \Phi^{X_2}_y = \Phi^{X_{12}}_y\circ \Psi_2
\]
and so
\[
\Phi^{X_1}_y\circ\rho_1=\Phi^{X_2}_y\circ\rho_2,
\]
which shows that $\Phi^X_y$ does not depend on the choice of $i$.

Take $a$ and $b$ to be actor indices in $\CatJ$. Since $\D$ is an ACM-diagram, $\star$ is in $C_a$ for every actor index $a$ in $\CatJ$, hence $\star$ is in $C_{a}\cap C_{b}$.

If $\{a,b\}$ lies in $\CatJ_1$, define 
\[
\Phi_{\{a,b\}}:=\Phi^{X_1}_{\{a,b\}}\circ\rho_1\colon X\to \D(\{a,b\}).
\]
If $\{a,b\}$ lies in $\CatJ_2$, define 
\[
\Phi_{\{a,b\}}:=\Phi^{X_2}_{\{a,b\}}\circ\rho_2\colon X\to \D({\{a,b\}}).
\] 
Otherwise, one of $a$ and $b$ lies in $\CatJ_1$ and the other lies in $\CatJ_2$.  In this case, define $\Phi_{\{a,b\}}$ directly from the already-defined legs to the two actors and to the shared constraints. Lemma~\ref{lem:projectionscancel} together with $\D$ decomposing into constraints implies that $\D(\{a,b\})$ is an $\F$-product $\prod_{c \in C_a \cup C_b} \D(c)$. Since $C_a \cup C_b$ is a subset of $\text{Ob}_C(\CatJ)$, associativity of $\F$-products as a consequence of Lemma~\ref{lem:projectionscancel} implies there is an $\F$-product projection 
\[
\Phi^X_{\{a,b\}}\colon X\to \D(\{a,b\})
\]
such that \[p_a\circ \Phi^{X_i}_{\{a,b\}}=\Phi^X_{\{a,b\}} \quad \text{and} \quad p_b\circ \Phi^X_{\{a,b\}}=\Phi^{X_j}_{b},\] where $a$ is an object of $\CatJ_i$ and $b$ is an object of $\CatJ_j$. Since $\Phi^X_{\{a,b\}}$ is an $\F$-product projection, it is a cone morphism over the shared constraints $\D(c)$ for constraint indices $c$ in $C_a \cup C_b$.

To demonstrate the universal property of the limit after applying $\F$, take any cone $\Phi^S$ over $\F\circ\D$ with apex $S$.  Restrict $\Phi^S$ to
$\F\circ\D_i$.  The cone $\F(\Phi_i^{X_i})$ is a limit of $\F\circ\D_i$, so there are unique cone morphisms
\[
\alpha_i\colon S\to \F(X_i)\quad (i\in\{1,2\}).
\]
Restrict $\Phi^S$ to $\F\circ\D_{12}$.  The cone $\F(\Phi_{12}^{X_{12}})$ is a limit of $\F\circ\D_{12}$, so there is a unique cone morphism
\[
\alpha_{12}\colon S\to \F(X_{12}).
\]
Since each $\Psi_i$ is a cone morphism $\Phi_i^{X_i}\to \Phi_{12}^{X_{12}}$, functoriality gives
\[
\F(\Psi_i)\circ \alpha_i=\alpha_{12}\quad (i\in\{1,2\}),\quad \text{hence} \quad \F(\Psi_1)\circ\alpha_1=\F(\Psi_2)\circ\alpha_2.
\]
Since $\F(\langle\rho_1,\rho_2\rangle)$ is a pullback of the cospan $\rangle \F(\Psi_1),\F(\Psi_2)\langle$, there is a unique morphism
\[
\alpha\colon S\to \F(X)
\]
such that \[\F(\rho_i)\circ\alpha=\alpha_i.\]

To verify that $\alpha$ is a cone morphism $S\to \F(X)$, take $x$ to be any object of $\CatJ$.  If $x$ lies in $\CatJ_1$, then
\[
\F(\Phi^X_x)\circ \alpha
=\F(\Phi^{X_1}_x\circ\rho_1)\circ\alpha
=\F(\Phi^{X_1}_x)\circ \F(\rho_1)\circ\alpha
=\F(\Phi^{X_1}_x)\circ \alpha_1,
\]
which is equal to the $x$-leg of $\Phi^S$ since $\alpha_1$ is a cone morphism.  A similar calculation holds for $x$ in $\CatJ_2$.  For interaction objects, the equality follows from functoriality and the defining equalities of the legs $\Phi_{\{a,b\}}$ through the corresponding $\F$-pullbacks. Therefore, the morphism $\alpha$ defines a cone morphism.

To prove uniqueness, take any $\alpha'\colon S\to\F(X)$ that defines a cone morphism.  Each $\F(\rho_i)\circ \alpha'$ defines a cone morphism $S\to\F(X_i)$, hence equals $\alpha_i$ by uniqueness of the factorization through the limit $\F(\Phi_i^{X_i})$. The pullback universal property implies
\[
\alpha'=\alpha.
\]

Therefore $\F(\Phi^X)$ is a limit of $\F\circ\D$, so $\Phi^X$ is an $\F$-limit of $\D$.
\end{proof}

\begin{lemma}\label{lem:SupersimpleSubdiags}
    If $\D_1\colon \CatJ_1 \rightarrow \C$ \simple\ ACM-diagram and $\D_2\colon \CatJ_2 \rightarrow \C$ is an ACM subdiagram of $\D_1$, then $\D_2$ \simple. 
\end{lemma}

\begin{proof}
Take $a$ to be any actor index of $\CatJ_2$ and $\iota$ to be the inclusion for $\D_2$. Since $\D_2$ is an ACM subdiagram of $\D_1$, for any actor index $i$, 
\begin{equation}\label{eq:subdiagsubconstraints}
    \iota(C_i) \subseteq C_{\iota(i)} \quad \text{and consequently} \quad \iota(\text{Ob}_A(\CatJ_2) )\subseteq \text{Ob}_A(\CatJ_1).
\end{equation}
Injectivity of $\iota$ on objects implies that  
\[
\iota(\Ext[a:\CatJ_2]) = \iota(C_a \cap C_a^\perp) = \iota(C_a) \cap \iota(C_a^\perp).
\]
Write $\iota(C_a^\perp)$ as a union to obtain the equalities
\[
\iota(\Ext[a:\CatJ_2]) = \iota(C_a) \cap \iota\Big(\bigcup_{i \in \text{Ob}_A(\CatJ_2) \setminus \{a\}} C_i\Big) = \iota(C_a) \cap \bigcup_{i \in \text{Ob}_A(\CatJ_2) \setminus \{a\}} \iota(C_i).
\]
The containments~\eqref{eq:subdiagsubconstraints} therefore imply that
\begin{align*}
    \iota(\Ext[a:\CatJ_2])&\subseteq C_{\iota(a)} \cap \bigcup_{i \in \text{Ob}_A(\CatJ_2) \setminus \{a\}} C_{\iota(i)}\\
    &\subseteq C_{\iota(a)} \cap \bigcup_{i \in \text{Ob}_A(\CatJ_1) \setminus \{\iota(a)\}} C_{i} = \Ext[\iota(a): \CatJ_1].
\end{align*}

Associativity of $\F$-products as a consequence of Lemma~\ref{lem:projectionscancel} determines an $\F$-product morphism 
\[
\pi\colon \prod_{c \in \Ext[\iota(a): \CatJ_1]} \D_1(c) \rightarrow \prod_{c \in \iota(\Ext[a:\CatJ_2])} \D_1(\iota(c)). 
\] 
Since $\D_2$ is an ACM subdiagram of $\D_1$, 
\[
\prod_{c \in \iota(C_a)} \D_1(\iota(c)) = \prod_{c \in C_a} \D_2(c) \quad \text{and} \quad \D_1(\iota(a)) = \D_2(a).
\] 

The diagram $\D_1$ \simple, so there is an $\F$-product morphism 
\[
f_a\colon \D_1(\iota(a)) \rightarrow \prod_{c \in \Ext[\iota(a):\CatJ_1]} \D_1(c),
\]
and so $\pi \circ f_a$ is an $\F$-product morphism from $\D_2(a)$ to $\prod_{c \in \Ext[a:\CatJ_2]} \D_2(c)$.

\end{proof}


\section{The rigid inclusion category}\label{Sec4}


Different observers of a system may choose different models for the same physical system. The source of an $\F$-limit of an ACM-diagram is the configuration space of the corresponding ACM-system, and non-uniqueness of $\F$-limits reflects that observers may obtain different, but compatible, models of the total system. Cone isomorphisms express compatibility at the level of configuration space. However, configuration-space compatibility alone does not determine the compositional structure. Proposition~\ref{lem:IsoDiagIsoDiagX} shows that natural isomorphisms of ACM-diagrams induce natural isomorphisms of the corresponding ACM-systems, so the compositional structure depends only on the natural isomorphism class of the diagram.

\begin{proposition}\label{lem:IsoDiagIsoDiagX}
For any naturally isomorphic diagrams $\D_1$ and $\D_2$ in $\C$ with $\F$-limits $\Phi^{X_1}$ and $\Phi^{X_2}$, respectively, if $\F$ is cone-tight, then $\D_1^{X_1}$ and $\D_2^{X_2}$ are naturally isomorphic.
\end{proposition}

\begin{proof}
Take $\CatJ$ to be the shape of both $\D_1$ and $\D_2$.  For each object $i$ of $\CatJ$, denote by $\Phi^1_i$ and $\Phi^2_i$ the component morphisms \[\Phi^1_i\colon X_1\to \D_1(i) \quad \text{and} \quad \Phi^2_i\colon X_2\to \D_2(i).\] For $j$ in $\{1,2\}$, since $\Phi^{X_j}$ is an $\F$-limit, $\F(\Phi^{X_j})$ is a limiting cone of $\F\circ\D_j$.  

For any natural isomorphism $\eta$ from $\D_1$ to $\D_2$, $\F(\eta)$ is a natural isomorphism from $\F\circ\D_1$ to $\F\circ\D_2$. For each $i$ in $\CatJ$, denote by $\eta_i$ the component of $\eta$ that takes $\D_1(i)$ to $\D_2(i)$.  Since $\F(\eta)\circ \F(\Phi^{X_1})$ is a cone over $\F\circ\D_2$, the universal property of $\F(\Phi^{X_2})$ guarantees that there is a unique morphism $u^\prime$ so that for all $i$ in $\CatJ$, 
\begin{equation}\label{lem:isolimdiag}
u^\prime\colon\F(X_1)\to \F(X_2)\quad\text{with}\quad \F(\Phi^2_i)\circ u^\prime=\F(\eta_i)\circ \F(\Phi^1_i).
\end{equation}
Repeating the procedure with $\eta^{-1}$ instead yields a unique morphism \[v^\prime\colon\F(X_2)\to\F(X_1),\] and uniqueness implies that \[u^\prime\circ v^\prime=\id \quad\text{and}\quad v^\prime \circ u^\prime=\id,\] so $u^\prime$ is an isomorphism.

Cone-tightness guarantees that $u^\prime$ lifts to a cone isomorphism $u$ from $X_1$ to $X_2$ over $\D_1$ in $\C$, that is, $\F(u)$ is equal to $u^\prime$ and for any $i$ in $\CatJ$, 
\[
\mu^2_i\circ u=\eta_i\circ \mu^1_i.
\]

Define a natural isomorphism \[\eta'\colon\D_1^{X_1}\to \D_2^{X_2}\] by its components
\[
\eta_i':=\eta_i\quad \text{and} \quad \eta_x:=u
\]%
where $x$ is the cone apex index in $\CatJ^X$, the index category for $\D_1^{X_1}$ and $\D_2^{X_2}$, depicted here:

\smallskip

\begin{center}
\begin{tikzcd}
	{X_1} && {X_2} \\
	\\
	{\D_1(i)} && {\D_2(i)}
	\arrow["{\eta'_x = u}", from=1-1, to=1-3]
	\arrow["{\Phi^1_i}"', from=1-1, to=3-1]
	\arrow["{\Phi^2_i}", from=1-3, to=3-3]
	\arrow["{\eta'_i = \eta_i}"', from=3-1, to=3-3]
\end{tikzcd}
\end{center}

\smallskip

For each leg $x \to i$ in $\CatJ^X$,
\begin{align}\label{lem:nat}
\D_2^{X_2}(x\to i)\circ \eta_x' & =\Phi^2_i\circ u \notag\\ & =\eta_i \circ \Phi^1_i\notag\\ & =\eta_i'\circ \D_1^{X_1}(x\to i).
\end{align}
Equation~\eqref{lem:nat} and the naturality of the restriction of $\eta$ to the objects and morphisms of $\D_1$ together imply that $\eta'$ is a natural transformation. Each component morphism of $\eta$ and $u$ are each isomorphisms, so $\eta'$ is a natural isomorphism. Thus, $\D_1^{X_1}$ and $\D_2^{X_2}$ are naturally isomorphic.
\end{proof}

\begin{definition}\label{def:ACMDX}
Given any ACM-diagram $\D$ of shape $\CatJ$, take $[\D]$ to be the natural isomorphism class of diagrams of shape $\CatJ$ in $\CatC$ that contains $\D$.  The \emph{actor-constraint mediated system} for $[\D]$ is $[\D^X]$, where $\Phi^X$ is an $\F$-limit over $\D$.
\end{definition}

Proposition~\ref{lem:IsoDiagIsoDiagX} guarantees that $[\D^X]$ depends only on $[\D]$.  For any two representatives $\D_1$ and $\D_2$ of $[\D]$, and any $\F$-limits $\Phi^{X_1}$ and $\Phi^{X_2}$ of $\D_1$ and $\D_2$, respectively, \[[\D_1^{X_1}] = [\D_2^{X_2}],\] making the association between $[\D]$ and $[\D^X]$ in Definition~\ref{def:ACMDX} well-defined.  

The phrase \emph{open system} should conjure the image of a system that interacts with a larger, unseen world. Sometimes that larger world supplies additional actors; other times it imposes extra geometric constraints on the same actors. In either case, an observer who only sees the subsystem may detect fewer degrees of freedom than a naive count predicts, indicating something hidden with observable effect. 

As a motivating example, consider two point particles $A_1$ and $A_2$ moving freely in the plane. Example~\ref{ex:RigidBar} revisits this system in a more precise way using the ACM framework, but it is helpful to begin with a few informal observations. For now, take the configuration manifold to be  $\mathds R^2 \times \mathds R^2$; a path in this space is simply a pair of particle paths. If a rigid bar of fixed length $d$ joins the particles, the admissible configurations are restricted to
\[
\big\{(x_1,x_2)\in\mathds R^2\times\mathds R^2 \mid \|x_1-x_2\|=d\big\}.
\]
An observer initially unaware of the constraint would discover that one relative degree of freedom is fixed. What appears to be an uncoupled two--body system is in fact \emph{open}: its apparent motions factor through a hidden constraint supplied by the environment. Inclusion of one system within another must therefore account for both additional constraints and additional actors. As this example suggests, any ACM-system should, up to isomorphism, appear as a subset of the $\F$-product of its actors. The goal is to realize this subset as an $\F$-limit.

The present framework studies open kinematic systems compositionally. It builds configuration spaces from elementary actor--constraint components, composing them so that admissible paths arise exactly from actor paths that agree wherever they meet a shared constraint. A forthcoming work develops the dynamical analogue, in which unseen forces arise from both unseen interactions and unseen geometric constraints.

Lemma~\ref{lem:natisosacrossinclusions} is critical for formalizing the above notion of subsystem inclusions as morphisms. It shows that natural isomorphisms, and hence natural isomorphism classes, restrict and extend along inclusion functors. In particular, it shows that openness of a system is independent of the model of that system.

\begin{lemma}\label{lem:natisosacrossinclusions}
 Take $\iota^*\colon \D_1 \rightarrow \D_2$ to be an ACM subdiagram determined by an inclusion $\iota\colon \CatJ_1 \rightarrow \CatJ_2$. If $\beta^{2,a}\colon \D_2 \Rightarrow \D_{2,a}$ is a natural isomorphism, then there is an ACM-diagram $\D_{1,a}$ defined by 
 \[
 \D_{1,a} := \D_{2,a} \circ \iota
 \] 
 and a collection of isomorphisms given by \[\beta^{1,a}_x := \beta^{2,a}_{\iota(x)}\] for any $x$ in $\CatJ_1$ which form a natural isomorphism $\beta^{1,a}\colon \D_1 \Rightarrow \D_{1,a}$. If $\beta^{1,b}\colon \D_1 \Rightarrow \D_{1,b}$ is a natural isomorphism, then there is an ACM-diagram $\D_{2,b}$ and there is a natural isomorphism $\beta^{2,b}\colon\D_2 \Rightarrow \D_{2,b}$ such that 
 \[
 \D_{1,b} = \D_{2,b} \circ \iota \quad \text{and} \quad \beta^{2,b}_{\iota(x)} = \beta^{1,b}_x
 \] 
 for any $x$ in $\CatJ_1$. 
 \end{lemma}
    
For any $i$ is in $\{a, b\}$, Lemma \ref{lem:natisosacrossinclusions} establishes that the following diagrams commute:

\smallskip

\begin{center}
\begin{tikzcd}
	{\CatJ_1} && {\CatJ_2} \\
	\\
	& \CatC
	\arrow["\iota", from=1-1, to=1-3]
	\arrow[""{name=0, anchor=center, inner sep=0}, "{\D_1}"', shift right=3, bend left=-30, from=1-1, to=3-2]
	\arrow[""{name=1, anchor=center, inner sep=0}, "{\D_{1,i}}", from=1-1, to=3-2]
	\arrow[""{name=2, anchor=center, inner sep=0}, "{\D_{2,i}}"', from=1-3, to=3-2]
	\arrow[""{name=3, anchor=center, inner sep=0}, "{\D_2}", shift left=3, bend left=30, from=1-3, to=3-2]
	\arrow["{\beta^{1,i}}"', Rightarrow, from=0, to=1]
	\arrow["{\beta^{2,i}}", Rightarrow, from=3, to=2]
\end{tikzcd}
\end{center}

\smallskip
 
\begin{center}
\begin{tikzcd}[cramped]
	{\D_1(x)} & {\D_1(y)} && {\D_2(\iota(x))} & {\D_2(\iota(y))} \\
	{\D_{1,i}(x)} & {\D_{1,i}(y)} && {\D_{2,i}(\iota(x))} & {\D_{2,i}(\iota(y))}
	\arrow["{\D_1(\leq)}", from=1-1, to=1-2]
	\arrow["{\beta^{1,i}_x}"', from=1-1, to=2-1]
	\arrow[""{name=0, anchor=center, inner sep=0}, "{\beta^{1,i}_y}", from=1-2, to=2-2]
	\arrow["{\D_2(\iota(\leq))}", from=1-4, to=1-5]
	\arrow[""{name=1, anchor=center, inner sep=0}, "{\beta^{2,i}_{\iota(x)}}"', from=1-4, to=2-4]
	\arrow["{\beta^{2,i}_{\iota(y)}}", from=1-5, to=2-5]
	\arrow["{\D_{1,i}(\leq)}"', from=2-1, to=2-2]
	\arrow["{\D_{2,i}(\iota(\leq))}"', from=2-4, to=2-5]
	\arrow["\iota", rightarrow, maps to, shorten <= 20pt, shorten >= 25pt, from=0, to=1]
\end{tikzcd}
\end{center}

\smallskip

\begin{proof}
Take $y \leq x$ to be a morphism in $\CatJ_1$. Naturality of $\beta^{2,a}$ implies that
\begin{equation}\label{eq:Lem:naturalitybeta}
\beta^{2,a}_{\iota(y)} \circ \D_2(\iota(\leq)) = \D_{2,a}(\iota(\leq)) \circ \beta^{2,a}_{\iota(x)}.
\end{equation} 
Since $\D_{1,a}$ is $\D_{2,a} \circ \iota$ and $\beta^{1,a}$ is $\beta^{2,a}_{\iota(-)}$, Equation \eqref{eq:Lem:naturalitybeta} implies that 
\[
\beta^{1,a}_y \circ \D_1(\leq) = \D_{1,a}(\leq) \circ \beta^{1,a}_x.
\]
Each $\beta^{1,a}_j$ is an isomorphism for $j$ in $\{x, y\}$, so  $\beta^{1,a}$ is a natural isomorphism.

For any $\D_{1,b}$ that is naturally isomorphic to $\D_1$ and any $y$ in $\CatJ_2$, the index $y$ will have one of four properties: 
\begin{enumerate}
    \item there is a unique $x$ in $\CatJ_1$ so that \[y = \iota(x),\] where $x$ is unique since $\iota$ is faithful and injective on objects;
    \item $y$ is a morphism with neither domain nor codomain in $\iota(\CatJ_1)$;
    \item for some unique object $x$ in $\CatJ_1$, $y$ is a morphism from $\iota(x)$ to an object not in $\iota(\CatJ_1)$;
    \item for some unique $x$ in $\CatJ_1$, $y$ is a morphism from an object not in $\iota(\CatJ_1)$ to $\iota(x)$.
\end{enumerate}

For each case respectively, define $\D_{2,b}\colon \CatJ_2 \rightarrow \C$ by:
\begin{enumerate}
        \item\label{item:faithful} $\D_{2,b}$ takes $y$ to $\D_{1,b}(x)$;
        \item $\D_{2,b}$ takes $y$ to $\D_2(y)$;
        \item $\D_{2,b}$ takes $y$ to $\D_2(y) \circ (\beta^{1,b}_x)^{-1}$;
        \item $\D_{2,b}$ takes $y$ to $\beta^{1,b}_{x} \circ \D_2(y)$.
\end{enumerate}  
Item \ref{item:faithful} in the definition of $\D_{2,b}$ implies 
\[
\D_{1,b} = \D_{2,b} \circ \iota.
\]

Since $\D_{1,b}$, and $\D_2$ are functors and $\beta^{1,b}$ is a natural transformation, for any composite morphism $g \circ f$, 
\[
\D_{2,b}(g \circ f) = \D_{2,b}(g) \circ \D_{2,b}(f)
\] 
which demonstrates that $\D_2^b$ is a functor. Define $\beta^{2,b}\colon \D_2 \Rightarrow \D_{2,b}$ by
\[
\beta^{2,b}_{\iota(x)} = \beta^{1,b}_{x} \quad \text{and} \quad \beta^{2,b}_y = \id_y.
\]  
Take $y' \leq x'$ to be a morphism in $\CatJ_2$. Either $\beta^{2,b}_i$ is equal to $\beta^{1,b}_{\iota(j_i)}$, or $\beta^{2,b}_i$ is equal to $\id_i$
    where $i$ is in $\{x', y'\}$ and any $j_i$ is in $\CatJ_1$. In either case, cancellation laws of identities and naturality of $\beta_j^{1,b}$ together imply
\[
\beta^{2,b}_{y'} \circ \D_2(\leq) = \D_{2,b}(\leq) \circ \beta^{2,b}_{x'},
\] 
which together with each $\beta^{2,b}_i$ being an isomorphism implies that $\beta^{2,b}$ is a natural isomorphism.
\end{proof}

To formally capture the notion of subsystem inclusion, take $\D_1$ and $\D_2$ to be ACM-diagrams in $\CatC$ that actor-index categories $\CatJ_1$ and $\CatJ_2$ index,  respectively.  If $\iota^*\colon \D_1 \rightarrow \D_2$ is an ACM subdiagram and if $[\D_1^{X_1}]$ and $[\D_2^{X_2}]$ are the corresponding ACM-systems for $[\D_1]$ and $[\D_2]$, then define
\[
  [\D_1^{X_1}]\subseteq[\D_2^{X_2}]
\]
and say that $[\D_1^{X_1}]$ is an \emph{ACM subsystem} of $[\D_2^{X_2}]$. The symbol $\subseteq$ denotes \emph{subsystem inclusion} and refers to the specific inclusion functor $\iota\colon \CatJ_1 \rightarrow \CatJ_2$, which gives rise to an inclusion of the isomorphism classes of diagrams in the skeleton of $\Cat \downarrow \C$. Lemma~\ref{lem:natisosacrossinclusions} implies that $[\D_1^{X_1}]\subseteq[\D_2^{X_2}]$ is well-defined as a subclass containment between natural isomorphism classes of diagrams and that the natural isomorphisms between diagrams are compatible with subsystem inclusions.

The goal is to identify a category where these inclusions serve as morphisms that decompose an open system into simple pieces. The diagram $\D_2$ could consist of exactly the actors of $\D_1$ and additional constraints. If $\C$ and $\C'$ are $\mathsf{Set}$ and $\F$ is the identity functor, this results in the diagrams $\D_1$ including into $\D_2$ but apices of $\F$-limits (which in this case are just limits) $X_2$ including into $X_1$.

\subsection{The rigid inclusion category}

\begin{definition}
Take $[\D_1^{X_1}]$ and $[\D_2^{X_2}]$ to be ACM-systems in $\CatC$ with representative ACM-diagrams $\D_1$ and $\D_2$ which have shape $\CatJ_1$ and $\CatJ_2$, respectively.  A subsystem inclusion ${[\D_1^{X_1}]\subseteq[\D_2^{X_2}]}$ defined by an ACM subdiagram $\iota^*\colon \D_1 \rightarrow \D_2$ is a \emph{simple rigid inclusion} if precisely one of the following hold:
\begin{enumerate}
\item[$\bullet$] $\iota$ is an isomorphism; 
\item[$\bullet$] $\CatJ_2$ has exactly one additional constraint index (that is, one additional constraint morphism) for an existing actor index in $\CatJ_2$ than has $\CatJ_1$;
\item[$\bullet$] $\CatJ_2$ has exactly one more actor index than $\CatJ_1$, with only a trivial constraint morphism.
\end{enumerate}
The subsystem inclusion $[\D_1^{X_1}]\subseteq[\D_2^{X_2}]$ is a \emph{rigid inclusion} if it is a finite composite of simple rigid inclusions.  If $\CatJ_2$ has an additional actor index, then say that $\subseteq$ \emph{includes an additional actor}. If $\CatJ_2$ has an additional constraint index, then say that $\subseteq$ \emph{includes an additional constraint}.
\end{definition}

In addition to the conditions on $\F$ that Section~\ref{ACM-diagrams-and-weldings} initially imposed, henceforth require that $\F$ has nondegenerate $\F$-pullbacks in $\CatC^\prime$, $\F$ is faithful, cone-tight, and span-extendable. The principal abstract results of this work are Theorems~\ref{MainA} and ~\ref{MainB}.

\begin{theorem}\label{MainA}
Composition of rigid inclusions is composition in a category $\Kin(\F)$. Furthermore, composition of inclusion functors of ACM-diagrams uniquely determines the composition of rigid inclusions.
\end{theorem}

\begin{proof}
For any rigid inclusions $f$ and $g$ where 
\[
f\colon [\D_1^{X_1}] \rightarrow [\D_2^{X_2}] \quad \text{and} \quad g\colon [\D_2^{X_2}]  \rightarrow [\D_3^{X_3}],
\] 
take 
\[
\D_1\colon \CatJ_1 \rightarrow \C, \quad \D_2\colon \CatJ_2 \rightarrow \C, \quad \text{and} \quad \D_3\colon \CatJ_3 \rightarrow \C
\] 
to be representative ACM-diagrams of the ACM-systems $[\D_1^{X_1}]$, $[\D_2^{X_2}]$, and $[\D_3^{X_3}]$, respectively. The rigid inclusions $f$ and $g$ are determined by functors $F\colon \CatJ_1 \rightarrow \CatJ_2$ and $G\colon \CatJ_2 \rightarrow \CatJ_3$ with 
\[
\D_2 \circ F = \D_1 \quad \text{and} \quad \D_3 \circ G = \D_2.
\] 
The equality 
\[
\D_3 \circ G \circ F = \D_2 \circ F = \D_1
\] 
and compositionality of inclusion functors together imply that $G \circ F$ is an ACM subdiagram of $\D_3$. Lemma~\ref{lem:natisosacrossinclusions} implies that $g \circ f$ is uniquely determined by the composition of the underlying inclusion functors which have a well-defined domain $\CatJ_1$ and codomain $\CatJ_3$. Since $f$ and $g$ are finite composites of simple rigid inclusions, so is their composite $g \circ f$, hence it is a rigid inclusion. 

Composition of functors is associative and rigid inclusions and their composition are uniquely determined by functors and the composition of those functors, so composition of rigid inclusions is also associative.

For any ACM-system $[\D^X]$ with representative 
\[
\D^X\colon \CatJ^X \rightarrow \C,
\] 
the identity functor $\id_\CatJ\colon \CatJ \rightarrow \CatJ$ on the shape $\CatJ$ is an isomorphism of actor-index categories. Thus, $\id_\CatJ^*\colon \D \rightarrow \D$ defines a simple rigid inclusion $\id_{[\D^X]}\colon [\D^X] \rightarrow [\D^X]$.  Lemma~\ref{lem:natisosacrossinclusions} implies that this simple rigid inclusion acts on each ACM-diagram for a diagram in $[\D^X]$ as the same identity functor on the shape $\CatJ$, hence it is an identity morphism in $\Kin(\F)$.

Since $\Kin(\F)$ is closed under an associative composition operation and has identity morphisms, $\Kin(\F)$ is a category whose objects are ACM-systems and whose morphisms are rigid inclusions.
    
\end{proof}

\begin{remark}
Every ACM-diagram is a functor $\D\colon\CatJ\to\CatC$. Hence ACM-systems correspond to objects of the slice category $\mathsf{Cat}\downarrow\CatC$ modulo natural isomorphism of diagrams. A rigid inclusion determined by an inclusion functor
$\iota\colon\CatJ_1\hookrightarrow\CatJ_2$ yields a commuting triangle
\[
\D_1=\D_2\circ\iota,
\]
and therefore defines a morphism in $\mathsf{Cat}\downarrow\CatC$. Theorem~\ref{MainA} implies that rigid inclusions depend only on the natural
isomorphism classes of the diagrams, so this assignment defines a faithful functor
\[
\Kin(\F)\hookrightarrow (\mathsf{Cat}\downarrow\CatC)/\!\!\sim
\]
where $\sim$ denotes equivalence of objects under natural isomorphism.
\end{remark}

Simple rigid inclusions depend only on the inclusions of the underlying ACM-diagrams, so if $[\D_1^{X_1}]$ is an ACM subsystem of $[\D_2^{X_2}]$ such that $\CatJ_2$ has either exactly one more actor index than $\CatJ_1$ with a trivial constraint morphism, or exactly one more constraint morphism than $\CatJ_1$, then there is a simple rigid inclusion $[\D_1^{X_1}] \rightarrow [\D_2^{X_2}]$.

A straightforward induction argument proves Lemma~\ref{lem:SimpleHasRigidInclusion}.  The challenge lies not in constructing an appropriate chain of decomposing rigid inclusions from an ACM subdiagram $\D_1$ of an ACM-diagram $\D_2$ to $\D_2$, but in showing that each ACM-diagram that appears in this chain admits an $\F$-limit. If $\D_2$ decomposes external constraints, then Lemma~\ref{lem:SupersimpleSubdiags} guarantees that every intermediate ACM-diagram in this chain has an $\F$-limit.

\begin{lemma}\label{lem:SimpleHasRigidInclusion}
For any ACM subdiagram $\D_1$ of an ACM-diagram $\D_2$ that decomposes external constraints, there is a rigid inclusion $\iota\colon [\D_1^{X_1}] \rightarrow [\D_2^{X_2}]$ from the system $[\D_1^{X_1}]$ for $\D_1$ to a system $[\D_2^{X_2}]$ for $\D_2$.
\end{lemma}

The open ACM-systems are the morphisms of $\Kin(\F)$ and contain all information about how systems are part of a larger system, and how they may combine to form larger systems. Theorem~\ref{thm:UnionIsSystem} shows how a total system may be constructed from two subsystems that intersect only along constraints.  This theorem demonstrates how information about constraints may be used to not only decompose a system into open subsystems, but to compose subsystems into larger systems.

Proposition~\ref{prop:prodUnionHasLimit} and Lemma~\ref{lem:SimpleHasRigidInclusion} together imply Theorem~\ref{thm:UnionIsSystem}.

\begin{theorem}\label{thm:UnionIsSystem}
For any ACM-systems $[\D_1^{X_1}]$ and $[\D_2^{X_2}]$, if $\D_1$ and $\D_2$ are ACM subdiagrams of an ACM-diagram $\D$ that decomposes into constraints so their intersection $\D_{12}$ is either
\begin{enumerate}
	\item a diagram with no actors, or
	\item an ACM subdiagram of $\D_1$ and $\D_2$ such that for every actor $a$ in $\D_{12}$, \[\iota_1(C_a) = C_{\iota_1(a)} \quad \text{and} \quad \iota_2(C_a) = C_{\iota_2(a)},\]
\end{enumerate}
    and if $\D$ is their union over $\D_{12}$, then there is an ACM-system $[\D^X]$ for $\D$ and rigid inclusions \[[\D_1^{X_1}] \rightarrow [\D^{X}] \quad \text{and} \quad [\D_2^{X_2}] \rightarrow [\D^{X}].\]
\end{theorem}

\begin{definition}\label{Def:Skeleton}
A \textit{constraint skeleton} $G_{\D}$ of an ACM-diagram $\D$ is an undirected graph whose vertices are actor indices of $\D$. There is an edge between vertices $a_i$ and $a_j$ if there is a nontrivial constraint index between $a_i$ and $a_j$, which is to say that 
\[
C_i \cap C_j \neq \{\star\}.
\]
\end{definition}

\begin{theorem}\label{Thm:Acyclic}
If $\D$ is an ACM-diagram whose constraint skeleton is acyclic, then $\D$ is reducible to decomposable.
\end{theorem}

\begin{proof}
Every ACM-diagram $\D$ with only one actor decomposes external constraints because its unique actor has no nontrivial external constraints.

Fix a natural number $n$. Assume the theorem holds for every ACM-diagram with $n$ or fewer actors whose constraint skeleton is connected and acyclic. Take an ACM-diagram $\D$ with $n+1$ actors, and write $\Gamma$ for its constraint skeleton. Denote by $V(\Gamma)$ and $E(\Gamma)$ the vertex and edge sets of $\Gamma$.

If $\Gamma$ is not connected, then it has finitely many connected components, each of which has $n$ or fewer vertices.  Under the inductive hypothesis, each component corresponds to an ACM subdiagram that is reducible to decomposabe.  Proposition~\ref{prop:redtosimpthenFlimit} shows that each of these subdiagrams is weldable to a single welded actor. These welded actors have no external constraints, and so $\D$ is reducible to decomposable.

If $\Gamma$ is connected and every vertex of $\Gamma$ has degree at least two, then
\[
2|E(\Gamma)| = \sum_{v\in V(\Gamma)}\deg(v) \ge 2|V(\Gamma)|, \quad \text{hence} \quad |E(\Gamma)|\ge |V(\Gamma)|.
\]
Since $\Gamma$ is connected and acyclic, it is a tree, so \[|E(\Gamma)|=|V(\Gamma)|-1,\quad \text{contradicting the inequality} \quad |E(\Gamma)|\ge |V(\Gamma)|.\] Therefore, if $\Gamma$ is connected, then it has at least one vertex of degree one. Equivalently, $\D$ has an actor index $a$ whose actor shares nontrivial constraints with only one other actor of $\D$.

Take $\CatJ$ to be the shape of $\D$ and $\CatJ'$ to be the full subcategory on the objects of $\CatJ$ obtained by removing $a$, any interaction index involving $a$, and any constraint index internal to the actor associated with $a$. Take
\[
\D' := \D|_{\CatJ'},
\]
so that $\D'$ has $n$ actors, and its constraint skeleton is acyclic. The inductive hypothesis implies that $\D'$ is reducible to decomposable. Proposition~\ref{prop:redtosimpthenFlimit} therefore provides an $\F$-limit of $\D'$, whose apex $X$ is isomorphic to the single actor obtained by reducing $\D'$ by welding.

Since $a$ shares nontrivial constraints with at most one actor of $\D'$, every external constraint connecting $\D(a)$ to $\D'$ already occurs as an external constraint connecting $\D(a)$ to a single actor of $\D'$.  Denote by $C$ the $\F$-product of these external constraints. The actor $\D(a)$ admits a morphism to $C$, and $X$ admits a morphism to $C$. Therefore, the actors $\D(a)$ and $X$ form a two-actor system that decomposes external constraints, and so $\D$ is reducible to decomposable.
\end{proof}

\begin{definition}
A set of constrained actors $\mathcal A$ is \emph{admissible} if there is an ACM-system $[\D^X]$ with $\mathcal A$ as its set of constrained actors.
\end{definition}

\begin{theorem}\label{MainB}
For any set of constrained actors $\A$ in $\CatC$ that are the constrained actors for an ACM-diagram $\D$, if $\D$ is reducible to an ACM-diagram that decomposes external constraints, then $\A$ is admissible.  Furthermore, there is a rigid inclusion whose target is $[\D^{\prime X'}]$, where $\D'$ is a decomposing reduction of $\D$, and whose source is an isomorphism class of $\F$-limits over a single actor of $\D'$.
\end{theorem}

Theorem \ref{MainB} supplies a concrete sufficient condition for a set of actors to constitute an ACM-system and furnishes a compositional framework that reflects the system’s openness.  

\begin{proof}
Proposition~\ref{prop:redtosimpthenFlimit} implies that any ACM-diagram $\D$ that is reducible to an ACM-diagram that decomposes external constraints has an $\F$-limit $\Phi^X$. The $\F$-limit is defined only up to a cone isomorphism, which is unique since $\F$ is cone-tight, so Proposition~\ref{lem:IsoDiagIsoDiagX} implies that there is a unique ACM-system $[\D^X]$ of $\D$.  The system $\A$ is, therefore, an admissible system of constrained actors.

A diagram consisting of a single constrained actor $A$ of a decomposing reduction $\D'$ of $\D$ is an ACM-diagram that decomposes external constraints and is an ACM subdiagram of $\D'$. Lemma~\ref{lem:SimpleHasRigidInclusion} implies that there is a rigid inclusion $[A] \rightarrow [\D^{\prime X'}]$ from the system $[A]$ with underlying diagram having only the constrained actor $A$. 
\end{proof}


\section{Open CMK systems as rigid inclusions}\label{Sec:CMK-systems-as-rigid-inclusions}


A classical kinematic system is a model for describing the motion of point particles without considering the forces that cause their motion.  The configuration space of the system is a smooth manifold $M$.  A state of the system is a point in the tangent or cotangent bundle of $M$, $TM$ or $T^\ast M$, respectively \cite{AM, Arn}.  A path of motion is a smooth function $c$ that takes values in $M$ and whose domain is an interval $I$.  With this in mind, it is important to specialize $\Kin(\F)$ by specializing $\F$.

\subsection{CMK systems and open CMK systems}

Denote by $\Diff$ the category whose objects are smooth manifolds and whose morphisms are smooth functions between smooth manifolds.  Denote by $\SurjSub$ the subcategory of $\Diff$ whose morphisms are surjective submersions.  

\begin{definition}\label{def:CMK-systeM}
A CMK system is an ACM-system in $\SurjSub$ where the functor $\F$ is the inclusion functor from $\SurjSub$ to $\Diff$.  
\end{definition}

The principal concrete result of the current paper is Theorem~\ref{Main:SurjSub:A}. Henceforth, take $\F$ to be the inclusion functor from $\SurjSub$ to $\Diff$.

\begin{theorem}\label{Main:SurjSub:A}
Composition of rigid inclusions is composition in a category $\mathsf{Kin}(\F)$ whose objects are ACM-systems in $\SurjSub$ and whose morphisms are rigid inclusions.
\end{theorem}

\begin{proof}

The inclusion functor $\F$ from $\SurjSub$ to $\Diff$ is span-tight and $\SurjSub$ has $\F$-pullbacks, \cite[Theorem~4.5]{WY}.  Furthermore, both $\SurjSub$ and $\Diff$ have terminal objects, namely one-point manifolds, which $\F$ preserves. Since $\SurjSub$ is a subcategory of $\Diff$ with the same objects, the functor $\F$ is the inclusion on hom-sets, hence faithful.

Take any cospan $C$ and any $\F$-pullback $\langle f, g\rangle$ of $C$ in $\SurjSub$, where
\[
f\colon X \to Y \quad \text{and} \quad g\colon X \to Z.
\]
If $X$ is terminal, then $f$ and $g$ are surjective, so both $Y$ and $Z$ are one-point manifolds and hence terminal. In particular, if at least one of $Y$ or $Z$ is not terminal, then $X$ is not terminal, so $\SurjSub$ has non-degenerate $\F$-pullbacks.

To obtain Theorem~\ref{Main:SurjSub:A} from Theorem~\ref{MainA}, it suffices to prove that $\F$ is cone-tight and span-extendable.

To show that $\F$ is cone-tight, take $\Phi^A$ and $\Phi^B$ to be any two $\F$-limits in $\SurjSub$ of a diagram $\D$. Since $\F(\Phi^A)$ and $\F(\Phi^B)$ are limits in $\Diff$ of $\F \circ \D$, there is a unique cone diffeomorphism
\[
\Psi \colon \F(A) \to \F(B).
\]
Every diffeomorphism is a surjective submersion, so $\Psi$ is a morphism in $\SurjSub$. Regard $\Psi'$ as the morphism in $\SurjSub$ that acts as $\Psi$ on the underlying objects; then $\Psi'$ is a cone isomorphism between $\Phi^A$ and $\Phi^B$. The functor $\F$ is, therefore, cone-tight.  The faithfulness of $\F$ implies that $\Psi'$ is unique.

To show that $\F$ is span-extendable, take
\[
f\colon A \to B, \quad g\colon C \to B, \quad f'\colon A' \to A, \quad \text{and} \quad g'\colon C' \to C
\]
to be morphisms in $\SurjSub$, the span $\langle \pi_{A'}, \pi_{C'}\rangle$ an $\F$-pullback of the cospan $\langle f \circ f', g \circ g'\rangle$ such that the source of this $\F$-pullback is the fibered product $A' \times_B C'$, and the span $\langle \pi_A, \pi_C \rangle$ and $\F$-pullback of the cospan $\rangle f, g \langle$ such that the source of this $\F$-pullback is the fibered product $A \times_B C$. The surjective submersions $\pi_{A'}$, $\pi_{C'}$, $\pi_A$, and $\pi_C$ are the restricted projections from the product.

Since $\langle f' \circ \pi_{A'}, g' \circ \pi_{C'}\rangle$ is a span over the cospan $\langle f \circ f', g \circ g'\rangle$ and $\langle \pi_A, \pi_C \rangle$ is a pullback in $\Diff$, there is a unique span morphism
\[
\Psi\colon A' \times_B C' \to A \times_B C,
\]
given by
\[
\Psi(a',c') = \big(f'(a'), g'(c')\big).
\]

To show that $\Psi$ is surjective, take $(a,c)$ in $A \times_B C$, so
\[
f(a) = g(c).
\]
The maps $f'$ and $g'$ are surjective, so there exist $a'$ in $A'$ and $c'$ in $C'$ with
\[
f'(a') = a \quad \text{and} \quad g'(c') = c.
\]
The equalities
\[
f(f'(a')) = f(a) = g(c) = g(g'(c')),
\]
imply that
\[
f \circ f'(a') = g \circ g'(c').
\]
The pair $(a',c')$ is in $A' \times_B C'$, hence
\[
\Psi(a',c') = \big(f'(a'), g'(c')\big) = (a,c),
\]
so $\Psi$ is surjective.

To show that $\Psi$ is a submersion, take $(a',c')$ to be in $A' \times_B C'$ and $v$ to be in $T_{\Psi(a',c')}(A \times_B C)$. Write \[\Psi(a',c') = (a,c),\quad \text{so that} \quad f(a) = g(c).\]  Since $A \times_B C$ is a submanifold of $A \times C$,
\[
v = (v_1,v_2), \quad \text{for some} \quad v_1 \in T_aA\quad \text{and}\quad v_2 \in T_cC.
\]
Because $A \times_B C$ is the fiber product of the submersions $f$ and $g$, its tangent space is
\[
T_{(a,c)}(A \times_B C)
  = \{(u_1,u_2) \in T_aA \times T_cC \mid {\rm d}f_a(u_1) = {\rm d}g_c(u_2)\}.
\]
Since $v$ is in $T_{(a,c)}(A \times_B C)$, 
\[
{\rm d}f_a(v_1) = {\rm d}g_c(v_2).
\]

The map $\Psi$ is given by
\[
\Psi = (f' \circ \pi_{A'},\, g' \circ \pi_{C'}),
\]
so the isomorphism \[T(A\times_BC) \cong TA\times_{TB}TC\] and the chain rule together imply that
\[
{\rm d}\Psi_{(a',c')} = ({\rm d}f'_{a'} \circ {\rm d}\pi_{A'},\, {\rm d}g'_{c'} \circ {\rm d}\pi_{C'}).
\]
Since $f'$ and $g'$ are submersions, the linear maps
\[
{\rm d}f'_{a'} \colon T_{a'}A' \to T_aA \quad \text{and} \quad {\rm d}g'_{c'} \colon T_{c'}C' \to T_cC
\]
are surjective, and so there exist
\[
w_1 \in T_{a'}A' \quad \text{and} \quad w_2 \in T_{c'}C'
\]
such that
\[
{\rm d}f'_{a'}(w_1) = v_1 \quad \text{and} \quad {\rm d}g'_{c'}(w_2) = v_2.
\]

Use the condition \[{\rm d}f_a(v_1) = {\rm d}g_c(v_2)\] and the equalities \[v_1 = {\rm d}f'_{a'}(w_1)\quad \text{and}\quad v_2 = {\rm d}g'_{c'}(w_2)\] to obtain the equalities
\[
{\rm d}f_a\big({\rm d}f'_{a'}(w_1)\big)
  = {\rm d}f_a(v_1)
  = {\rm d}g_c(v_2)
  = {\rm d}g_c\big({\rm d}g'_{c'}(w_2)\big),
\]
hence
\[
{\rm d}f_a \circ {\rm d}f'_{a'}(w_1) = {\rm d}g_c \circ {\rm d}g'_{c'}(w_2).
\]
The tangent space of the fiber product $A' \times_B C'$ is given by the equality
\[
T_{(a',c')}(A' \times_B C')
  = \{(u_1,u_2) \in T_{a'}A' \times T_{c'}C' \mid {\rm d}f_a \circ {\rm d}f'_{a'}(u_1)
      = {\rm d}g_c \circ {\rm d}g'_{c'}(u_2)\},
\]
so
\[
(w_1,w_2) \in T_{(a',c')}(A' \times_B C').
\]

Finally, the equalities
\[
{\rm d}\Psi_{(a',c')}(w_1,w_2)
  = \big({\rm d}f'_{a'}(w_1), {\rm d}g'_{c'}(w_2)\big)
  = (v_1,v_2)
  = v
\]
imply that ${\rm d}\Psi_{(a',c')}$ is surjective for every $(a',c')$ in $A' \times_B C'$, and so $\Psi$ is a submersion.

\end{proof}

The functor $\F$ is cone-tight, span-extendable, and $\SurjSub$ has non-degenerate $\F$-pullbacks.  Theorem~\ref{MainB} implies that $\SurjSub$ has $\F$-limits of reducible to decomposable ACM-diagrams.  However, $\SurjSub$ does not have limits since it does not have pullbacks \cite[p. 451]{WY}, nor does it have $\F$-limits, as Example~\ref{nonexample:Flimitb} demonstrates.

\begin{example}\label{nonexample:Flimitb}

Take $h_1$ and $h_2$ to be the functions
\[h_1\colon \R \rightarrow \R \quad \text{and} \quad h_2\colon \R^2 \rightarrow \R\]
given by
\[h_1(x) = \begin{cases}
    0 & x \leq 0 \\
    {\rm e}^{-\frac{1}{x}} & x > 0,
\end{cases}\]
and
\[h_2(x,y) = h_1(x)h_1(y) + h_1(-x)h_1(y) + h_1(x)h_1(-y) + h_1(-x)h_1(-y).\]
Take $\D$ to be a diagram in $\SurjSub$ with actors $A_1$, $A_2$, and $A_3$, and constraints $C_{4}$, $C_{5}$, and $C_{6}$, all equal to $\R^2$, and constraint morphisms 
\begin{align*}
&f_{1,4}\colon A_1 \rightarrow C_4, \quad f_{1,5}\colon A_1 \rightarrow C_5, \quad f_{2,4}\colon A_2 \rightarrow C_4, \quad
    f_{2,6}\colon A_2 \rightarrow C_6, \quad \text{and} \quad f_{3,5}\colon A_3 \rightarrow C_5,
\end{align*}
given by \[f_{1,4} = f_{1,5} = f_{2,4} = f_{2,6} = f_{3,5} = \id_{\R^2},\quad \text{where} \quad \id_{\mathds R^2}(x,y) = (x,y),\] and \[f_{3,6}\colon A_3 \rightarrow C_6, \quad \text{where} \quad f_{3,6}(x,y) = (x+h_2(x,y), y+h_2(x,y)).\]

A straightforward calculation shows that all maps $f_{i,j}$ above are surjective submersions. Since $\SurjSub$ has $\F$-pullbacks, take $\D$ to also have interactions, which are $\F$-pullbacks of the cospans $\rangle f_{1,4}, f_{2,4} \langle$, $\rangle f_{1,5}, f_{3,5} \langle$, and $\rangle f_{2,6}, f_{3,6} \langle$ and are all diffeomorphic to $\R^2$. Therefore, $\D$ is an ACM-diagram that looks like this:

\[\begin{tikzcd}[ampersand replacement=\&,column sep=tiny]
	\&\& {A_1 \times_{C_4} A_2} \\
	{A_1} \&\& {C_4} \&\& {A_2} \\
	\& {C_5} \&\& {C_6} \\
	{A_1 \times_{C_5} A_3} \&\& {A_3} \&\& {A_2 \times_{C_6} A_3}
	\arrow[from=1-3, to=2-1]
	\arrow[from=1-3, to=2-5]
	\arrow["{f_{1,4}}", from=2-1, to=2-3]
	\arrow["{f_{1,5}}"', from=2-1, to=3-2]
	\arrow["{f_{2,4}}"', from=2-5, to=2-3]
	\arrow["{f_{2,6}}", from=2-5, to=3-4]
	\arrow[from=4-1, to=2-1]
	\arrow[from=4-1, to=4-3]
	\arrow["{f_{3,5}}", from=4-3, to=3-2]
	\arrow["{f_{3,6}}"', from=4-3, to=3-4]
	\arrow[from=4-5, to=2-5]
	\arrow[from=4-5, to=4-3]
\end{tikzcd}\]

If it exists, an $\F$-limit of $\D$ is a limit of $\F \circ \D$ in $\Diff$, which is isomorphic to an equalizer of the morphisms \[f_{1,4} \times f_{3,5} \times f_{2,6}\colon A_1 \times A_2 \times A_3 \rightarrow C_4 \times C_5 \times C_6\] with \[f_{2,4} \times f_{1,5} \times f_{3,6}\colon A_1 \times A_2 \times A_3 \rightarrow C_4 \times C_5 \times C_6.\] Since the forgetful functor from $\Diff$ to $\Set$ is representable \cite[Lemma~3.2]{WY}, to show that an $\F$-limit of $\D$ does not exist, it suffices to show that this set $X$ with the subspace topology of $\R^6$ does not admit a manifold structure:
\[
X = \{p = (x_1, y_1, x_2, y_2, x_3, y_3) \in A_1 \times A_2 \times A_3 \mid f_{1,4} \times f_{3,5} \times f_{2,6}(p) = f_{2,4} \times f_{1,5} \times f_{3,6}(p) \}.
\]
If $X$ were a manifold, then it would be diffeomorphic to 
\[
\{(x,y) \in \R^2 \mid x = 0 \text{ or } y = 0\},
\]
which does not admit a manifold structure. Hence $\D$ does not have an $\F$-limit. 
\end{example}

The diagram $\D$ does not decompose external constraints, since all constraints are external to each actor and the product morphisms $f_{1,4} \times f_{1,5}$, $f_{2,4} \times f_{2,6}$, and $f_{3,5} \times f_{3,6}$ are from $2$-manifolds to $4$-manifolds, so they are not surjective submersions. However, the subdiagram $\D_1$ of $\D$ that lacks $A_1$ and the subdiagram $\D_2$ of $\D$ that lacks $A_2$, $A_3$, and $C_6$ do decompose external constraints and intersect only along constraints $C_4$ and $C_5$. Further, $\D$ is a union of $\D_1$ and $\D_2$. This example shows that decomposing external constraints is insufficient to guarantee that the union of two diagrams that each decompose external constraints has an $\F$-limit.

\subsection{The Newton Daemon}

Recall that the category $\Kin(\F)$ has CMK systems as objects and open CMK systems as morphisms, where $\F$ is the inclusion fuctor $\F\colon\SurjSub\to\Diff$. For a CMK system $[\D^X]$ with a constraint index category $\CatJ$ and constraint indices $\text{Ob}_C(\CatJ)$, choose an $\F$-limit $X$ and write 
\[
  \Phi^X_C\colon \F(X) \longrightarrow \prod_{c\in \text{Ob}_C(\CatJ)} \F \circ \D(c)
\]
for the product of the legs of $\F(\Phi^X)$ from $\F(X)$ to constraints.  The paths of motion of the CMK system $[\D^X]$ with $X$ as the specified configuration space are the smooth paths in $X$.  Some physically realistic systems require time-dependent restrictions that lie outside the basic ACM construction. To describe them, introduce the following notion.

\begin{definition}
Take a non-empty subset $\Lambda$ of the constraint index set $\text{Ob}_C(\CatJ)$ and the function $\pi$ to be the product
\[
\pi = \bigtimes_{c\in\Lambda} \Phi^X_c \colon X \to \prod_{c\in\Lambda} \D(c).
 \]
For any $\Lambda$ so that $\pi$ is a surjective submersion, and for any nonempty real interval $I$ and any $r$ in $\mathds N_0$, a $C^r$-\emph{Newton Daemon} with time interval $I$ is a $\mathcal C^r$ path
\[
\mathcal N\colon I \to \prod_{c\in\Lambda} \D(c), \quad t \;\mapsto\; \mathcal N_t .
\]
For each $t$ in $I$ define the \emph{daemon-restricted configuration manifold}
\[
M_t := \big\{x\in X \mid \pi(x)=\mathcal N_t\big\}.
\]
\end{definition}

A Newton Daemon is an exogenous, non-responsive controller: at each $t$ in $I$ it specifies the slice $M_t$, thereby restricting the configuration space available to the system.  The construction permits a description within the ACM framework of the admissible paths of systems that are over-constrained or required to interact with their environment in a prescribed way. Since $\pi$ in the definition above is a submersion, every $M_t$ embeds smoothly in $X$.  

If $\D$ decomposes into constraints, then $\Phi_C^X$ is a surjective submersion, so $\pi$ is also a surjective submersion for any choice of $\Lambda$. The framework of ACM-diagrams and the category $\Kin(\F)$ supports the construction of daemon-restricted configuration spaces and the study of time varying configuration manifolds. An upcoming work will explore this from a dynamical perspective.


\subsection{Planar and spatial linkages}


The classification of planar and spatial linkages within the framework of ACM-systems requires consideration of the ambient space and its symmetry group. Denote by $\SE(n)$ the special Euclidean group in dimension $n$, acting on $\mathds E^n$, the $n$-dimensional Euclidean space. Although the present discussion focuses on the physical cases when $n$ is $2$ or $3$, the basic definitions make sense in any dimension. The classical theory of linkages treats a linkage as a collection of rigid bodies, or \emph{links}, attached at joints that permit some degree of relative motion of the links. In dimension $n$, the group $\SE(n)$ preserves all rigid bodies; hence it describes all possible orientations of the links of a linkage. 

It is helpful to use a consistent diagrammatical scheme to visualize the modeling of linkages by ACM-diagrams.  This scheme enriches the labeling of interactions provided by a skeleton because it captures the way in which actors interact over shared constraints.  

\begin{example}\label{ex:planarActor}

View an actor $A$ in the plane  with configuration space $\SE(2)$ as a circular bearing, like this:

\medskip

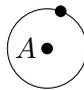
\begin{center} 
\begin{tikzpicture}[scale = 1]
\begin{scope}[xshift = 0in]
\coordinate(A) at (180:2);

\coordinate[shift = {(180:2)}](D) at (70:15pt);


\draw[] (A) circle (15pt);
\draw[fill = black] (D) circle (2pt);

\draw[fill = black] (A) circle (2pt) node[font = \small, anchor = east, outer sep = 0pt] at (A) {$A$};

\end{scope}
\end{tikzpicture}
\end{center} 

\noindent The center point indicates the position of the axle in space and the point on the circle indicates the orientation of the outer casing of the circular bearing, so the configuration space of the bearing is $\SE(2)$.

Take $\theta_a$ to indicate an element of $\mathds S^1$, viewed as the unit circle in $\mathds R^2$.  Identify a point in $\SE(2)$ by a triple $(x_a \theta_a + y_a \theta_a^\perp, \theta_a)$, where $x_a$ and $y_a$ are real numbers and $\theta_a^\perp$ is the point on $\SS^1$ that is rotated counterclockwise from $\theta$ by a fourth of a circle.
\end{example}

\begin{example}\label{ex:RigidBar}

Two planar actors $A$ and $B$ may be joined together by a rigid bar that forces an agreement between the orientation of the actors, like this:

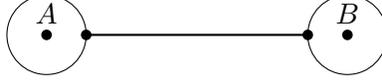
\begin{center} 
\begin{tikzpicture}[scale = 1]
\begin{scope}[xshift = 0in]

\coordinate (A) at (180:2);
\coordinate (B) at (0:2);

\coordinate[shift = {(180:2)}] (D) at (0:15pt);
\coordinate[shift = {(0:2)}]   (E) at (180:15pt);
\draw (B) circle (15pt);
\draw (A) circle (15pt);


\draw[ultra thick, white] (D) -- (E);
\draw[thick]             (D) -- (E);

\fill (D) circle (2pt);
\fill (E) circle (2pt);

\fill (A) circle (2pt);
\fill (B) circle (2pt);

\node[font=\small, above] at ($(A)$) {$A$};
\node[font=\small, above] at ($(B)!0.5!(C)$) {$B$};

\end{scope}
\end{tikzpicture}
\end{center}

\noindent Note that the placement of the bar in the diagram indicates an agreement of the angles $\theta_a$ and $\theta_b$.

The constraint space of the system is $\SE(2)$.  Given that the bar has (positive) length $L$, view this space as a fibered product by taking both $A$ and $B$ to be $\SE(2)$, and their shared constraint $C$ to also be $\SE(2)$. Define the constraint morphisms $\pi_{a,c}$ and $\pi_{b,c}$ by \[\pi_{a,c}\colon (x_a \theta_a + y_a \theta_a^\perp, \theta_a) \mapsto ((x_a+L) \theta_a + y_a \theta_a^\perp, \theta_a)\] and \[\pi_{b,c}\colon (x_b\theta_b + y_b \theta_b^\perp, \theta_b) \mapsto (-x_b \theta_b - y_b \theta_b^\perp, -\theta_b).\] 

An $\F$-pullback of $\rangle \pi_{a,c}, \pi_{b,c}\langle$  is span isomorphic to the fiber product %
\begin{align*}
X &= \big\{(x_a \theta_a + y_a \theta_a^\perp, \theta_a,  x_b \theta_b + y_b \theta_b^\perp, \theta_b) \in A \times B \\
&\hspace{2in}\mid ((x_a+L)\theta_a + y_a \theta_a^\perp, \theta_a) = (-x_b\theta_b - y_b\theta_b^\perp, -\theta_b)\big\}.
\end{align*}
The manifold $X$ is diffeomorphic to $\SE(2)$, and a reduction of the diagram by welding together $A$ and $B$ produces an externally unconstrained one-link linkage isomorphic to $A$ and $B$ individually, but further contains information about the distance between $A$ and $B$, which will become important in studying dynamics. 

Since the angles $\theta_a$ and $-\theta_b$ agree, it is more efficient to represent the system like this:

\smallskip

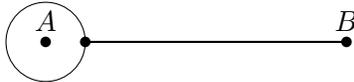
\begin{center} 
\begin{tikzpicture}[scale = 1]
\begin{scope}[xshift = 0in]

\coordinate (A) at (180:2);
\coordinate (B) at (0:2);
\coordinate[shift = {(0:15pt)}] (a) at (A);

\draw[ultra thick, white] (a) -- (B);
\draw[thick]             (a) -- (B);

\draw (A) circle (15pt);

\fill (A) circle (2pt);
\fill (a) circle (2pt);
\fill (B) circle (2pt);

\node[font=\small, above] at ($(A)$) {$A$};
\node[font=\small, above] at ($(B)!0.5!(C)$) {$B$};

\end{scope}
\end{tikzpicture}
\end{center}

\end{example}

The present modeling of linkages with the ACM framework characterizes linkages as an assemblage of oriented point-masses with respect to kinematic constraints.

\begin{enumerate}
  \item[$\bullet$] A \emph{linkage} in $n$-dimensional Euclidean space is an object of $\Kin(\F)$ such that every actor is isomorphic to $\SE(n)$.

\medskip

  \item[$\bullet$] An \emph{open linkage} is a morphism in $\Kin(\F)$ whose source and target are both linkages.

\medskip

  \item[$\bullet$] A (lower) \emph{kinematic pair} is a linkage with two actors.

\medskip

  \item[$\bullet$] An \emph{open kinematic pair} is a morphism in $\Kin(\F)$ whose source is a kinematic pair.
\end{enumerate}

\noindent Morphisms between two actors into the same constraint represent links in the ACM framework by encoding the constraint data. The modeling of linkages typically (but not always) involves fixing the position and orientation of a rigid body in the linkage, for example, a bar in a system of bars connected to other bars. The ACM framework does not initially fix the motion of a specified actor, but may do so using a Newton Daemon after constructing the linkage. Adjoining an additional actor and its constraint morphisms with an existing actor is kinematically equivalent to adjoining a link in a classical linkage since an additional link may have only one additional attachment point. It is also dynamically equivalent if the links are considered massless and the attachment points are massive. A forthcoming work treats dynamics in an ACM framework.
 
Modeling physical linkages imposes two additional structural features: Each actor and each constraint between actors carries a transitive smooth left action of $\SE(n)$, that is, it is an $\SE(n)$-manifold. Furthermore, constraint morphisms are smooth and $\SE(n)$-equivariant. This is to say that for any actor $A$, constraint $C$, constraint morphism $\pi_{a,c}$, element $(g,x)$ of $\SE(n)\times A$, \[\pi_{a,c}(gx) = g\pi_{a,c}(x).\]

In the classical theory of linkages, the group $\SE(n)$ appears as the group under which rigid bodies are invariant. The ACM framework makes the role of $\SE(n)$ explicit and mechanical. An actor in the plane or in space is modeled as a framed point particle that one may think of physically as a small bearing, whose configuration space is naturally $\SE(n)$. Constraints then specify how two such bearings may be attached or interact, encoding allowable relative configurations via constraint cospans. Kinematic pairs arise as pullbacks determined by these constraints. Accordingly, $\SE(n)$ appears not merely as a descriptive tool for relative motion, but as the intrinsic configuration space of the state-bearing entities themselves. In this way, the ACM framework does not introduce new structure; it relocates the role of $\SE(n)$ from an informal description of relative motion in the classical literature to a precise, compositional description of actor states and their interactions.

\begin{definition}
For any linkage $L$ that has a representative $\D^X$, where 
\[
\D\colon\CatJ\to\SurjSub,
\] 
the linkage $L$ is \emph{overconstrained} if 
\begin{equation}\label{Sec5:Eq:Overconstrained}
\sum_{a \in \text{Ob}_A(\CatJ)} \dim \D(a) < \dim \SE(n) + \sum_{c \in \Ext[\CatJ:\CatJ]} \dim \D(c). 
\end{equation}
\end{definition}

\begin{theorem}
If an ACM-diagram representing a linkage \simple, then the linkage is not overconstrained. 
\end{theorem}
\begin{proof}
For any linkage $[\D^X]$ with a representative ACM-diagram $\D$ that decomposes external constraints, there is a surjective submersion from every actor $A$ in $\D$ to the product of the external constraints of $A$. The codomain of a surjective submersion has dimension no larger than the domain and 
\[
\dim\Big(\prod_{c \in \Ext[a:\CatJ]} \D(c)\Big) = \sum_{c \in \Ext[a:\CatJ]} \dim \D(c), \quad \text{so} \quad \dim A \geq \sum_{c \in \Ext[a:\CatJ]} \dim\D(c). 
\]
When there is only a single link in the linkage $[\D^X]$, there are no external constraints, which implies the negation of \eqref{Sec5:Eq:Overconstrained}.

Take $k$ to be any natural number and assume that the negation of \eqref{Sec5:Eq:Overconstrained} holds for any ACM-diagram with $k$ actors. For any linkage $[\D^X]$ with $k+1$ actors and actor index category $\CatJ$, take $A_i$ to be an actor in a representative ACM-diagram $\D$ and $\D'\colon \CatJ' \rightarrow \SurjSub$ to be an ACM subdiagram of $\D$ where $\CatJ'$ has all actor indices except $i$, all constraint indices not internal to $i$, and all interaction indices determined by the remaining actor indices and constraint indices. Lemma~\ref{lem:SupersimpleSubdiags} implies $\D'$ decomposes external constraints,  so Proposition~\ref{prop:redtosimpthenFlimit} implies that $\D'$ has an $\F$-limit $\Phi^{\prime X'}$, so it has an associated linkage $[\D^{\prime X'}]$. Since $[\D'^{X'}]$ is a linkage with $k$ actors, %
\begin{equation}\label{Sec5:Eq:OverconstrainedThB}
\sum_{a \in \text{Ob}_A(\CatJ')} \dim A_a \geq \dim \SE(n) + \sum_{c \in \Ext[\CatJ':\CatJ']} \dim \D(c).
\end{equation}
Since $\D$ decomposes external constraints, 
\begin{equation}\label{Sec5:Eq:OverconstrainedThC}
\dim(A_i) \geq \sum_{c \in \Ext[i:\CatJ]} \D(c).
\end{equation}
Ineequalities~\eqref{Sec5:Eq:OverconstrainedThB} and \eqref{Sec5:Eq:OverconstrainedThC} together imply that
    \begin{align*}
        \sum_{a \in \text{Ob}_A(\CatJ)} \dim A_a &= \dim(A_i) + \sum_{a \in \text{Ob}_A(\CatJ')} \dim(A_a) \\ &\geq \sum_{c \in \Ext[i:\CatJ]} \D(c) +  \dim \SE(n) + \sum_{c \in \Ext[\CatJ':\CatJ']} \dim \D(c) \\
        &\geq \dim \SE(n) + \sum_{c \in \Ext[\CatJ:\CatJ]} \dim \D(c)
    \end{align*}
    where \[ \sum_{c \in \Ext[\CatJ:\CatJ]} \D(c) \leq \sum_{a \in \text{Ob}_A(\CatJ)} \sum_{c \in \Ext[a:\CatJ]} \D(c)\] implies the ultimate inequality.
\end{proof}

\begin{lemma}\label{lem:EquivariantMapsAreHomogeneous}
Take $G$ to be a Lie group with identity $e$ acting on itself by left multiplication and $M$ to be a smooth $G$--manifold. Take $p$ to be a smooth $G$--equivariant surjective submersion from $G$ to $M$ and write \[m_0:=p(e)\quad \text{and} \quad H:=\Stab(m_0),\] where $\Stab(m_0)$ indicates the stabilizer of $m_0$. The map
\[
\phi\colon G/H\to M \quad \text{by} \quad [g]\mapsto g\cdot m_0,
\]
is a $G$--equivariant diffeomorphism.
\end{lemma}

\begin{proof}

For any $m$ in $M$, surjectivity of $p$ implies that there is a $g$ in $G$ so that
\[
m = p(g) = p(g \cdot e) = g \cdot p(e) = g \cdot m_0.
\]
Thus $M$ is the orbit of a single point. For any $g$ and $g'$ in $G$, $p(g)$ is equal to $p(g')$ if and only if
\[
g^{-1}g' \cdot m_0 = m_0,
\]
which holds if and only if $g^{-1}g'$ is in $H$. Hence the fibers of $p$ are precisely the left cosets of $H$.

Since $H$ is the preimage of $\{m_0\}$ under the continuous map \[g\mapsto g\cdot m_0,\] it is a closed subgroup of $G$. Therefore $G/H$ is a smooth manifold and the quotient map
\[
\pi\colon G\to G/H
\]
is a surjective submersion.  Define $\phi\colon G/H\to M$ by
\[
\phi([g]) = g\cdot m_0.
\]
If $[g]$ is equal to $[g']$, then $g^{-1}g'$ is in $H$, hence \[g\cdot m_0=g'\cdot m_0,\] so $\phi$ is well-defined, and 
\[
p = \phi\circ\pi.
\]

For each $g$ in $G$, the kernel of $d\pi_g$ is $T_g(gH)$, while equivariance of $p$ implies that the kernel of ${\rm d}p_g$ is $T_g(gH)$ as well. The equality \[{\rm d}p_g = d\phi_{[g]}\circ d\pi_g\] together with the surjectivity of ${\rm d}\pi_g$ implies that ${\rm d}\phi_{[g]}$ is an isomorphism for every $[g]$ in $G/H$. Hence $\phi$ is a local diffeomorphism.

Since $\phi$ is bijective, it is a diffeomorphism. The identity
\[
\phi(k\cdot[g]) = \phi([kg]) = kg\cdot m_0 = k\cdot \phi([g])
\]
implies that $\phi$ is $G$-equivariant.
\end{proof}

\begin{lemma}\label{lem:FiberProductNormalForm}
For any $i$ in $\{1,2\}$, take $p_i$ to be a smooth $G$--equivariant surjective submersion from $G$ to $M$.  Write \[m_i := p_i(e) \quad \text{and} \quad H_i:=\Stab(m_i).\]
Take $P$ to be the fibered product \[P:=G\times_M G = \{(g_1,g_2)\in G\times G \mid p_1(g_1) \\= p_2(g_2)\}.\] For any $g_0$ in $G$ such that $m_1$ is equal to $g_0\cdot m_2,$
\begin{align*}
P = \big\{(g_1,g_2)\in G\times G \mid g_2^{-1}g_1g_0\in H_2\big\}.
\end{align*}
Moreover, the map
\begin{equation}\label{Lemmaeq:FiberProdJointDecomA}
\Phi\colon P\to G\times H_2 \quad \text{by} \quad \Phi(g_1,g_2):=(g_2, g_2^{-1}g_1g_0)
\end{equation}
is a $G$--equivariant diffeomorphism, where $G$ acts by left multiplication on the first factor and trivially on the second.
\end{lemma}

\begin{proof}
Follow Lemma~\ref{lem:EquivariantMapsAreHomogeneous} by taking \[m_i := p_i(e)\] so that the fiber product condition is equivalent to the equality
\[
g_2^{-1}g_1\cdot m_1 = m_2.
\]
Transitivity of the action of $G$ on $M$ implies that there is a $g_0$ in $G$ so that 
\[
m_1 = g_0\cdot m_2,
\] 
hence
\[
g_2^{-1}g_1\cdot (g_0\cdot m_2) = m_2  \quad \text{if and only if} \quad g_2^{-1}g_1g_0\in H_2,
\]
which implies \eqref{Lemmaeq:FiberProdJointDecomA}.

The condition that $g_2^{-1}g_1g_0$ is in $H_2$ makes $\Phi$ well-defined. Define
\[
\Phi^{-1}\colon G\times H_2\to P \quad \text{by} \quad  \Phi^{-1}(g,h) := (ghg_0^{-1},\ g).
\]
The equalities
\[
\Phi(\Phi^{-1}(g,h))=\Phi(ghg_0^{-1},g)=(g,\ g^{-1}(ghg_0^{-1})g_0)=(g,h),
\]
and
\[
\Phi^{-1}(\Phi(g_1,g_2))=\Phi^{-1}(g_2,\ g_2^{-1}g_1g_0)=(g_2(g_2^{-1}g_1g_0)g_0^{-1},\ g_2)=(g_1,g_2).
\]
imply that $\Phi$ is a diffeomorphism. The equalities 
\[
\Phi(kg_1,kg_2)=(kg_2,\ (kg_2)^{-1}(kg_1)g_0)=(kg_2,\ g_2^{-1}g_1g_0)=k\cdot \Phi(g_1,g_2)
\]
imply that $\Phi$ is $G$--equivariant under the diagonal left action on $P$ and the given $G$-action on $G\times H_2$.
\end{proof}

\paragraph{\bf The universal joint.}

Recall that the configuration space for a universal joint is isomorphic to $\SE(3)\times \mathds S^1\times \mathds S^1$.  In the literature, the orientation and position of a linkage is fixed by a fixed link, so the configuration space for a universal joint is typically given to be $\mathds S^1\times\mathds S^1$.  Lemma~\ref{lem:FiberProductNormalForm} guarantees that the system will have the same configuration space on fixing the location and orientation of one of the actors.

The ACM framework describes how a linkage may be assembled and thereby identifies obstructions to the existence of linkages arising from a given collection of actors and constraints. In this sense, the universal joint serves as a minimal and physically meaningful application showing that the ACM framework imposes genuine and nontrivial constraints on compositional modeling in classical mechanics.

\begin{theorem}[No Two--Actor Realization of the Universal Joint]
\label{thm:NoTwoActorUniversalJoint}
Take $\SE(3)$ to act on itself by left multiplication.  Take $X$ to be a smooth $\SE(3)$--manifold and
\[
p_1, p_2\colon \SE(3) \to X
\]
to be $\SE(3)$--equivariant surjective submersions. Denote by $P$ the $\F$-pullback source
\[
P = \SE(3) \times_X \SE(3)
\]
with the diagonal $\SE(3)$--action.  

There are no choices of $X$, $p_1$, and $p_2$ so that $P$ is $\SE(3)$--equivariantly diffeomorphic to $\SE(3) \times {\mathds S}^1 \times {\mathds S}^1$ with $\SE(3)$ acting on the first factor (by left multiplication).
\end{theorem}

\begin{proof}
For any $\SE(3)$--equivariant surjective submersion $p\colon \SE(3) \to X$, surjectivity and equivariance imply that the $\SE(3)$--action on $X$ is transitive. Transitivity identifies $X$ as isomorphic to a homogeneous space $\SE(3)/H$, where $H$ is the stabilizer of a point \cite{HelgasonDG}. Up to an $\SE(3)$--equivariant diffeomorphism of $X$, both $p_1$ and $p_2$ agree with the canonical quotient map
\[
\pi\colon \SE(3) \to \SE(3)/H.
\]
Replacing the original cospan by this canonical one does not change the pullback up to $\SE(3)$--equivariant diffeomorphism.

Define the smooth map
\[
\Phi\colon \SE(3) \times H \to \SE(3) \times_{\SE(3)/H} \SE(3), \quad \text{by} \quad \Phi(g,h) = (g, gh).
\]
The map $\Phi$ is an $\SE(3)$--equivariant diffeomorphism with smooth inverse
\[
(g_1,g_2) \to \big(g_1, g_1^{-1} g_2\big),
\]
so $P$ is isomorphic to $\SE(3) \times H$ as $\SE(3)$--manifolds. If $P$ is $\SE(3)$--equivariantly diffeomorphic to $\SE(3) \times {\mathds S}^1 \times {\mathds S}^1$, then $H$ is diffeomorphic to ${\mathds S}^1 \times {\mathds S}^1$ and hence is a $2$--dimensional compact Lie subgroup of $\SE(3)$.

Every compact subgroup of $\SE(3)$ is isomorphic to a conjugate subgroup of $\SO(3)$. The connected closed subgroups of $\SO(3)$ have dimensions $0$, $1$, or $3$. There is no $2$--dimensional closed connected subgroup of $\SO(3)$ and hence none in $\SE(3)$. This contradicts the requirement that $H$ be a $2$--torus. Therefore no such $X$, $p_1$, and $p_2$ exist.
\end{proof}

\begin{example}\label{Ex:Pendulum}
A pendulum may be constructed using a Newton Daemon and a rigid bar construction.  Take $A$ to be an actor whose position, but not orientation, a Newton Daemon fixes.  Take $B$ to be an actor connected to $A$ by a rigid bar so that aligns the rotation of $B$ by that of $A$, like this:

\medskip

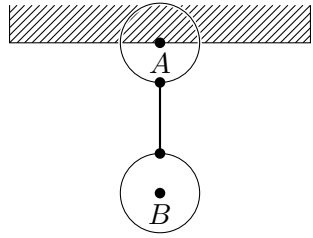
\begin{center}
\begin{tikzpicture}

\begin{scope}
\clip[] (-2, .5) -- (-2, 0) -- (2, 0) -- (2, .5) -- cycle;
\foreach \x in {-2.5, -2.4, ..., 2} {
        \pgfmathsetmacro{\vx}{\x + .5}
        \draw[thin] (\x, 0) -- (\vx, .5);
        }      
\end{scope}

\draw[] (-2, .5) -- (-2, 0) -- (2, 0) -- (2, .5);


\coordinate (A) at (0,0);
\coordinate (B) at (0,-2);

\coordinate[shift = {(0,0)}] (a) at (270:15pt);

\coordinate[shift = {(0,-2)}] (D) at (90:15pt);

\draw[ultra thick, white] (A) circle (15pt);
\draw (A) circle (15pt);
\draw (B) circle (15pt);

\draw[ultra thick, white] (a) -- (D);
\draw[thick] (a) -- (D);


\fill (D) circle (2pt);

\fill (A) circle (2pt);
\fill (a) circle (2pt);
\fill (B) circle (2pt);


\node[font=\small, below] at (A){$A$};

\node[font=\small, below] at (B) {$B$};

\end{tikzpicture}
\end{center}

The configuration space $X$ is given in Example~\ref{ex:RigidBar}.  Take $z$ to be the constraint index for the $\mathds R^2$ constraint for both $A$ and $B$, so that Example~\ref{ex:RigidBar} also provides a surjective submersion $\pi_{X,z}$ from $X$ to $\mathds R^2$ that is given by
\begin{align*}
&\pi_{X,z}\colon (x_a \theta_a + y_a \theta_a^\perp, \theta_a, (x_a+L)\theta_a + y_a \theta_a^\perp, -\theta_a) \to x_a\theta_a+y_b\theta_b^\perp
\end{align*}

Take $\gamma$ to be the path in $\mathds R^2$, the Newton Daemon, given by \[\gamma(t) = (0, 0),\] so that \[M_t = \pi_{X,z}^{-1}(\gamma(t)) = \{(0\theta_1+0\theta_1^\perp, \theta_1, L\theta_1 - 0\theta_1^\perp, -\theta_1)\},\] which is isomorphic to $\mathds S^1$, the configuration space of the pendulum.

\end{example}

\begin{example}\label{Ex:RevoluteJoint}
A \emph{revolute joint} or hinge has two actors $A$ and $B$ and an $\mathds R^2$ constraint $C$, with respective indices $a$, $b$, and $c$. Here is a representation of the system with actors viewed as bearings:

 \medskip
 
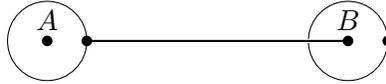
\begin{center} 
\begin{tikzpicture}[scale = 1]
\begin{scope}[xshift = 0in]


\coordinate (A) at (180:2);
\coordinate (B) at (0:2);

\coordinate (C) at (0:6);

\coordinate[shift = {(180:2)}] (a) at (0:15pt);

\coordinate[shift = {(180:2)}] (D) at (180:15pt);
\coordinate[shift = {(0:2)}]   (E) at (0:15pt);
\coordinate[shift = {(0:6)}]   (F) at (180:15pt);

\draw (A) circle (15pt);
\draw (B) circle (15pt);

\draw[ultra thick, white] (a) -- (B);
\draw[thick] (a) -- (B);


\fill (E) circle (2pt);

\fill (A) circle (2pt);
\fill (a) circle (2pt);
\fill (B) circle (2pt);

\node[font=\small, above] at (A){$A$};
\node[font=\small, above] at (B) {$B$};

\end{scope}
\end{tikzpicture}
\end{center}

\noindent The constraint space of the system is $\mathds R^2$.  Given that the bar has (positive) length $L_{ab}$, and that $A$ and $B$ are both copies of $\R^2\times \SS^1$, and that the shared constraint $C$ is $\mathds R^2$, define $\pi_{a,c}$ and $\pi_{b,c}$ by \[\pi_{a,c}\colon (x_a \theta_a+ y_a \theta_a^\perp, \theta_a) \mapsto (x_a-L_{ab})\theta_a + y_a \theta_a^\perp \quad \text{and} \quad \pi_{b,c}\colon (x_b \theta_b + y_b \theta_b^\perp, \theta_b) \mapsto x_b\theta_b + y_b \theta_b^\perp.\] 

An $\F$-pullback of $A$ and $B$ over this constraint $C$ is isomorphic to the fiber product \[\big\{(x_a \theta_a + y_a \theta_a^\perp, \theta_a,  x_b \theta_b + y_b \theta_b^\perp, \theta_b) \in A \times B \mid (x_a-L_{ab})\theta_a + y_a\theta_a^\perp = x_b\theta_b + y_b\theta_b^\perp\big\}.\]
This configuration manifold is diffeomorphic to $\SE(2) \times \mathds S^1$. 
\end{example}

\begin{example}\label{Ex:RevoluteJoint2} This diagram represents two linked revolute joints:

 \medskip
 
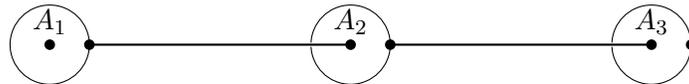
\begin{center} 
\begin{tikzpicture}[scale = 1]
\begin{scope}[xshift = 0in]


\coordinate (A1) at (180:2);
\coordinate (A2) at (0:2);
\coordinate (A3) at (0:6);

\coordinate[shift = {(180:2)}] (a1) at (0:15pt);
\coordinate[shift = {(0:2)}] (a2) at (0:15pt);
\coordinate[shift = {(0:6)}]   (a3) at (0:15pt);

\draw (A1) circle (15pt);
\draw (A2) circle (15pt);
\draw (A3) circle (15pt);

\draw[ultra thick, white] (a1) -- (A2);
\draw[thick] (a1) -- (A2);
\draw[ultra thick, white] (a2) -- (A3);
\draw[thick] (a2) -- (A3);


\fill (A1) circle (2pt);
\fill (A2) circle (2pt);
\fill (A3) circle (2pt);
\fill (a1) circle (2pt);
\fill (a2) circle (2pt);
\fill (a3) circle (2pt);

\node[font=\small, above] at (A1){$A_1$};
\node[font=\small, above] at (A2) {$A_2$};
\node[font=\small, above] at (A3) {$A_3$};

\end{scope}
\end{tikzpicture}
\end{center}

\noindent Actors $A_1$, $A_2$, and $A_3$ have actor indices $a_1$, $a_2$, and $a_3$. Actor $A_2$ has two surjective submersions onto $\mathds R^2$, and so the diagram for this system is not simple.  A welding $I_{A_1A_2}$ of actors $A_1$ and $A_2$ guarantees the existence of a configuration space for the total system.  Take $\pi_{a_1}$ and $\pi_{a_2}$ to be the morphisms for a cospan that describes a revolute joint formed from $A_1$ and $A_2$, with the length of the bar connecting $A_1$ and $A_2$ to be $L_1$. Use Example~\ref{Ex:RevoluteJoint} to see that the welded actor $I_{12}$ is isomorphic to $\SE(2)\times \mathds S^1$, and may be identified with the subset
\[
I_{12} = \big\{(x_2\theta_2 - L_1\theta_1 + y_2\theta_2^\perp, \theta_1,  x_2 \theta_2 + y_2 \theta_2^\perp, \theta_2)\mid (x_2,y_2)\in \mathds R^2,\; (\theta_1, \theta_2) \in \SS^1\times\SS^1\big\}
\] 
 of $\SE(2)\times\SE(2)$.
Take $\pi_{12,c}$ and $\pi_{3,c}$ to be the surjective submersions and $L_2$ to be the length of the bar connecting $A_2$ to $A_3$, so that
\[\pi_{12,c}\colon (x_2\theta_2 - L_1\theta_1 + y_2\theta_2^\perp, \theta_1,  x_2 \theta_2 + y_2 \theta_2^\perp, \theta_2) \mapsto (x_2-L_2)\theta_2 + y_2\theta_2^\perp\]
and
\[
\pi_{3,c}\colon (x_3\theta_3+y_3\theta_3^\perp, \theta_3) \mapsto x_3\theta_3+y_3\theta_3^\perp.
\]
An $\F$-pullback of $I_{12}$ and $A_3$ over their shared $\mathds R^3$ constraint is isomorphic to the fiber product $I_{123}$
\begin{align*}
I_{123} & = \big\{(x_2\theta_2 - L_1\theta_1 + y_2\theta_2^\perp, \theta_1,  x_2 \theta_2 + y_2 \theta_2^\perp, \theta_2, x_3\theta_3 + y_3\theta_3^\perp, \theta_3) \\ & \qquad \mid (x_2,y_2,x_3,y_3)\in \mathds R^4,\; (\theta_1, \theta_2, \theta_3) \in (\SS^1)^3, \; (x_2-L_2)\theta_2 + y_2\theta_2^\perp = x_3\theta_3 + y_3\theta_3^\perp\big\}.
\end{align*}
In later examples, if it is useful to treat these two linked revolute joints as a welded actor in order to have actor $A_3$ interact with an additional actor, it will be useful to rewrite the configuration space for $I_{123}$ as 
\begin{align*}
I_{123} & = \big\{(x_3\theta_3 - L_1\theta_1 - L_2\theta_2 + y_3\theta_3^\perp, \theta_1,  x_3 \theta_3 - L_2\theta_2 + y_3 \theta_3^\perp, \theta_2, x_3\theta_3 + y_3\theta_3^\perp, \theta_3) \\ & \qquad \mid (x_3,y_3)\in \mathds R^2,\; (\theta_1, \theta_2, \theta_3) \in (\SS^1)^3\big\}.
\end{align*}
\end{example}

\begin{example}\label{Ex:SliderConstruct} 
A \emph{slider} consists of two actors $A$ and $B$ with aligned angular orientation that can move freely along that orientation, like this:

\medskip

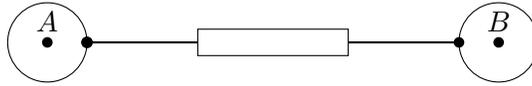
\begin{center} 
\begin{tikzpicture}[scale=1]
\begin{scope}[xshift=0in]

\coordinate (AL)  at (180:2);
\coordinate (AR)  at (0:2); 
\coordinate[yshift = 0pt] (BL)  at (0:0);
\coordinate[yshift = 0pt] (BR)  at (0:4); 
\coordinate[shift = {(180:2)}] (a) at (0:15pt);
\coordinate[yshift = 0pt, shift = {(0:4)}] (b) at (180:15pt);

\draw (AL) circle (15pt);
\draw (BR) circle (15pt);

\draw[thick] (a) -- (AR);
\draw[thick] (b) -- (BL);

\draw[fill = white] (0,-5pt) rectangle (2,5pt);

\draw[fill = black] (a) circle (2pt);
\fill (a) circle (2pt);
\fill (b) circle (2pt);
\fill (AL) circle (2pt);
\fill (BR) circle (2pt);
\end{scope}

\node[font=\small, above] at (AL){$A$};
\node[font=\small, above] at (BR) {$B$};

\end{tikzpicture}
\end{center}

\medskip

\noindent It is helpful to simplify the picture by sketching this as a welded slider, meaning that the welded slider is to be considered as a single actor, like this:

\medskip

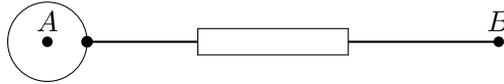
\begin{center} 
\begin{tikzpicture}[scale=1]
\begin{scope}[xshift=0in]

\coordinate (AL)  at (180:2);
\coordinate (AR)  at (0:2); 
\coordinate[yshift = 0pt] (BL)  at (0:0);
\coordinate[yshift = 0pt] (BR)  at (0:4); 
\coordinate[shift = {(180:2)}] (a) at (0:15pt);
\coordinate[yshift = 0pt, shift = {(0:4)}] (b) at (180:15pt);

\draw (AL) circle (15pt);

\draw[thick] (a) -- (AR);
\draw[thick] (BL) -- (BR);

\draw[fill = white] (0,-5pt) rectangle (2,5pt);

\draw[fill = black] (a) circle (2pt);
\fill (a) circle (2pt);
\fill (AL) circle (2pt);
\fill (BR) circle (2pt);
\end{scope}

\node[font=\small, above] at (AL){$A$};
\node[font=\small, above] at (BR) {$B$};

\end{tikzpicture}
\end{center}

Take $a$ and $b$ to be the actor indices for $A$ and $B$, respectively, and $c$ to be the constraint index for the $\mathds R\times \SS^1$ constraint between the two actors.  Define surjective submersions $\pi_{a,c}$ and $\pi_{b,c}$ by 
\[
\pi_{a,c}\colon (x_1 \theta_1 + y_1 \theta_1^\perp, \theta_1) \mapsto (y_1\theta_1^\perp, \theta_1)\quad \text{and} \quad \pi_{b,c}\colon (x_2 \theta_2 + y_2 \theta_2^\perp, \theta_2) \mapsto (y_2\theta_2^\perp, -\theta_2).
\] 
An $\F$-pullback $I_{AB}$ of the cospan $\rangle \pi_{ac}, \pi_{bc}\langle$ is the fibered product \[I_{AB} = \{(x_1 \theta_1 + y_1 \theta_1^\perp, \theta_1,  x_2 \theta_1 + y_1 \theta_1^\perp, -\theta_1) \in A \times B \mid (x_1, y_1) \in \mathds R^2, \; x_2 \in \mathds R, \; \theta_1 \in \SS^1\}.\]
This fibered product is diffeomorphic to $\SE(2) \times \R$.
\end{example}

\begin{example}\label{Ex:Slidinghinge2DConstruct}
A sliding hinge should be a system that looks like this:

\medskip

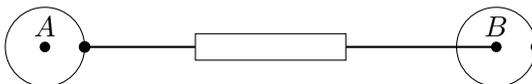
\begin{center} 
\begin{tikzpicture}[scale=1]
\begin{scope}[xshift=0in]

\coordinate (AL)  at (180:2);
\coordinate (AR)  at (0:2); 
\coordinate[yshift = 0pt] (BL)  at (0:0);
\coordinate[yshift = 0pt] (BR)  at (0:4); 
\coordinate[shift = {(180:2)}] (a) at (0:15pt);
\coordinate[yshift = 0pt, shift = {(0:4)}] (b) at (0:15pt);

\draw (AL) circle (15pt);
\draw (BR) circle (15pt);

\draw[thick] (a) -- (AR);
\draw[thick] (BL) -- (BR);

\draw[fill = white] (0,-5pt) rectangle (2,5pt);

\draw[fill = black] (a) circle (2pt);
\fill (a) circle (2pt);
\fill (b) circle (2pt);
\fill (AL) circle (2pt);
\fill (BR) circle (2pt);
\end{scope}

\node[font=\small, above] at (AL){$A$};
\node[font=\small, above] at (BR) {$B$};

\end{tikzpicture}
\end{center}

\noindent The configuration space should be $\SE(2)\times \mathds R \times \mathds S^1$. It appears that this linkage may be constructed using two actors, as sketched, but Theorem~\ref{thm:NoTwoActorPlanarSlidingHinge} shows that this is not possible in two dimensions.  To construct this linkage in two dimensions, begin with the slider, viewed as the welded actor with configuration space $I_{12}$.  View a third actor $A_3$ connecting to $A_2$ like this:

\medskip

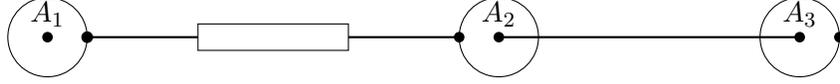
\begin{center} 
\begin{tikzpicture}[scale=1]
\begin{scope}[xshift=0in]

\coordinate (AL)  at (180:2);
\coordinate (AR)  at (0:2); 
\coordinate[yshift = 0pt] (BL)  at (0:0);
\coordinate[yshift = 0pt] (BR)  at (0:4); 
\coordinate[shift = {(180:2)}] (a) at (0:15pt);
\coordinate[yshift = 0pt, shift = {(0:4)}] (b) at (180:15pt);
\coordinate[yshift = 0pt] (C)  at (0:8); 
\coordinate[yshift = 0pt, shift = {(0:8)}] (c)  at (0:15pt);

\draw (AL) circle (15pt);
\draw (BR) circle (15pt);
\draw (C) circle (15pt);

\draw[thick] (a) -- (AR);
\draw[thick] (b) -- (BL);
\draw[thick] (BR) -- (C);

\draw[fill = white] (0,-5pt) rectangle (2,5pt);

\draw[fill = black] (a) circle (2pt);
\fill (a) circle (2pt);
\fill (b) circle (2pt);
\fill (AL) circle (2pt);
\fill (BR) circle (2pt);
\fill (C) circle (2pt);
\fill (c) circle (2pt);
\end{scope}

\node[font=\small, above] at (AL){$A_1$};
\node[font=\small, above] at (BR) {$A_2$};
\node[font=\small, above] at (C) {$A_3$};
\end{tikzpicture}
\end{center}

Take $c$ to be the actor index for the $\mathds R^2$ constraint between $I_{12}$ and $A_3$.  Denote by $\pi_{12,c}$ and $\pi_{3,c}$ the constraint morphisms
\[
\pi_{12,c} \colon (x_1 \theta_1 + y_1 \theta_1^\perp, \theta_1,  x_2 \theta_1 + y_1 \theta_1^\perp, -\theta_1) \to x_2 \theta_1 + y_1 \theta_1^\perp
\]
and
\[\pi_{3,c} \colon (x_3\theta_3 + y_3\theta_3^\perp, \theta_3) \mapsto x_3\theta_3 + y_3\theta_3^\perp.\]

The fibered product $X$ is an $\F$-pullback of the cospan $\rangle \pi_{12,c}, \pi_{3,c}\langle$, where 
\begin{align*}
X &= \{(x_1 \theta_1 + y_1 \theta_1^\perp, \theta_1,  x_2 \theta_1 + y_1 \theta_1^\perp, -\theta_1, x_2 \theta_1 + y_1 \theta_1^\perp, \theta_3) \\& \hspace{2in}\mid (x_1, y_1) \in \mathds R^2, \; x_2 \in \mathds R, \; (\theta_1, \theta_3) \in (\SS^1)^2\},
\end{align*}
which is isomorphic to $\SE(2)\times \mathds R \times \mathds S^1$.

\end{example}

\begin{theorem}[No Two--Actor Realization of the Planar Sliding Hinge]
\label{thm:NoTwoActorPlanarSlidingHinge}
Take $\SE(2)$ acting on itself by left multiplication and $M$ to be a smooth $\SE(2)$--manifold with equivariant surjective submersions
\[
p_1,p_2\colon \SE(2)\to M.
\]
Denote by $P$ the $\F$-pullback
\[
P = \SE(2)\times_M \SE(2)
\]
with the diagonal $\SE(2)$--action induced by that on the $\F$-product $\SE(2)\times\SE(2)$. There are no choices of $M$, $p_1$, and $p_2$ so that $P$ is $\SE(2)$--equivariantly diffeomorphic to $\SE(2)\times \R \times \SS^1$.
\end{theorem}

\begin{proof}
Since $p_1$ and $p_2$ are surjective submersions, the fiber product $P$ is a smooth manifold and the dimension formula for submersions gives
\[
\dim(P)=\dim(\SE(2))+\dim(\SE(2))-\dim(X)=6-\dim(X).
\]
If $P$ is $\SE(2)$--equivariantly diffeomorphic to $\SE(2)\times \R\times \SS^1$ with $\SE(2)$ acting on the first factor, as Lemma~\ref{lem:FiberProductNormalForm} would require, then %
\begin{equation}\label{Sec5:hinge:eqA}
\dim(M) = 1.
\end{equation}

Surjectivity and $\SE(2)$--equivariance of $p_1$ imply that the $\SE(2)$--action on $M$ is transitive.  Fix $x_0$ in $M$, take $H$ to be the stabilizer of $x_0$, and identify $M$ to be isomorphic to the homogeneous space $\SE(2)/H$ \cite{HelgasonDG}. Equation \eqref{Sec5:hinge:eqA} implies that
\begin{equation}\label{Sec5:hinge:eqB}
\dim(H)=\dim(\SE(2))-\dim(M)=3-1=2.
\end{equation}

Compute the pullback as in the universal--joint argument to obtain an $\SE(2)$--equivariant diffeomorphism
\[
\SE(2)\times_{\SE(2)/H}\SE(2)\;\cong\;\SE(2)\times H, \quad \text{by} \quad (g,h)\mapsto (g,gh).
\]  %
The pullback source $P$ is therefore $\SE(2)$--equivariantly diffeomorphic to $\SE(2)\times H$.

Choose a basis $\{A,X,Y\}$ for the Lie algebra $\mathfrak{se}(2)$ so that $A$ generates rotations and $X,Y$ generate translations, hence
\[
[A,X]=Y,\qquad [A,Y]=-X,\qquad [X,Y]=0.
\]
Denote by $H^\circ$ the connected component of $H$ that contains the identity and set \[\mathfrak h=\Lie(H^\circ).\]  Since $\dim(H^\circ)$ is equal to $\dim(H)$, equation \eqref{Sec5:hinge:eqB} implies that $\dim(\mathfrak h)$ is equal to $2$.

If $A$ is in $\mathfrak h$, then $\mathfrak h$ is the linear span $\LSpan\{A,v\}$ for some nonzero $v$ in $\LSpan\{X,Y\}$.  There are real numbers $\alpha$ and $\beta$ so that \[v=\alpha X+\beta Y.\] Linearity of the Lie bracket gives the equality
\[
[A,v]=\alpha[A,X]+\beta[A,Y]=\alpha Y-\beta X.
\]
The vector $\alpha Y-\beta X$ is linearly independent from $v$ unless $v$ is the zero vector. Thus $[A,v]$ is not in $\LSpan\{A,v\}$ for every nonzero $v$, which contradicts that $\mathfrak h$ is a Lie subalgebra. Therefore $A$ is not in $\mathfrak h$.  Furthermore, if there is a $w$ in $\LSpan\{X,Y\}$ and real numbers $a$ and $b$ so that $aA + bw$ is in $\mathfrak h$, then then the closure of $\mathfrak h$ under the bracket implies that $A$ is also in $\mathfrak h$, or $a$ is zero.  The Lie subalgebra $\mathfrak h$ is, therefore, a subset of $\LSpan\{X,Y\}$.  Since $\LSpan\{X,Y\}$ is abelian and $\dim(\mathfrak h)$ is equal to $2$, \[\mathfrak h = \LSpan\{X,Y\},\] and so $H^\circ$ is the translation subgroup $\R^2$ of $\SE(2)$.

The identity component of $H$ is diffeomorphic to $\R^2$, so $H$ is not diffeomorphic to $\R\times \SS^1$. Therefore $\SE(2)\times H$ is not diffeomorphic to $\SE(2)\times \R\times \SS^1$. This contradicts the assumption that $P$ is $\SE(2)$--equivariantly diffeomorphic to $\SE(2)\times \R\times \SS^1$.  Hence no such $M$, $p_1$, and $p_2$ exist.
\end{proof}

\paragraph{{\bf Sliding hinge motion set.}}

Fix two orthonormal unit vectors $u_{\mathrm h}$ and $u_{\mathrm s}$ in $\SS^2$. For each $\theta$ in $\SS^1$, identify $\theta$ with an angle measure and write $R_{u_{\mathrm h}}(\theta)$ for the rotation in $\SO(3)$ about the axis $u_{\mathrm h}$ through angle $\theta$. Identify \[G=\SE(3)=\R^3\rtimes\SO(3).\] Define the subset $S$ of $G$, the \emph{sliding hinge motion set}, by
\[
S:=\big\{(t u_{\mathrm s},R_{u_{\mathrm h}}(\theta)) \mid t\in\R,\; \theta\in\SS^1\big\}.
\]

\begin{lemma}
\label{lem:SlidingHingeNotSubgroup}
The sliding hinge motion set $S$ is not a subgroup of $G$.
\end{lemma}

\begin{proof}
Fix $\theta$ in $\SS^1$ to be neither the identity for $\SS^1$ nor half of the circle. Since $u_{\mathrm s}$ and $u_{\mathrm h}$ are perpendicular, the vector $R_{u_{\mathrm h}}(\theta)u_{\mathrm s}$ is not a scalar multiple of $u_{\mathrm s}$. The semidirect product law implies that
\[
(tu_{\mathrm s},R_{u_{\mathrm h}}(\theta))(t'u_{\mathrm s},I) = (tu_{\mathrm s}+t'R_{u_{\mathrm h}}(\theta)u_{\mathrm s}, R_{u_{\mathrm h}}(\theta)).
\]
For any nonzero $t'$, there is no real number $C$ with \[tu_{\mathrm s}+t'R_{u_{\mathrm h}}(\theta)u_{\mathrm s} = Cu_{\mathrm s},\] so the sum is not a scalar multiple of $u_{\mathrm s}$, hence the product is not in $S$.  Thus $S$ is not closed under multiplication.
\end{proof}

\begin{lemma}
\label{lem:ClosureForcesPlane}
For any subgroup $H$ of $G$, if for all $\theta$ in $\SS^1$ the pair $(0,R_{u_{\mathrm h}}(\theta))$ is in $H$ and for some nonzero $t$ the pair $(tu_{\mathrm s},I)$ is in $H$, then the pair $(v,I)$ is in $H$ for all $v$ in $u_{\mathrm h}^{\perp}$.
\end{lemma}

\begin{proof}
Fix $t$ to be nonzero with $(tu_{\mathrm s},I)$ in $H$. Closure of $H$ under conjugation (by elements of $H$) implies that
\[
(0,R_{u_{\mathrm h}}(\theta))(tu_{\mathrm s},I)(0,R_{u_{\mathrm h}}(\theta))^{-1}\in H.
\]
The equality \[(0,R)^{-1} = (0,R^{-1})\] implies that
\[
(0,R)(tu_{\mathrm s},I)(0,R^{-1}) = (tRu_{\mathrm s},I),
\]
hence for all $\theta$, $(tR_{u_{\mathrm h}}(\theta)u_{\mathrm s},I)$ is in $H$.

As $\theta$ varies, $R_{u_{\mathrm h}}(\theta)u_{\mathrm s}$ traces the unit circle in the plane $u_{\mathrm h}^{\perp}$. Therefore $H$ contains translations along every direction in $u_{\mathrm h}^{\perp}$. Closure under addition of translations and the equality
\[
(v,I)(w,I)=(v+w,I),
\]
together imply that $(v,I)$ is in $H$ for every $v$ in $u_{\mathrm h}^{\perp}$.
\end{proof}

\begin{theorem}[No Two--Actor Realization of the Spatial Sliding Hinge]
\label{thm:NoTwoActorFiberProductSpatialSlidingHinge}
There is no smooth $G$--manifold $M$ with smooth $G$--equivariant surjective submersions $p_1$ and $p_2$ from $G$ to $M$ such that the fiber product
\[
P:=G\times_M G
\]
has relative motion set equal to $S$.

\end{theorem}

\begin{proof}
Take any smooth $G$--manifold $M$ and smooth $G$--equivariant surjective submersions $p_1$ and $p_2$ from $G$ to $M$. Lemma~\ref{lem:FiberProductNormalForm} implies that there exists a subgroup $H$ of $G$ such that $P$ is $G$--equivariantly diffeomorphic to $G\times H$. Under this identification, the relative motion set
\[
\{g_2^{-1}g_1g_0 \mid (g_1,g_2)\in P\}
\]
is exactly $H$.

If $P$ were to model a spatial sliding hinge with rotation axis $u_{\mathrm h}$ and sliding direction $u_{\mathrm s}$, which is perpendicular to $u_{\mathrm h}$, then $H$ would have to contain all rotations $(0,R_{u_{\mathrm h}}(\theta))$ and at least one nontrivial translation $(tu_{\mathrm s},I)$. Lemma~\ref{lem:ClosureForcesPlane} then forces $(v,I)$ to lie in $H$ for every $v$ in $u_{\mathrm h}^{\perp}$. Thus any such subgroup $H$ necessarily contains all planar translations orthogonal to $u_{\mathrm h}$, rather than translations along a single sliding direction. Consequently, the relative motion set cannot equal the sliding hinge motion set $S$.

\end{proof}

\begin{example}\label{Ex:CylindricalJointConstruct}
Although Theorem~\ref{thm:NoTwoActorPlanarSlidingHinge} shows that the sliding hinge is not a kinematic pair in two dimensions, and Theorem~\ref{thm:NoTwoActorFiberProductSpatialSlidingHinge} extends this to the spatial case, the present example shows that the cylindrical joint can be constructed as a spatial kinematic pair. The configuration space for this joint is $\SE(3)\times \mathds R\times \SS^1$, which  gives the proper degrees of freedom for the sliding hinge, but the geometry of the system is different.  To this end, recall that every actor of a spatial linkage is an element of
\[
\SE(3)=\R^3\rtimes\SO(3), \quad \text{where} \quad (a,R)(a',R')=(a+Ra', RR').
\]
The $\R^3$ component records the position of the actor in space. Fix a unit vector $u$ in $\SS^2$ to designate the oriented hinge axis of the actor. The unit sphere $\SS^2$ parameterizes the possible oriented hinge-axis directions in space.

The homogeneous-space description of $\SS^2$ as a quotient of $\SO(3)$ plays a central role in the construction of the three-dimensional sliding hinge.
Denote by $\SS^1$ the stabilizer of $u$ under the left action of $\SO(3)$ on $\SS^2$, that is, 
\[
\SS^1 := \Stab_{\SO(3)}(u) = \{R\in\SO(3) \mid Ru = u\}.
\]
This subgroup is isomorphic to $\SO(2)$. Define the surjective submersion
\[
\pi\colon\SO(3)\to \SS^2 \quad \text{by} \quad \pi(R):=Ru.
\]
The quotient $\SO(3)/\SS^1$ is diffeomorphic to $\SS^2$, and $\pi$ exhibits $\SO(3)$ as a non-trivial principal $\SS^1$--bundle over $\SS^2$.

The sliding hinge constraint must fix not only an axis direction but also an axis line in space. Define the smooth manifold
\[
X := \{(x,v)\in\R^3\times \SS^2 \mid x\cdot v=0\},
\]
which parameterizes oriented lines in $\R^3$ by associating to each $(x,v)$ the oriented line
\[
\ell(x,v)=\{x+t v\mid t\in\R\},
\]
where $v$ is the unit direction and $x$ is the unique point on the line closest to the origin.  Define the surjective submersion
\[
p\colon\SE(3)\to X \quad \text{by} \quad p(a,R):=(x,v),
\]
where
\[
v:=\pi(R)=Ru \quad \text{and} \quad x:=a-(a\cdot v)v.
\]
The equality \[x\cdot v=0\] ensures that $p(a,R)$ is in $X$.  The geometric significance is that $v$ is the spatial direction of the hinge axis of the actor in the state $(a,R)$, and $x$ is the perpendicular component of the actor position $a$, which
determines the unique axis line through $a$ in direction $v$.

The group $\SE(3)$ acts smoothly on $X$ by transporting oriented lines:
\begin{align*}
(a_0,R_0)\cdot(x,v) := (x',v') \quad \text{by} \quad \begin{cases} v':=R_0 v & \text{}\\x':=\proj_{(R_0 v)^\perp}(R_0 x + a_0), & \text{}\end{cases}
\end{align*}
where 
\[
\proj_{w^\perp}(y)=y-(y\cdot w)w.
\]
The equality
\[
p\big((a_0,R_0)(a,R)\big)=(a_0,R_0)\cdot p(a,R),
\]
shows that $p$ is $\SE(3)$--equivariant.

Fix the reference element $e$ in $\SE(3)$ to be $(0,I)$ and write
\[
(x_0,v_0) := p(e) = (0,u) \in X,
\]
so the reference oriented line is
\[
\ell_0=\{t v_0\mid t\in\R\}.
\]
The subgroup $H$ of $\SE(3)$ that is given by
\[
H := \Stab_{\SE(3)}(x_0,v_0) = \{g\in\SE(3)\mid g\cdot(x_0,v_0)=(x_0,v_0)\},
\]
consists exactly of translations along $\ell_0$ and rotations about $\ell_0$.  It is, therefore, given by
\[
H=\{(t v_0, R_{v_0}(\theta))\mid  t\in\R,\ \theta\in\R/2\pi\ZZ\},
\]
hence isomorphic to $\R\times \SS^1$.

Since $p$ is equivariant and \[p(e)=(x_0,v_0),\] for all $g$ in $\SE(3)$, 
\[
p(g)=g\cdot(x_0,v_0).
\]
In particular, the $\SE(3)$--action on $X$ is transitive, and $p$ identifies $X$ with the homogeneous space $\SE(3)/H$. Hence $p$ is a surjective submersion.

Take $A_1$ and $A_2$ to be actors, each equipped with the left $\SE(3)$--action by left multiplication, and define for each $i$ in $\{1,2\}$ the surjective submersion $p_i$ by
\[
p_i:=p\colon A_i\to X .
\]
Define the constrained configuration space as the pullback
\[
P := A_1\times_X A_2 = \{(g_1,g_2)\in\SE(3)\times\SE(3)\mid p(g_1)=p(g_2)\},
\]
so that $(g_1,g_2)$ is in $P$ exactly when the two actors determine the same oriented axis line in space.

For any $(g_1, g_2)$ in $P$, define $\Phi(g_1, g_2)$ by 
\[
\Phi(g_1,g_2):=(g_1,\,g_1^{-1}g_2).
\]
The equality \[p(g_1)=p(g_2)\] implies that
\[
g_1\cdot(x_0,v_0)=g_2\cdot(x_0,v_0),
\]
hence \[(g_1^{-1}g_2)\cdot(x_0,v_0)=(x_0,v_0)\] and $g_1^{-1}g_2$ is in $H$. Therefore $\Phi$ is a well-defined function from $P$ to $\SE(3)\times H$.
Its inverse is
\[
\Phi^{-1}(g,h)=(g,gh),
\]
hence $\Phi$ is an $\SE(3)$--equivariant diffeomorphism. Consequently,
\[
P \cong \SE(3)\times H \cong \SE(3)\times\R\times \SS^1,
\]
with $\SE(3)$ acting by left multiplication on the first factor.

Assume the center of each actor is the origin of its own frame and write
\[
g_i=(a_i, R_i)\in\SE(3),
\]
so that 
\[
g_1^{-1}g_2=\bigl(R_1^{-1}(a_2-a_1),R_1^{-1}R_2\bigr).
\]
For any $(g_1,g_2)$ in $P$, there exist unique $\theta$ in $\SS^1$ and $t$ in $\R$ so that
\[
R_1^{-1}R_2=R_{u}(\theta) \quad \text{and} \quad R_1^{-1}(a_2-a_1)=t u.
\]
Thus $\theta$ is the relative rotation about the common axis and $t$ is the signed displacement of the second body origin along that axis, expressed in the frame of the first actor.

Note that a cylindrical joint without the sliding component is the construction of a circular bearing, or hinge, in space.
\end{example}

\begin{example}\label{example:lock3Bar}
A three bar linkage that forms a closed loop (a truss) is an example of a linkage that does not decompose external constraints, and that has a cyclic skeleton. This linkage reflects the challenge of modeling any such system in the ACM framework.  Here is a proposed diagram for such a system:

\medskip

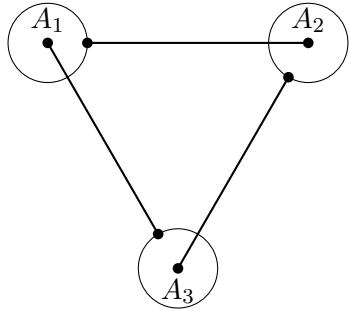
\begin{center}
\begin{tikzpicture}[scale=1]
\begin{scope}


\coordinate (A)  at (150:2);
\coordinate (B) at (30:2);
\coordinate (C) at (270:2);
\coordinate[shift = {(150:2)}] (a) at (0:15pt);
\coordinate[shift = {(30:2)}] (b) at (240:15pt);
\coordinate[shift = {(270:2)}] (c) at (120:15pt);

\draw (A) circle (15pt);
\draw (B) circle (15pt);
\draw (C) circle (15pt);

\fill (A) circle (2pt);
\fill (B) circle (2pt);
\fill (C) circle (2pt);

\draw[thick] (a) -- (B);
\draw[thick] (b) -- (C);
\draw[thick] (c) -- (A);

\fill (a) circle (2pt);
\fill (b) circle (2pt);
\fill (c) circle (2pt);

\node[anchor = south, font = \small] at (A) {$A_1$}; 
\node[anchor = south, font = \small] at (B) {$A_2$}; 
\node[anchor = north, font = \small] at (C) {$A_3$}; 

\end{scope}
\end{tikzpicture}
\end{center}

Example~\ref{Ex:RevoluteJoint2} already gives the configuration space for the system without the interaction between $A_1$ and $A_3$. The interaction between actor $A_1$ and $A_3$, with their separation required to be $L_3$, introduces two new surjective submersions $\pi_{a_1c}$ and $\pi_{a_3c}$ that are given by 
\[\pi_{a_1c} \colon (x_1\theta_1 + y_1\theta_1^\perp, \theta_1) \mapsto x_1\theta_1 + y_1\theta_1^\perp\]
and
\[\pi_{a_3c} \colon (x_3\theta_3 + y_3\theta_y^\perp, \theta_3) \mapsto (x_3-L_3)\theta_3 + y_3\theta_3^\perp.\]
The total space for the linkage would be the set of six tuples of the form
\[
(x_3\theta_3 - L_1\theta_1 - L_2\theta_2 + y_3\theta_3^\perp, \theta_1,  x_3 \theta_3 - L_2\theta_2 + y_3 \theta_3^\perp, \theta_2, x_3\theta_3 + y_3\theta_3^\perp, \theta_3)
\]
with an additional condition imposed by the equality
\[
x_3\theta_3 - L_1\theta_1 - L_2\theta_2 + y_3\theta_3^\perp = (x_3-L_3)\theta_3 + y_3\theta_3^\perp.
\]
Which is, of course, equivalent to the equality \[L_1\theta_1 + L_2\theta_2 + L_3\theta_3 = \Vec{0},\] where $\Vec{0}$ is the zero vector in $\mathds R^2$.  If this sum is not the zero vector, then the system cannot be constructed. If the sum is the zero vector, then the configuration space $X$ of the system is simply $\SE(2)$, and so there is no surjective submersion from $X$ to the configuration space for any interaction between actors. This system is the \emph{locked three bar linkage} because the only rotational degree of freedom for the system is the rotation that comes from the diagonal action of $\SE(2)$ on the total configuration space.

The three bar linkage may, however, be constructed from a sliding three bar linkage with the use of a Newton Daemon.  The Newton Daemon fixes two bar lengths and an angle between bars, and can be constructed on realizing the sliding three bar system as an interaction between two welded actors, each with new constraints and constraint morphisms.  This system motivates the idea of \emph{compound actors}, to be developed in a subsequent work that treats modeling systems using the ACM framework. 
\end{example}

 \end{document}